\documentclass[11pt]{article}
\usepackage{cite}
\usepackage{scalerel}
\usepackage{xcolor}
\usepackage{shuffle}
\usepackage{amsmath,amsfonts,amssymb}
\usepackage{latexsym,epsfig}
\usepackage{dsfont}
\usepackage{hyperref}
\usepackage{extarrows}
\usepackage{comment} 
\usepackage{tikz-cd}
\usepackage{amsthm}

\newtheorem{thm}{Theorem}
\newtheorem*{prop}{Proposition}

\theoremstyle{definition}
\newtheorem*{defn}{Definition}

\def\hybrid{
        \topmargin -20pt
        \oddsidemargin 0pt
        \headheight 0pt \headsep 0pt
        \textwidth 6.25in 
        \textheight 9.5in 
        \marginparwidth .875in
        \parskip 5pt plus 1pt \jot = 1.5ex}

\hybrid

 \csname
@addtoreset\endcsname{equation}{section}

\allowdisplaybreaks

\def\bpm{\begin{pmatrix}}
\def\epm{\end{pmatrix}}
\newcommand{\bra}[1]{\langle#1\rvert}
\newcommand{\ket}[1]{\lvert#1\rangle}

\usepackage[all,cmtip]{xy}
\usepackage{mathtools}
\usepackage{mathrsfs}

\usepackage{float}
\usepackage[font={footnotesize}]{caption}
\usepackage{subcaption}

\usepackage{geometry}
\usepackage{comment}
\usepackage[absolute,overlay]{textpos}
\usepackage[section]{placeins}

\newcommand{\expval}[1]{\langle{#1}\rangle}

\DeclareMathOperator{\ima}{im}
\DeclarePairedDelimiterX\braket[2]{\langle}{\rangle}{#1 \delimsize\vert #2}

\usepackage{calrsfs}
\DeclareMathAlphabet{\pazocal}{OMS}{zplm}{m}{n}

\newcommand\restr[2]{{\left.\kern-\nulldelimiterspace #1 \littletaller \right|_{#2}}}
\newcommand{\littletaller}{\mathchoice{\vphantom{\big|}}{}{}{}}

\begin{document}

\begin{titlepage}
\rightline{}
\rightline{December 2024}
\rightline{HU-EP-24/37} 
  
\begin{center}
\vskip 1.5cm
{\Large \bf{Off-Shell Quantum Mechanics \\[1.5ex] as Factorization Algebras on Intervals }}
\vskip 1.7cm

{\large\bf {Christoph Chiaffrino, Noah Hassan and Olaf Hohm}}
\vskip 0.4cm

\vskip .2cm

\end{center}

\bigskip\bigskip
\begin{center} 
\textbf{Abstract}

\end{center} 
\begin{quote}

We present, for the  harmonic oscillator 
and the spin-$\frac{1}{2}$ system, 
 an alternative formulation of quantum mechanics 
 that is `off-shell': it  is based on 
classical off-shell configurations and thus  similar to the path integral. 
The core elements  are  Batalin-Vilkovisky (BV) algebras and 
factorization algebras, following a program by Costello and Gwilliam. 
The BV algebras are the  spaces of quantum observables ${\rm Obs}^q(I)$  given by 
the symmetric algebra of polynomials in compactly 
supported functions on some interval $I\subset \mathbb{R}$, 
which can be viewed as functionals on the dynamical variables. 
Generalizing  associative algebras, 
factorization algebras include in their data  a topological space, which here is $\mathbb{R}$, 
and an assignment of a vector space to each open set, which here is the assignment  
of  ${\rm Obs}^q(I)$ to each open interval $I$. The central structure maps are bilinear 
${\rm Obs}^q(I_1)\otimes {\rm Obs}^q(I_2)\rightarrow {\rm Obs}^q(J)$ for disjoint intervals $I_1$ and $I_2$
contained in an interval $J$, which here is the  wedge product of the symmetric algebra. 
We prove, as the central result of this paper, 
that this factorization algebra is quasi-isomorphic to the factorization algebra of 
`on-shell' quantum mechanics. In this we extend 
previous work by including  half-open and closed intervals, 
and by generalizing to the  spin-$\frac{1}{2}$ system.

\end{quote} 
\vfill
\setcounter{footnote}{0}
\end{titlepage}

\tableofcontents

\vspace{5mm}

\section{Introduction}

In this paper we present an alternative `off-shell' formulation of quantum mechanics for the paradigmatic  cases of the harmonic oscillator and the spin-$\frac{1}{2}$ system. This formulation is based on  Batalin-Vilkovisky (BV) algebras and their  cohomology and the notion of `factorization algebras'. In this we apply and illustrate a program for the axiomatization of 
quantum field theory (QFT) due to Costello and Gwilliam \cite{Costello_Gwilliam_2016,Costello_Gwilliam_2021} in the 
simpler setting of quantum mechanics. 
This paper also generalizes the work of two of us, together with Pinto, 
on the formulation of quantum mechanics in terms of the cohomology of a BV algebra \cite{Chiaffrino:2021pob}. 
These formulations amount to a slight reinterpretation of the BV formalism, which was originally introduced within the path integral quantization of gauge theories, but was more recently understood as an alternative to the path integral that is useful also for theories without gauge redundancies. 
(For related applications of factorization algebras see \cite{BWilliam,GwilliamRejzner,Gwilliam:2022vja,Costello:2022wso,nishinaka2024}.)

In order to explain the new formulation, in particular in which sense it is off-shell, 
let us  recall  the standard formulation of quantum mechanics. Suppose 
we are given a classical theory whose dynamical variable is a smooth function $\phi(t)$
on the real line, i.e., $\phi$ lives in the vector space $C^\infty(\mathbb{R})$. 
While this space is infinite-dimensional, the space of solutions is typically finite-dimensional.  
For instance, if $\phi(t)$ is subject to equations of motion that are of second-order in time derivatives 
a solution is uniquely determined by specifying, for instance,  $\phi(0)=q$, $\dot{\phi}(0)=p$. 
Thus, we can identify a point $(q,p)$  in the phase space  $\mathbb{R}^2$ with a solution 
of the equations of motion. The canonical quantization now promotes $q$ and $p$ to operators 
$\hat{q}$ and $\hat{p}$
on some Hilbert space and postulates the Heisenberg commutation relation $[\hat{q},\hat{p}]=i\hbar$. 
Since it is thus based on the `phase space of classical solutions' one may call  canonical quantization 
`on-shell quantum mechanics'. 

This canonical formulation is to be contrasted with the path integral approach to quantum mechanics. 
This formulation is off-shell in that the path integral expresses the quantum expectation value or overlap of 
states as an integral over the infinite-dimensional space of \textit{all}  classical configurations 
subject to certain boundary conditions. 
Concretely, the 
overlap of states $\ket{q_i,t_i}$ and $\ket{q_f,t_f}$, where the notation indicates position eigenstates 
of the Heisenberg-picture operator $\hat{q}$ at time $t$: $\hat{q}(t)\ket{q;t}=q\ket{q;t}$, 
is written as the integral 
 \begin{equation}
  \langle q_f,t_f \ket{q_i,t_i} = \int_{\phi(t_i)=q_i,\, \phi(t_f) =q_f} {\cal D}\phi\, \exp\Big(\frac{i}{\hbar} S\big[\phi(t)\big]\Big)  \;, 
 \end{equation} 
where $S$ denotes the classical action, 
and the notation indicates 
that the integral is to be taken over all fields $\phi$ 
obeying the 
boundary conditions $\phi(t_i)=q_i,\, \phi(t_f) =q_f$. The path integral  is of crucial importance 
in  QFT, where the proper definition of the objects  of canonical quantization 
(Hilbert space, field operators, Hamiltonian, etc.) is very involved
if not impossible. In particular in non-abelian gauge theories such as Yang-Mills theory the use of the path integral is 
indispensable. The trouble is that in QFT there does not appear to be a mathematically rigorous definition of the measure 
${\cal D}\phi$ of the path integral.

The off-shell formulation of quantum mechanics to be discussed in this paper is analogous to the 
path integral in that there is no Hilbert space nor operators satisfying canonical commutation relations; 
rather, it is based on  off-shell classical field configurations, and the usual operator algebra of quantum mechanics 
only emerges upon `going on-shell'. 
The first crucial step is the following change of 
attitude: instead of talking about the space of all kinematically allowed 
classical field configurations and the `Euler-Lagrange subspace' of solutions to the  equations 
of motion, one talks about a chain complex and its cohomology. Let us illustrate this for the harmonic oscillator 
with  dynamical variable $\phi(t)$ in  $C^\infty(\mathbb{R})$ 
and equations of motion $\ddot{\phi} + \omega^2\phi = 0$. These data are encoded 
in the following chain complex:  
 \begin{equation}\label{oringalcomplex} 
 0\xrightarrow[]{d\ = \ 0}
 C^\infty(\mathbb{R}) \xrightarrow[]{d \ = \ \partial_t^2 + \omega^2} C^\infty(\mathbb{R})
\xrightarrow[]{d \ = \ 0} 0 \,. 
 \end{equation} 
Specifically, this is  a sequence of vector spaces, consisting of two copies of $C^\infty(\mathbb{R})$ and 
the trivial space indicated by $0$, where $\phi(t)$ lives in the space on the left  that we take to be 
the space in degree zero, which in turn is mapped via the differential $d  = \partial_t^2 + \omega^2$ 
to the space on the right, which we take to be of degree one. Note that $d^2=0$  holds trivially since one of 
the $d$ always vanishes. We can now consider the cohomology, which is the quotient vector space 
$H={\rm ker} \,d/{\rm im} \,d$, which in degree zero is just $\ker(\partial_t^2 + \omega^2)$, 
the space of solutions of the equations of motion.

In a second step one passes over to the space of observables, which one models by certain functionals 
of the dynamical variables as follows.\footnote{This step is well motivated 
in \textit{derived} algebraic geometry. 
Roughly speaking, the difference between ordinary geometry and algebraic geometry 
is to work, instead of a subspace $S \subseteq M$ specified by a family of equations $f^i = 0$, with the algebraic space of functions on $M$ modulo those proportional to the 
functions $f^i$ that define $S$, i.e., in our example, the full space modulo the equations of motion. In the {derived} approach one then realizes this 
quotient space as the cohomology of a chain complex, which is defined in terms of a so-called Tor group, and works on this chain complex instead \cite{gwilliam2012factorization,Costello_Gwilliam_2016}.}
We fix an interval $[t_i,t_f]$ between some initial and final time 
once and for all. We then consider compactly supported functions, i.e.,  functions $f$ 
that are `localized' by being  
non-vanishing  only on an open interval $I\subset [t_i,t_f]$. Such a function $f$ gives rise to 
a linear \textit{functional} $F$ acting as  
\begin{equation}\label{FactionIntro} 
    F[\phi] := \int_Idt\, f(t)\phi(t) \equiv\langle f, \phi\rangle \;.  
\end{equation}
We can think of such functionals  as measuring devices or machines that record 
certain properties of the dynamical variable  $\phi$. 
Apart from making everything that follows mathematically rigorous, 
the assumption of compact support also seems physically 
reasonable, for it means that the time period  within which the measurement is made is  finite.  
For instance, one can find functions $f_Q$ and $f_P$ with arbitrarily small compact support that give rise to functionals $Q$ and $P$ that measure position and momentum at time zero, i.e., 
$Q[\phi] = \langle f_Q, \phi\rangle=\phi(0)$ and $P[\phi] = \langle f_P, \phi\rangle = \dot{\phi}(0)$, 
when $\phi$ is on-shell.  
This last assumption allows one to avoid the use of  the delta function and its derivative, 
since one can find infinitely many compactly supported functions giving rise to functionals having the above characteristics \textit{on 
solutions}, yet they are   well-defined on the full space.\footnote{Interestingly,  $f_Q$ may even vanish at $t=0$, much as a machine may operate at a later or earlier time in order 
to determine   the position at time zero.}

Next we can define certain non-linear observables, such as the square of the position or momentum at time zero, 
by passing  over to the symmetric algebra, which consists of strings, or rather polynomials, 
of the above compactly supported functions $f$ 
(with support inside the interval $I$), 
 \begin{equation}
   F =  f_1 \wedge f_2 \wedge \ldots \wedge f_n\;, 
\end{equation}
where the (graded symmetric) wedge product is just the formal concatenation of symbols. 
Such functionals act on the dynamical variables as 
 \begin{equation}
  F[\phi] = \langle f_1,\phi\rangle \cdot \ldots \cdot \langle f_n, \phi\rangle \;, 
 \end{equation}
using the notation in (\ref{FactionIntro}). The differential of the original complex (\ref{oringalcomplex}) 
`dualizes' to a differential $d$  on the complex of linear observables and is then extended via the Leibniz rule 
to the full symmetric algebra, which in turn  renders the symmetric algebra a differential graded  algebra, 
which we denote by ${\rm Sym} \big({C^\infty_c(I)} \rightarrow  C^\infty_c(I) \big)$, 
where $C^\infty_c(I)$ denotes the functions compactly supported in $I$. 

In order to encode quantum mechanics the crucial step is now to deform the differential $d$ 
on the symmetric algebra of observables by the BV operator $\Delta$. On a quadratic monomial 
consisting of two functions $f$ and $g$ of degrees zero and $-1$, respectively, 
$\Delta$ acts as 
\begin{equation}\label{DeltaIntro} 
    \Delta(f \wedge g)  = i\hbar \int_I dt f(t)g(t) \;, 
 \end{equation}
which is extended to the full symmetric algebra by requiring that $\Delta$ acts as a \textit{second order} 
operator, an algebraic notion to be explained in the main text.  
One defines the algebra of quantum observables as the 
symmetric algebra equipped with the deformed differential 
$\delta_{\rm BV} := d + \Delta$ obeying $(\delta_{\rm BV})^2=0$:
\begin{equation}
    {\rm Obs}^q(I) := {\rm Sym} 
    \Big({\rm Obs}^{\rm lin}(I), \delta_{\rm BV}\Big)\;,     
    \label{eq:obs_q_defINTRO}
\end{equation}
relative to the interval $I$ on which its observables are supported. 
Here ${\rm Obs}^{\rm lin}(I)$ denotes the graded vector space consisting of compactly supported functions in degrees zero and $-1$. 
The basic  idea is now that quantum mechanics is defined on the cohomology of such a  BV algebra. 
Unfortunately,  one cannot hope to obtain the algebra of quantum operators on the cohomology 
of the BV algebra directly, for its algebra structure given by the wedge product does not descent to 
a well-defined algebra on the cohomology, owing to $\delta_{\rm BV}$ being a second-order operator. 
In order to relate to the algebra of quantum operators we have to consider  a more subtle notion of 
algebra: factorization algebras. 

Ordinary algebras are given by a vector space $A$ together with a bilinear product $\mu: A\otimes A\rightarrow A$ 
satisfying associativity. Factorization algebras \cite{beilinson2004chiral,Costello_Gwilliam_2016,Costello_Gwilliam_2021} 
generalize this notion by including in its data 
a topological space $X$ and an assignment of a vector space $\mathfrak{F}(U)$ to each open set $U\subset X$. 
Furthermore, there are two kinds of structure maps: \textit{i)} for each inclusion of open sets 
${U} \subseteq {V} \subseteq X$ one has a linear map of vector spaces 
$m_{{U}}^{{V}} : \mathfrak{F}({U}) \longrightarrow \mathfrak{F}({V})$ 
and \textit{ii)}  for each two disjoint open sets ${U}_1, {U}_2$ contained in some open set ${V}$ 
there is a bilinear map of vector spaces  
$m_{{U}_1\, {U}_2}^{{V}} : \mathfrak{F}({U}_1) \otimes \mathfrak{F}({U}_2) \longrightarrow \mathfrak{F}({V})$. 
The structure maps are subject to certain relations. An ordinary algebra $A$ is a special case of a 
factorization algebra where one takes the topological space to be $\mathbb{R}$ and assigns to each 
open interval $I\subset \mathbb{R}$ the vector space $A$ of the algebra. Since all vector spaces 
$\mathfrak{F}(I)$ are then equal we can take the first structure map to be just the identity and the second structure map 
to be the bilinear product $\mu$, in which case the factorization algebra relations 
reduce to associativity.

An important observation  is that canonical  quantum mechanics itself 
can be viewed as a factorization algebra as follows. The topological space is the closed interval 
$[t_i,t_f]\subset \mathbb{R}$, 
which has four kinds of relevant open subsets: 
\textit{i)} open intervals $I=(a,b)$ to which one assigns the space of operators ${\cal O}$; \textit{ii)} 
half-open intervals $[t_i, a)$ to which one assigns the Hilbert space ${\cal H}$ of ket states $\ket{\psi}$;
\textit{iii)} 
half-open intervals $(b, t_f]$ to which one assigns the dual Hilbert space ${\cal H}^*$ of bra states $\bra{\psi}$; 
and \textit{iv)} the closed interval $[t_i,t_f]$ to which one assigns the complex numbers $\mathbb{C}$. 
The structure maps are given by combinations of the four basic operations in quantum mechanics: 
the composition of operators ${\cal O}_1$, ${\cal O}_2$ to form ${\cal O}_1{\cal O}_2$, the  action  with 
an operator ${\cal O}$ on ket or bra states to form ${\cal O}\ket{\psi}$ or $\bra{\psi}{\cal O}$ and, finally, 
the  pairing between  bra and ket states $\bra{\psi}$ and $\ket{\eta}$ to form the complex 
number (probability amplitude) $\bra{\psi}\eta\rangle$.

We are now ready to state the central result of the present paper: For the harmonic oscillator (and similarly the 
spin-$\frac{1}{2}$ system)  there is a factorization algebra 
$\mathfrak{F}_{\rm QM}$, based on the off-shell BV algebras
 ${\rm Obs}^q(I)$ in (\ref{eq:obs_q_defINTRO}),  that is quasi-isomorphic 
--- the relevant notion of equivalence in the realm of chain complexes and factorization algebras --- 
to ordinary `on-shell' quantum mechanics viewed as a factorization algebra as sketched above. 
In the following we outline  this factorization algebra of off-shell quantum mechanics, 
which is also defined  on the closed interval 
$[t_i,t_f]\subset \mathbb{R}$. The factorization algebra $\mathfrak{F}_{\rm QM}$ assigns to each open interval 
$I$ the BV algebra ${\rm Obs}^q(I)$ in (\ref{eq:obs_q_defINTRO}).  The first structure map is defined, for open intervals 
$I_1\subseteq I_2$, as the trivial  inclusion that views on object in ${\rm Obs}^q(I_1)$, i.e.~a function, or polynomial  of functions, 
that is compactly supported in $I_1$, as such an object in ${\rm Obs}^q(I_2)$. The second structure map 
 $m_{I_1I_2}^J:\mathfrak{F}_{\rm QM}(I_1)\otimes \mathfrak{F}_{\rm QM}(I_2)\rightarrow \mathfrak{F}_{\rm QM}(J)$  for disjoint intervals $I_1, I_2$ contained in some interval  $J$ 
 is  given by the wedge product, 
  \begin{equation}\label{StructureMapIntro} 
  m_{I_1I_2}^J(F_1\otimes  F_2) \ = \ {F_1} \wedge {F}_2 \ \in \  \mathfrak{F}_{\rm QM}(J)\;. 
  \end{equation} 
Here one uses that $F_1$ and $F_2$ are compactly supported in $I_1$ and $I_2$, respectively, and 
hence are also compactly supported in $J$, and so the wedge product on the right-hand side 
is the one in $\mathfrak{F}_{\rm QM}(J)={\rm Obs}^q(J)$. Finally, to half-open intervals  $[t_i, a)$ and $(b, t_f]$, 
and to the closed interval $[t_i,t_f]$, the factorization algebra assigns the same kind of BV algebra, 
but  based on function spaces with boundary conditions on one or two boundary points. 
The  structure maps are defined as before.

The striking  economy of the above scheme must be emphasized: the only non-trivial operation 
(structure map of the factorization algebra) 
is the wedge product in (\ref{StructureMapIntro}), which we recall is just the concatenation of symbols. 
In particular, the alleged hallmark of 
quantum physics, the Heisenberg commutation relation, is not present, 
as the wedge product is graded symmetric.  
This should be compared  with the corresponding operations of canonical  quantum mechanics, 
which include  the composition of the operator algebra and the action of operators on state vectors. 
For instance, in the case of an $n$-dimensional Hilbert space the structure maps involve 
the multiplication of $n\times n$ matrices whose computational complexity 
is ${\cal O}(n^3)$.\footnote{There are, however, surprising  reductions of the computational complexity of matrix multiplication: 
\textit{https://www.quantamagazine.org/new-breakthrough-brings-matrix-multiplication-closer-to-ideal-20240307/}.}
In contrast, the wedge product of two words just puts them 
next to each other.\footnote{Just to make this very plain: the multiplication of two numbers of even moderate 
size, such as $785$ and $1921$, is already pretty hard, while their wedge product just puts a wedge 
between them: $785\wedge 1921$.   } 
The catch is that in order to compute the quantities of physical interest, the probability amplitudes, in a last step 
one has to pass over to cohomology (one has to go `on-shell'), 
which largely reintroduces the computational complexity of quantum mechanics. 

Specifically, there is  a projection $\Pi$ from the off-shell to the  on-shell factorization algebra 
of quantum mechanics. In order to compute a 
familiar quantity  such as the amplitude $\bra{\chi}{\cal O}\ket{\psi}$ 
one has to pick functionals of the off-shell algebra  corresponding to the appropriate intervals: 
$F_{\psi}\in \mathfrak{F}_{\rm QM}([t_i,a))$, $F_{{\cal O}}\in \mathfrak{F}_{\rm QM}((a,b))$ and $F_{\chi}\in \mathfrak{F}_{\rm QM}((b,t_f])$, 
so that $\Pi(F_{\psi})=\ket{\psi}\in {\cal H}$, $\Pi(F_{\chi})=\bra{\chi}\in {\cal H}^*$ and $\Pi(F_{\cal O})={\cal O}$. 
One then takes their wedge product
 \begin{equation}
   F_{\chi}\wedge F_{\cal O}\wedge F_{\psi} \ \in \ \mathfrak{F}_{\rm QM}([t_i,t_f])\;, 
 \end{equation}   
which yields  a functional corresponding to the full closed interval $[t_i,t_f]$. For this interval, $\Pi$ just projects to 
a complex number, which is in fact the desired probability amplitude: 
 \begin{equation}\label{finalamplitude} 
  \bra{\chi}{\cal O}\ket{\psi} = \Pi\big(  F_{\chi}\wedge F_{\cal O}\wedge F_{\psi}\big) \ \in \ \mathbb{C}\;. 
 \end{equation}    
While the wedge product is arguably the simplest operation one can imagine, 
the projector $\Pi$ is quite complicated as a consequence of it needing 
to commute with the BV differential $\delta_{\rm BV}$. Thus, in a sense, we have traded 
the complexity inherent in the operator algebra of quantum mechanics for the complexity of the projector 
down to cohomology. (One may, however,  envision  other and more efficient methods to compute the cohomology.) 
It should also be mentioned that, as to be expected, the off-shell formulation introduces a significant  
redundancy as there are infinitely many  functionals that project to a given state or operator. The 
final complex number in (\ref{finalamplitude}) is, of course, independent of the choice of functionals. 
The situation here is thus quite analogous to gauge theories where one introduces a gauge redundancy 
but physical observables are gauge invariant. Indeed, the passing over from gauge redundant to gauge invariant fields can also be interpreted as a quasi-isomorphism \cite{Chiaffrino:2020akd}.

The rest of this paper is organized as follows. In sec.~2 we provide a review of the required background on chain complexes, the space of observables as functionals and the 
BV algebra. Then we motivate and introduce in sec.~3 pre-factorization algebras and present the example of the off-shell and on-shell versions of the operator algebra of quantum mechanics. This follows the treatment by 
Costello and Gwilliam in  \cite{Costello_Gwilliam_2016} 
whose results we extend  by giving the explicit quasi-isomorphism between the off-shell and on-shell algebras. Our main new results are contained 
in sec.~4 and 5. In sec.~4 we present the complete off-shell factorization algebra for the quantum mechanics of the harmonic oscillator, which requires the inclusion of intervals with one or two boundary points. 
In sec.~5 this is extended to the spin-$\frac{1}{2}$ system. We close in sec.~6 
with conclusions and outlook, while some technical background is contained in three appendices. In particular in appendix~A the notion of BV algebras and their cohomology 
is motivated with  
finite-dimensional Gaussian intergrals. 
This paper is intended to be understandable to a general theoretical physicist, and the  reader unfamiliar with any of this is advised to start with appendix A.

\section{Chain complexes and Batalin-Vilkovisky  algebras}

\subsection{Chain complex and observables  for harmonic oscillator}

We begin with the one-dimensional harmonic oscillator with dynamical variable $\phi(t)$, where for now 
we assume the time variable $t$ to live in the real line $\mathbb{R}$.  
More precisely, we assume that $\phi$ is a real-valued smooth function on $\mathbb{R}$:  
$\phi\in C^\infty(\mathbb{R})$. 
The classical equations of motion are given by 
\begin{equation}
        \ddot{\phi} + \omega^2\phi = 0\;, 
    \label{eq:KG}
\end{equation}
where $\omega$ is constant, and we denote time derivatives by $\partial_t=\dot{}\,$. 
We encode these data in a chain complex by introducing two copies 
of $C^\infty(\mathbb{R})$, the space of `fields'  $\phi(t)$ in degree zero and the `space of equations of motion' in degree one: 
 \begin{equation}
 0\xrightarrow[]{d_{-1}}
 C^\infty(\mathbb{R}) \xrightarrow[]{d_0= \partial_t^2 + \omega^2} C^\infty(\mathbb{R})
\xrightarrow[]{\;d_{1}\;} 0 \,. 
 \end{equation} 
Here the differential  `squares to zero' trivially, in the sense of $d_{i+1} \circ d_{i}=0$, thanks to either 
$d_{-1}$ or $d_1$ being trivial. (In the following the trivial spaces will be left implicit.)

Let us next recall that the cohomology of the above chain complex encodes the space of solutions 
of the classical equations of motion. A general solution of (\ref{eq:KG}) can be expressed in 
terms of the basis solutions 
\begin{align}\label{basisharmfunct} 
    \phi_q(t) := \cos(\omega t)\;, \qquad 
    \phi_p(t) := \frac{1}{\omega}\sin(\omega t)\;, 
\end{align}
as 
 \begin{equation}
    \phi(t) = q\phi_q(t) + p\phi_p(t)\, ,
    \label{eq:KG_Solution}
\end{equation}
where  $q$ and $p$ are real numbers. In fact, these two numbers are just the 
initial conditions at time zero: $\phi(0)=q$, $\dot{\phi}(0)=p$. 
Therefore, even though a solution $\phi$ is an element of the infinite-dimensional space of smooth functions $C^\infty(\mathbb{R})$, 
a solution  is fully encoded in the two numbers $q$ and $p$. That is, there is a one-to-one map between solutions of eq.~\eqref{eq:KG} and the two dimensional phase space $\mathbb{R}^2$. More generally, we can extend this to a non-invertible map (projection) 
from the full space $C^\infty(\mathbb{R})$ to $\mathbb{R}^2$: 
\begin{equation}
\begin{gathered}
    \pi: C^\infty(\mathbb{R}) \longrightarrow \mathbb{R}^2\,, \qquad 
    \pi(\phi) = (\phi(0), \Dot{\phi}(0))\,. 
\end{gathered}
\end{equation}
More precisely, this defines the projection map in degree zero, while in degree one it acts trivially, 
as indicated in the following diagram: 
\begin{equation}\label{quasiisoproj}
    \begin{tikzcd}
    C^\infty(\mathbb{R}) \arrow[d,"\pi"] \arrow[r,"\partial_t^2 + \omega^2"] & C^\infty(\mathbb{R}) \arrow[d] \\
    \mathbb{R}^2 \arrow[r] & 0
    \end{tikzcd} \, .
\end{equation}
It turns out that $\mathbb{R}^2$ is isomorphic to cohomology, which is defined to be the space of $d$-closed vectors modulo $d$-exact vectors:
\begin{equation}
    H_n := \ker d_n / \ima d_{n-1}\, .
\end{equation}
Indeed, the zeroth cohomology space $H_0$ is given by the space of solutions $\ker(\partial_t^2 + \omega^2)$, 
as the space in degree $-1$ is trivial, and hence, by the above identification,  $H_0$ can be viewed as $\mathbb{R}^2$. 
In order to determine the cohomology space $H_1$ we note that the kernel is the full space, 
but this full space is also the image of $\partial_t^2 + \omega^2$: 
to each function $g\in C^\infty(\mathbb{R})$ one can define $h(g)(t):=\int_0^{t}  ds\, \frac{\sin \omega (t-s)}{\omega} g(s)$ 
that obeys $(\partial_t^2 + \omega^2)h(g) = g$.  
Thus, the cohomology $H_1$ is trivial, as indicated in \eqref{quasiisoproj}. 
The map $\pi$ is of course not an isomorphism, but it is a 
\textit{quasi-isomorphism} \cite{weibel1994introduction}.
To explain this notion we first note that $\pi$ is  a chain map, which means that the
above diagram commutes. This in turn implies that $\pi$ is well-defined on cohomology. If, in addition, 
it is an isomorphism on  cohomology it is called a \textit{quasi-isomorphism}, which is 
the case here. 
Thus, the full `off-shell' complex is quasi-isomorphic to 
the `on-shell' complex of cohomology, which in this case is  the phase space of the  
harmonic oscillator.

\medskip

Let us then turn to the observables for these physical variables, which will be certain functionals 
on the above space. We want to think of functionals as representing  measurements on some time interval $I = (a,b)$. A special case are the  linear observables, 
which in view of their role for quantum mechanics we take to be complex valued. 
Specifically, we consider  linear functionals $F$ that  act on the total space 
as
\begin{equation}
    F[\phi] := \int_\mathbb{R} dt\, f(t)\phi(t)\;, 
    \label{eq:Observable_Def}
\end{equation}
where the $f$ are 
complex valued smooth functions on $\mathbb{R}$ \textit{with compact support} on 
the open interval $I=(a,b) \subseteq \mathbb{R}$ 
\cite{Costello_Gwilliam_2016}, 
\begin{equation}
    f \in C^\infty_c(I)\, .
\end{equation}
The functionals $F$ in \eqref{eq:Observable_Def} define a `dual'  space of linear functions on the field space $C^\infty(\mathbb{R})$.  The linear functional $F$ can then be identified with $f \in C^\infty_c(I)$. 
Note that since $f$ is compactly supported, the above integral is finite and hence $F[\phi]$ is 
well-defined.

In the following we will work almost exclusively  at the level of observables given by compactly supported smooth 
functions $f$. These also form a graded vector space since $f$ may be integrated against a function in degree zero, 
in which case it carries itself degree zero, or against a function in degree one, in which case it carries itself 
degree $-1$. (This is so in order for  the number computed by the integral to have degree zero.) 
In fact, the `dual'  space of linear observables is again a chain complex of the form 
\begin{equation}\label{LINOBS} 
    {\rm Obs}^{\rm lin}(I)\  := \ \,
 C_c^\infty(I) \xrightarrow[]{\;\;d_{-1}\;\;} C_c^\infty(I)\;, 
\end{equation}
where, however,  the non-trivial differential maps functions in degree $-1$ to 
functions in degree zero. This can be understood as follows: For a function $\phi$ of degree zero 
in the original complex we define the differential on the dual space of functionals  via 
 \begin{equation}\label{DualDiff} 
 \begin{split}
    dF[\phi] &= F[d\phi]  = F[(\partial_t^2 + \omega^2)\phi] \\
    &= \int_\mathbb{R} dt\, f(t)(\partial_t^2 + \omega^2)\phi(t) = 
    \int_\mathbb{R} dt\, (\partial_t^2 + \omega^2)f(t)\phi(t) \;, 
 \end{split}   
\end{equation}
where we integrated by parts.
Since $d\phi$ is of degree one, the function $f$ above is of degree $-1$ and hence we 
read off that the non-trivial differential on the dual complex is 
 \begin{equation}
  d_{-1} = \partial_t^2 + \omega^2\;. 
 \end{equation}

Above we projected the original chain complex to its cohomology 
(the phase space of classical solutions), 
and now we work out the corresponding projection 
of the `dual' complex by restricting the functionals  to on-shell fields. 
To this end we evaluate a linear functional of the form \eqref{eq:Observable_Def} 
on a solution written as \eqref{eq:KG_Solution}: 
\begin{equation}
\begin{aligned} 
    F[\phi] &= q\int_Idt\, f(t)\phi_q(t) + p\int_Idt\, f(t)\phi_p(t)\\
    &\equiv  q\expval{f, \phi_q} + p\expval{f, \phi_p}\,,
\end{aligned}
\end{equation}
where here $f$ is of degree zero.  
Moreover, here and in the following we abbreviate the one-dimensional time integration of the product of two functions by 
angle  brackets $\expval{\,,\;}$.
We infer  that the action of $F$ on solutions is uniquely specified by the two complex numbers  $\expval{f, \phi_q}$ and $\expval{f, \phi_p}$. 
This motivates us to define the projection of the complex of functionals, viewed in terms of the function $f$ 
representing it, as the map 
\begin{equation}\label{PiMap} 
\begin{aligned}
    \Pi  : \;C^\infty_c(I) \longrightarrow \mathbb{C}^2\;, \qquad 
   \Pi(f) = (\expval{f, \phi_q}, \expval{f, \phi_p}) \in \mathbb{C}^2\;, 
\end{aligned}
\end{equation}
for $f$ of degree zero, and zero otherwise. 
Put differently, the complex (\ref{LINOBS}) is projected as 
\begin{equation}\label{cohomologyPROJ} 
    \xymatrix{
    C^\infty_c(I) \ar[r]^d \ar[d] & C^\infty_c(I) \ar[d]^\Pi\\
    0 \ar[r] & \mathbb{C}^2 
    }
\end{equation}
where we recall that the leftmost space contains functions of degree $-1$, while the rightmost space contains functions of degree zero. 
Note that the functions in degree $-1$ being projected to zero is consistent with the cohomology $H_{-1}$ being trivial. 
This follows because the harmonic oscillator equations (\ref{eq:KG}) have no solutions that are 
compactly supported, hence  $\ker d_{-1} = 0$ and so $H_{-1}=\ker d_{-1}/{\rm im} \, d_{-2}=0$.
We will confirm next  that the total  map $\Pi$ is a quasi-isomorphism, hence 
the total complex is quasi-isomorphic to $\mathbb{C}^2$.

We prove this by showing that $\Pi$ is a chain map and has a homotopy inverse. First of all, $\Pi$ being a chain map means that it commutes with the differential, 
$\Pi\,d = d\,\Pi$, where by slight abuse of notation we denote the differentials upstairs and downstairs 
both by $d$.  Since the differential downstairs is actually trivial this in turn amounts to $\Pi\, d=0$, 
which is only non-trivial when acting on an $f$ in degree $-1$: 
 \begin{equation}\label{Pid=0} 
  (\Pi\,d_{-1})(f) = \Pi(df)  = \big(\langle df, \phi_q\rangle, \langle df, \phi_p\rangle\big) = (0,0)\;. 
 \end{equation} 
The last step follows, as in (\ref{DualDiff}), since $d=\partial_t^2+\omega^2$ can be integrated by parts, after which it annihilates the solutions $\phi_q$ and $\phi_p$.

Next we need to show that there is a homotopy inverse chain map $I: \mathbb{C}^2 \rightarrow C^\infty_c(I)$ 
in degree zero together with a \textit{homotopy map} $\mathfrak{h} : C^\infty_c(I)\rightarrow C^\infty_c(I)$ mapping degree zero functions 
to degree $-1$ functions (i.e.~having intrinsic degree $-1$),
 so that 
 \begin{equation}\label{homotopyREL} 
 \begin{split}
  \Pi \circ I &= {\rm id}_{\mathbb{C}^2}\;, \\
  I \circ \Pi  &= {\rm id} - d\circ \mathfrak{h} - \mathfrak{h}\circ d\;. 
 \end{split}
 \end{equation} 
Denoting the compactly supported function in degree zero obtained by the inclusion map from $(q,p)\in \mathbb{C}^2$ as 
$I(q,p)= \mathfrak{f}_{q,p} = q f_Q + p f_P$ with $f_Q,f_P \in C^\infty_c(I)$ the first condition implies 
 \begin{equation}
  \Pi(I(q,p)) = \Pi(\mathfrak{f}_{q,p}) = (\expval{\mathfrak{f}_{q,p}, \phi_q}, \expval{\mathfrak{f}_{q,p}, \phi_p}) = (q,p)\;. 
 \end{equation} 
Recalling the basis solutions (\ref{basisharmfunct}) this condition in turn requires 
 \begin{equation}\label{eq:FQFP}
 \begin{split}
  \int_I dt \,f_Q(t) \cos (\omega t) &= 1\;, \quad 
   \int_I dt \,f_{Q}(t) \sin (\omega t) = 0 \; , \\
  \int_I dt \,f_P (t) \cos (\omega t) &= 0\;, \quad 
   \int_I dt \,f_{P}(t) \sin (\omega t) = \omega \; .
\end{split}
 \end{equation} 
Upon choosing 
$f_Q \in C^\infty_c(I)$ satisfying the first two conditions, one can then define $f_P = - \dot f_Q$ in order to satisfy the last two conditions.

The second relation in (\ref{homotopyREL}) implies two relations depending on the degree of $f$ on which 
it is evaluated: 
 \begin{equation}\label{homotopycomp} 
 \begin{split}
  &|f|=0\,:\quad I\Pi(f) - f = -d(\mathfrak{h}(f))\;,  \\
  &|f|=-1\,: \quad -f = - \mathfrak{h}(df) \;, 
 \end{split}
 \end{equation} 
where we used that $\Pi$ and $\mathfrak{h}$ are only non-trivial when acting on degree zero functions. 
We will now construct such a homotopy map from the simpler map 
\begin{equation}
    h(f)(t) := \frac{i}{2\omega}\int ds\, e^{-i\omega |t - s|}f(s)\,, 
    \label{eq:h}
\end{equation}
which invertes the differential 
from the left and the 
right:
\begin{equation}\label{dhinvers}
 d(h(f)) =  f\;, \qquad h(df) = f \;. 
\end{equation}
The proof of this fact is given in appendix \ref{App:Propagator}.

By itself $h$ is not an adequate homotopy map, since $h(f)$ is not necessarily compactly supported. 
The proper homotopy map is instead 
given by 
\begin{equation}
    \mathfrak{h} = h \circ (1 - I \circ \Pi)\;. 
\end{equation}
Note that $\Pi \circ I = {\rm id}$ implies that  $1 - I \circ \Pi$ projects 
onto the kernel of $\Pi$ since any function of the form $f - I \circ \Pi(f)$ is annihilated by $\Pi$ and hence an element of $\ker\Pi$.  
Put differently, $h$ and $\mathfrak{h}$ differ on $C^\infty_c(I)$, but agree on $\ker\Pi$.
Since $\Pi \text d f = 0$, in degree zero we use 
the first relation in (\ref{dhinvers})
to infer 
\begin{equation}
    \begin{aligned}
        d\mathfrak{h}(f) &= (1 - I\Pi)(f)\;. 
    \end{aligned}
\end{equation}
Thus, in total, (\ref{homotopycomp}) is obeyed. 

It is left to show that $\mathfrak{h}$ is well defined as a map $\mathfrak{h}: C^\infty_c(I) \rightarrow C^\infty_c(I)$. The map $h$ is discussed further in appendix \ref{App:Propagator}, where it is shown that $h(f)$ is smooth for any $f \in C^\infty_c(I)$. Since also $f - I \circ \Pi(f)$ is smooth for any $f \in C^\infty_c(I)$, $\mathfrak{h}(f)$ is smooth. Further, it follows from the discussion in the appendix that for any $f \in \ker \Pi$, meaning that $f$ vanishes when integrated against solutions, $h(f)$ is contained in the support of $f$. Then, since $f - I \circ \Pi(f)$ is supported in $I$, so is $\mathfrak{h}(f)$. This shows that $\mathfrak{h}: C^\infty_c(I) \rightarrow C^\infty_c(I)$ is well defined as a map of smooth functions with support in $I$. 

\medskip

We now turn to a particular class of linear functionals that will be important in order  to connect to 
familiar quantum mechanics: the position and momentum observables. 
Denoting the position observable by $Q$ and the momentum observable by $P$ we parameterize 
them by compactly supported functions $f_Q$ and $f_P$, respectively, 
 \begin{equation}\label{momentumandpositionoperator} 
 Q[\phi] = \langle f_Q, \phi\rangle\;, \qquad P[\phi] = \langle f_P, \phi\rangle\;. 
 \end{equation} 
In principle we expect these functionals to reproduce, say at time $t = 0$,  
the position $\phi(0)$ and the momentum $\Dot{\phi}(0)$. However, setting 
$Q[\phi] = \phi(0)$ and $P[\phi] = \Dot{\phi}(0)$ would require $f_Q$ and $f_P$ 
to be the Dirac delta function and its derivative, respectively, which are not actually smooth functions.   
However, for the following it will be sufficient to impose $Q[\phi] = \phi(0)$ and $P[\phi] = \Dot{\phi}(0)$
only \textit{on solutions},  for which it is sufficient 
to approximate the delta function 
by a smooth function $f_Q$ whose support is 
contained in the interval $I$, and similarly for $f_P$. 
Evaluating (\ref{momentumandpositionoperator}) on a solution 
written as $\phi(t) = q\phi_q(t) + p\phi_p(t)$, c.f.~(\ref{eq:KG_Solution}), yields 
 \begin{equation}
 \begin{split}
  Q[\phi ] &= q \langle f_Q, \phi_q\rangle + p  \langle f_Q, \phi_p\rangle \stackrel{!}{=} \phi(0)= q \;, \\
  P[\phi] &= q \langle f_P, \phi_q\rangle + p  \langle f_P, \phi_p\rangle \stackrel{!}{=} \dot{\phi}(0) = p \;, 
 \end{split} 
 \end{equation} 
from which we infer the following conditions on $f_Q$, $f_P$: 
\begin{equation}\label{fQfPcond} 
    \begin{aligned}
        \expval{f_Q, \phi_q} = \expval{f_P, \phi_p} = 1\;, \\
        \expval{f_Q, \phi_p} = \expval{f_P, \phi_q} = 0\;.
    \end{aligned}
\end{equation}
Note that these are also the conditions we found in \eqref{eq:FQFP} when we constructed the inclusion map $I$. 
It should be emphasized that these conditions are the only ones we need to impose. 
In particular, $f_Q$ does not actually need to be peaked around zero, and one may even have 
two  different `representatives' $f_Q$ and $f_Q'$ satisfying the above conditions that have 
disjoint support inside $I$.  
One can easily construct compactly supported functions satisfying these integral conditions, say by the 
familiar method of combining non-analytic functions such as $e^{-1/x^2}$ and extending them by zero appropriately. 
There are of course infinitely many such functions; we assume that a specific choice has been made 
for the functions $f_Q$ and $f_P$ 
and hence for the observables $Q$ and $P$, 
satisfying the above conditions,  but we will show next  that the ambiguity entailed in this choice is immaterial 
in cohomology.\footnote{We can visualize $f_Q$ to be any smooth compactly supported finite-width approximation to the 
delta function, with the width and hence the order of approximation  being immaterial at the end. Informally we can think of 
the Dirac delta function as living in the same cohomology class as the well-defined $f_Q$, except of course 
that only equivalence classes of smooth functions actually live in the cohomology. }

Let us then return to  the homotopy retract to the cohomology $\mathbb{C}^2$ discussed above.  
We can think of the above chosen functions $f_Q$ and $f_P$ in degree zero as basis functions of the 
cohomology in that \textit{any} function $f$ in degree zero can be written as 
 \begin{equation}\label{generalFdecomp} 
  f = u f_Q + w f_P + df_{-1} \;, \qquad u, w \in \mathbb{C}\;, 
 \end{equation}
for some function $f_{-1}$ in degree $-1$. Under the projector (\ref{PiMap}) we have, using 
(\ref{fQfPcond}), $\Pi(f)=(u, w)\in \mathbb{C}^2$. Conversely, the inclusion $I$ maps  a given $(q,p)\in \mathbb{C}^2$ to 
the following function in degree zero:  
 \begin{equation}\label{InclusionFinally} 
  I(q,p) = q f_Q + p f_P\;. 
 \end{equation} 
The first relation in (\ref{homotopyREL}), $\Pi\circ I = {\rm id}$, then follows immediately. 
The relation (\ref{InclusionFinally}) 
explains the ambiguity in the choice of $f_P$ and $f_Q$: it corresponds to the 
ambiguity of how to embed the cohomology into the full space. Since, as we will argue, 
physical quantum mechanics in particular 
only lives in cohomology this  specific choice is irrelevant.

We close this subsection by introducing a particular change of basis of observables that  sometimes will 
be more convenient  for the relation with familiar quantum mechanics. 
Instead of working with $Q$ and $P$ we can also use the observables  $A$ and $A^\dagger$ defined as
\begin{align}
    A := \sqrt{\frac{\omega}{2\hbar}}\Big(Q + \frac{i}{\omega}P\Big)\;, \qquad 
    A^\dagger := \sqrt{\frac{\omega}{2\hbar}}\Big(Q - \frac{i}{\omega}P\Big)\;.
\end{align} 
Parameterizing these functionals in terms of compactly supported smooth functions as 
 \begin{equation}
  A[\phi] = \langle f_a, \phi\rangle\;, \qquad A^{\dagger}[\phi] = \langle f_{a^{\dagger}} , \phi\rangle  \,, 
 \end{equation} 
we have 
 \begin{equation}
  f_a = \sqrt{\frac{\omega}{2\hbar}}\Big( f_Q+\frac{i}{\omega}f_P\Big) \;, \qquad
  f_{a^{\dagger}}  = \sqrt{\frac{\omega}{2\hbar}}\Big( f_Q- \frac{i}{\omega}f_P\Big)\;. 
 \end{equation} 
In terms of these and the similarly redefined basis of classical solutions 
\begin{equation}\label{sigmaPLUSM} 
    \sigma_\pm := \sqrt{\tfrac{\hbar}{2\omega}}\,e^{\mp i\omega t} = \sqrt{\tfrac{\hbar}{2\omega}}\big(\phi_q \mp i\omega \phi_p\big)\;,
\end{equation}
with inverse $\phi_q=\sqrt{\frac{\omega}{2\hbar}}(\sigma_+ + \sigma_-)$, $\phi_p=i \sqrt{\frac{\omega}{2\hbar}}(\sigma_+ - \sigma_-)$, 
the normalization conditions (\ref{fQfPcond}) on $f_Q$ and $f_P$ translate to normalization  conditions on $f_a$ and $f_{a^\dagger}$: 
\begin{equation}\label{fafadaggercomp} 
\begin{split} 
    \expval{f_a, \sigma_+} = 1\,, &\qquad\; \expval{f_a, \sigma_-} = 0\;, \\
    \expval{f_{a^\dagger}, \sigma_+} = 0\, , &\qquad \expval{f_{a^\dagger}, \sigma_-} = 1\, .
\end{split} 
\end{equation}
Finally, the projection $\Pi : C^\infty_c(I) \rightarrow \mathbb{C}^2$ defined in  (\ref{PiMap})  
can be rewritten after a quick computation as follows: 
\begin{equation}
\begin{aligned}
    \Pi (f) &= \expval{f, \sigma_+}a + \expval{f, \sigma_-}a^\dagger 
    \equiv f_-a + f_+a^\dagger\;,
\end{aligned}
\label{eq:projector_classical}
\end{equation}
where 
 \begin{equation}\label{aadaggerbasis} 
  a \equiv \sqrt{\tfrac{\omega}{2\hbar}}\,(\,1\,,\,i\,) \;, \qquad  a^{\dagger}  \equiv \sqrt{\tfrac{\omega}{2\hbar}}\,(\,1\,,\,-i\, )\;, 
 \end{equation} 
defines a new basis of $\mathbb{C}^2$, which here is expressed in the canonical basis, 
although the specific form will be irrelevant for 
the following applications.  
We will frequently use the abbreviation $f_\pm = \expval{f, \sigma_\mp}$  in later sections.

\subsection{Symmetric algebra}

We now enlarge the vector space of linear observables or functionals to include a certain class 
of multi-linear operators. 
We could consider observables 
on fields $\phi$ of the form 
 \begin{equation}\label{Completion}
    F[\phi_1,\ldots,\phi_n] = \int dt_1 \cdots \int dt_n \,f(t_1, \ldots , t_n)\,\phi_1(t_1)\cdots \phi_n(t_n)\,,
\end{equation}
with smooth compactly supported coefficient functions $f$. 
However, 
in the following it will be sufficient to restrict to functionals of the above form where the coefficient functions are the 
product of smooth compactly supported functions: $f(t_1, \ldots , t_n)=f_1(t_1)\cdots f_n(t_n)$.
A functional $F[\phi_1,\ldots,\phi_n]$ can then be written in terms of our above notation as 
 \begin{equation}\label{functionalFF} 
  F[\phi_1,\ldots,\phi_n] = \langle f_1,\phi_1\rangle \cdots \langle f_n, \phi_n\rangle \;. 
 \end{equation}
This writing suggests to think of a non-linear functional as a formal string  $f_1f_2\cdots f_n$ of functions in 
the dual complex. The total vector space of such objects  is also known as the \textit{symmetric algebra} denoted Sym. 
We thus define 
\begin{equation}\label{ClassObs} 
    {\rm Obs}^{cl}(I) := {\rm Sym} \Big(\xymatrix{C^\infty_c(I) \ar[r]^{d_{-1}} & C^\infty_c(I)}\Big)\;.
\end{equation}
More precisely, elements of the symmetric algebra are defined to be formal products (polynomials) of functions
\begin{equation}\label{littlefWord} 
   F =  f_1 \wedge f_2 \wedge \ldots \wedge f_n\;, 
\end{equation}
which in the following we will typically denote by capital letters $F$, $G$, etc., with 
the understanding that they can also be viewed as functionals acting as (\ref{functionalFF}). 
Note that such words may contain no $f$, in which case it is just a (complex) number, one $f$  or any finite number of $f$'s, 
which defines a monomial degree or grading of the space. The previous space of linear observables 
is included as the subspace of linear words. Moreover, in the symmetric algebra the formal wedge product 
is graded commutative in that we demand 
\begin{equation}
    f_1 \wedge f_2 = (-1)^{|f_1||f_2|}f_2 \wedge f_1\;,
\end{equation}
and similarly for longer words. Here $|f|$ denotes the degree of $f$, i.e., $|f|=-1$ if $f$ belongs to the left-most space 
and $|f|=0$ is $f$ belongs to the right-most space.  

The above wedge product defines a graded commutative associative product on the symmetric algebra. 
The latter also becomes a \textit{differential graded} algebra, where the map $d$ is extended to the symmetric algebra by 
demanding that it acts 
like $d_{-1}$ in \eqref{ClassObs} on linear observables and on higher monomials 
via the Leibniz rule: 
\begin{equation}\label{LEIBNIZ} 
    d(F \wedge G) = dF \wedge G + (-1)^{|F|} F \wedge dG\, ,
\end{equation}
with $F$ and $G$ general elements of the symmetric algebra, 
and where the degree $|F|$ of a word $F$ is the sum of the degrees of its individual functions or `letters'. 
This makes  the symmetric algebra  ${\rm Obs}^{cl}(I) $ defined in (\ref{ClassObs})  
into a differential graded commutative associative algebra.

Let us next ask  how to lift  the projection $\Pi$ to the cohomology to the full 
symmetric algebra. This can simply be done by demanding that $\Pi$ 
acts like a morphism:\footnote{In particular, on a number $w\in \mathbb{C}$ (an object of monomial degree zero) 
we set $\Pi(w) = w$ so that for a degree zero function $f$ we have $\Pi(wf)=w\Pi(f) = \Pi(w)\Pi(f)$.   } 
\begin{equation}\label{Piviamorph}
    \Pi(f_1 \wedge f_2 \wedge \ldots \wedge f_n) = \Pi(f_1) \wedge \Pi(f_2) \wedge \ldots \wedge \Pi(f_n)\;.
\end{equation}
This defines a map 
 \begin{equation}
  \Pi: \;{\rm Obs}^{cl}(I) \ \longrightarrow \ {\rm Sym}(\mathbb{C}^2)\;, 
 \end{equation} 
from the symmetric algebra of classical observables to the symmetric algebra of $\mathbb{C}^2$. 
Note that since $\Pi$ is a chain map on the individual spaces and since the derivation acts 
via the Leibniz rule on the symmetric algebra, $\Pi$ is also a chain map on the full symmetric algebra, 
which will not be true for the BV differential to be discussed shortly. 
Also note that since the original $\Pi$ is only non-trivial on degree zero functions, $\Pi$ acting on a general word 
(\ref{littlefWord}) is only non-zero if and only if every letter is of degree zero.

Let us now return to the commutative associative algebra structure on ${\rm Obs}^{cl}(I)$. Does it descend 
to a  commutative associative product on the cohomology? 
To address this question we denote the equivalence class containing a word $F$ by square brackets, i.e., 
 \begin{equation}
    [F] := \left\{F' \in {\rm Obs}^{cl}(I) \;\Big|\; \Pi(F') = \Pi(F)\right\}\;, 
\end{equation}
which  is of course equivalent to  
demanding $[F]=[G]$ if and only if $F=G + dF'$, c.f., the discussion around (\ref{generalFdecomp}). 
The natural definition of the product on cohomology is given by 
\begin{equation}
    [F] \star [G] := [F \wedge G]\,. 
\end{equation}
This  is by definition graded commutative and associative, but we have to 
verify that it is well-defined on cohomology. To this end we have to show that
 the product is independent of the chosen representatives. Replacing then $F$ 
 by $F+dF'$ we have 
  \begin{equation}\label{firstorderwelldefined} 
   [F+dF'] \star [G] = [(F+dF')\wedge G]  = [F\wedge G + d(F'\wedge G)] = [F\wedge G] \;, 
  \end{equation} 
where we used the Leibniz rule (\ref{LEIBNIZ}) and $dG=0$ by degree reasons. 
Thus, the product on cohomology is well-defined, which we see here to be a direct consequence of the 
differential $d$ being first order in the sense of obeying the Leibniz rule. 
This will change in the BV algebra needed for quantization, to which we turn next.

\subsection{BV algebra}
\label{sec:BValgebra}

So far we have shown that the symmetric algebra ${\rm Obs}^{cl}(I)$ defined in (\ref{ClassObs}) carries 
a differential graded commutative algebra that descends on cohomology to the symmetric algebra of $\mathbb{C}^2$. 
In terms of the basis  in (\ref{aadaggerbasis}) this vector space  consists 
of all words in  $a$ and $a^\dagger$ which is fully commutative.  
In order to recover the Weyl algebra of quantum mechanics, with $a$ and  $a^\dagger$ satisfying the 
commutation relation $[a, a^\dagger] = 1$, we have to consider the Batalin-Vilkovisky (BV) algebra in which the differential is deformed by the BV operator, with $\hbar$ being the deformation parameter.

The BV operator acts on the symmetric algebra and reduces the monomial degree by two. 
It thus acts on numbers and linear monomials trivially. On a quadratic  monomial 
consisting of two functions $f$ and $g$ with degrees zero and $-1$, respectively,  it acts as \cite{Costello_Gwilliam_2016}\footnote{The form for the BV operator  given here is not the familiar one given by  Batalin and Vilkovisky \cite{batalin1981gauge, batalin1983quantization}, but rather uses the formulation of \cite{Costello_Gwilliam_2016}. In the more familiar formulation the BV operator is written 
as 
\begin{equation}
 \Delta : = i\hbar \int_I dt \frac{\delta^2}{\delta \phi(t) \delta \phi^*(t)}\,.  
\end{equation}
This acts on a functional of a degree zero field $\phi$  and a degree  one `anti-field' $\phi^*$,  
 \begin{equation}
  F[\phi,\phi^*] = \int dt_1 \phi(t_1)f(t_1) \int dt_2 \phi^*(t_2) g(t_2)\,, 
 \end{equation}
as $ \Delta  F =  i\hbar \int_I dt f(t)g(t)$, in agreement with (\ref{BVDELTA}).}
\begin{equation}\label{BVDELTA} 
    \Delta(f \wedge g) := i\hbar \expval{f, g} = i\hbar \int_I dt f(t)g(t) \;. 
 \end{equation}
Note that since $\Delta$ eliminates a degree $-1$ function  it effectively increases the total degree by $+1$. 
Like the original differential $d$, $\Delta$ thus has intrinsic degree $+1$. 
On the  full symmetric algebra $\Delta$
acts as 
\begin{equation}\label{fullDeltaAction} 
    \Delta(f_1 \wedge \ldots  \wedge f_n) := \sum_{i < j}(\pm)\Delta(f_i \wedge f_j)f_1 \wedge \ldots \wedge \widehat{f}_i \wedge 
    \ldots \wedge \widehat{f}_j \wedge \ldots  \wedge f_n\;,
\end{equation}
where the notation $\widehat{f}_i$ indicates that these elements are left out, and the sign is the natural one obtained by 
moving $f_i$ and $f_j$ to the front. 
This formula determines a second order differential operator on $\text{Obs}^{cl}(I)$.\footnote{In general, on any symmetric algebra $\text{Sym}(V)$ over a graded vector space $V$, an order-$n$ differential operator $\partial_D$ is in one to one correspondence with a linear map $D: \mathbb{C} \oplus V \oplus \ldots V^{\wedge n} \rightarrow \text{Sym}(V)$. For example, $d_{cl}: C^\infty_c(I) \rightarrow C^\infty_c(I) \subseteq \text{Obs}^{cl}(I)$ determines a first order operator on $\text{Obs}^{cl}(I)$ via the Leibniz rule and $\Delta: C^\infty_c(I) \wedge C^\infty_c(I) \rightarrow \mathbb{C} \subseteq \text{Obs}^{cl}(I)$ determines a second order differential operator via \eqref{fullDeltaAction}. The reader unfamiliar with these facts can consult for example \cite{Markl:1997bj}, where higher order differential operators, both in general and on symmetric algebras over vector spaces, are discussed.}
 
With the above rule one may verify that $\Delta^2=0$. For instance, for degree zero functions $f_1$, $f_2$ and degree 
$-1$ functions $g_1$, $g_2$ the BV operator acts on the quartic  monomial built with these as 
 \begin{equation}\label{Deltadegreefour} 
 \begin{split}
  \Delta(f_1\wedge g_1\wedge f_2\wedge g_2) = &\,\Delta(f_1\wedge g_1) \,f_2\wedge g_2
  -\Delta(f_1\wedge g_2) \,g_1\wedge f_2\\
  &\,+\Delta(g_1\wedge f_2) \,f_1\wedge g_2 -\Delta(f_2\wedge g_2) \,f_1\wedge g_1\;. 
 \end{split} 
 \end{equation} 
Here we have only displayed the terms where $\Delta$ acts on one $f$ and one $g$, for otherwise it gives 
zero by degree reasons. Note that a sign is generated whenever one $g$ function has been moved past another 
$g$ function, since they have degree $-1$ and are hence odd. Owing to these signs it is evident that 
in $\Delta^2(f_1\wedge g_1\wedge f_2\wedge g_2)$ terms cancel pairwise, proving $\Delta^2=0$ 
on quartic  monomials. 
The general case follows after observing that $\Delta^2 = \frac{1}{2}[\Delta,\Delta]$ is a third order operator, since it is the commutator of two second order operators. Therefore, $\Delta^2$ is necessarily determined by its action 
\begin{equation}
\text{Sym}^{\le 3}\big(\xymatrix{C^\infty_c(I) \ar[r]^{d_{-1}} & C^\infty_c(I)}\Big) \;,
\end{equation} 
where it is trivially zero.

While $\Delta$ is nilpotent it does not act via the Leibniz rule. To illustrate this let 
$f_1$, $f_2$, $g_1$, $g_2$ be defined as above, and we define the quadratic  monomials 
 $F_1 = f_1 \wedge g_1$, $F_2 = f_2 \wedge g_2$. 
Using (\ref{Deltadegreefour}) one quickly infers that $\Delta$ does not act like a derivation; 
rather, the failure of $\Delta$ to obey the Leibniz rule, 
\begin{equation}
     (-1)^{|F_1|}\{F_1, F_2\} : = \Delta(F_1 \wedge F_2) - \Delta(F_1) \wedge F_2 - (-1)^{|F_1|}F_1 \wedge \Delta(F_2) \;, 
\end{equation}
is given by 
\begin{equation}\label{antibracket} 
    \{F_1, F_2\} = -i\hbar \langle  f_1, g_2\rangle g_1\wedge f_2  +i\hbar \langle  g_1, f_2 \rangle f_1\wedge g_2\;.  
\end{equation}
This defines the `anti-bracket' that, thanks to $\Delta$ being second order,  obeys the graded Jacobi identity and 
thus defines a graded Lie algebra. 
The data of a vector space (here ${\rm Obs}^{cl}(I)$) equipped  with a graded symmetric product (here the wedge product 
of the symmetric algebra) and a degree one differential $\Delta$ that is of second order defines a \textit{BV algebra}.

The main point is now that quantum mechanics will be encoded in the cohomology of a BV algebra based on 
the differential 
 \begin{equation}\label{deltaBV} 
  \delta_{\rm BV} := d + \Delta\;. 
 \end{equation} 
This defines a BV algebra structure on the symmetric algebra  since the sum of a first order operator and a second order operator is still of second order. 
Note that  since both the original differential $d$ and the BV operator $\Delta$ are of degree one, $ \delta_{\rm BV}$ is of degree one. 
Moreover, we have $ \delta_{\rm BV}^2=0$  since  $d^2=\Delta^2=0$ and one may verify 
that $d\Delta+\Delta d=0$\footnote{This is the graded commutator of the first order operator $d$ and 
the second order operator $\Delta$, which is again a second order operator, for which in turn it is sufficient to verify this relation on quadratic monomials.}. 
We thus define the BV algebra underlying quantum mechanics as 
\begin{equation}
    {\rm Obs}^q(I) := {\rm Sym} ({\rm Obs}^{\rm lin}(I),\delta_{\rm BV})\;.
    \label{eq:obs_q_def}
\end{equation}

Replacing $d$ by $\delta_{\rm BV}$ leaves the cohomology as encoded in (\ref{cohomologyPROJ}) 
unchanged, which is thus still given by $\mathbb{C}^2$. This could be shown using the perturbation lemma \cite{crainic2004perturbationlemmadeformations}, viewing 
the BV operator $\Delta$ as a perturbation of $d$, but we will in a moment give an explicit quasi-isomorphism to prove this fact. 
The map $\Pi$ onto the cohomology, acting as a morphism via (\ref{Piviamorph}) on the symmetric algebra, 
is no longer a chain map w.r.t.~$\delta_{\rm BV}=d+\Delta$. 
To see this consider functions $f$ and $g$ with $|f|=0$, $|g|=-1$.  
Then 
 \begin{equation}\label{chainmapproppp} 
  (\Pi\circ \delta_{\rm BV})(f\wedge g) = \Pi(f\wedge dg + \Delta(f\wedge g)) 
  =\Pi(f)\wedge \Pi(dg) + \Delta(f\wedge g) = i\hbar \langle f, g \rangle \;, 
 \end{equation} 
where we used $df=0$ by degree reasons and $\Pi(dg)=0$ by (\ref{Pid=0}). 
The right-hand side is generally non-zero. 
However, for $\Pi$ and the differentials to commute this should really be zero since the cohomology 
has no differential left. We can next define a proper chain map by correcting $\Pi$. 
Let us rename the current map acting as a morphism on the symmetric algebra $\Pi_0$; 
the corrected map is then defined as 
\begin{equation}\label{improvedPI} 
    \Pi := \Pi_0\, e^{-C}\;, \quad \text{where}\quad C:= \Delta \circ (1 \otimes h)\;, 
\end{equation} 
with the map $h$ in (\ref{eq:h}), so that $C(f_1\wedge f_2)=\Delta(f_1\wedge h(f_2))$ for $f_1, f_2$ of degree zero and zero otherwise. $C$ is extended to $\text{Obs}^{cl}(I)$ as a second order derivation. It now follows that 
\begin{equation} 
[d,C] = -\Delta\;. 
\end{equation} 
For instance, on the quadratic monomial $f\wedge g$, with $f, g$ defined as above, 
we have 
 \begin{equation}
   [d,C](f\wedge g) = -Cd(f\wedge g) 
   =-C(f\wedge dg) = -\Delta(f\wedge h(dg)) 
   = - \Delta(f\wedge g)\;. 
 \end{equation}
The general case follows since $[d,C]$ is of second order. By the same argument that $[\Delta,\Delta] = 0$, it then also follows that $[C,\Delta] = 0$. From this we deduce that
\begin{equation}
    [d,e^{-C}] = -e^{-C}[d,C] = e^{-C} \Delta \, ,
\end{equation}
from which we deduce that 
\begin{equation}
d e^{-C} = e^{-C}(d+\Delta) = e^{-C} \delta_{\rm BV}\;. 
\end{equation}
Therefore, $e^{-C}$ defines a chain map
\begin{equation}
    e^{-C}: \,\text{Obs}^q(I) \rightarrow \text{Obs}^{cl}(I) \, ,
\end{equation}
which is in fact an isomorphism, since it has an inverse given by $e^{C}$. Then, since $\Pi_0: \text{Obs}^{cl}(I) \rightarrow \text{Sym}(\mathbb{C}^2)$ is a quasi-isomorphism, the composition $\Pi = \Pi_0 \circ e^{-C}$ defines a quasi-isormorphism from the quantum observables $\text{Obs}^q(I)$ to $\text{Sym}(\mathbb{C}^2)$.

While the above shows that also for the BV algebra of quantum observables there is a 
well-defined chain map to the cohomology $\mathbb{C}^2$, the 
algebra structure does not descent to an algebra on 
the cohomology. This is so because we needed the differential defining the cohomology  to be of first order 
with respect to the product, c.f., the discussion around (\ref{firstorderwelldefined}). 
This problem will be solved in the next section by considering factorization algebras 
which consist of a family of spaces ${\rm Obs}^q(I)$ for non-overlapping open intervals $I$, and on such 
intervals the failure of $\delta_{\rm BV}$ to be first order vanishes. 
This can be seen by inspecting the  anti-bracket (\ref{antibracket}), which  vanishes if the two arguments are defined 
on disjoint intervals, for then integrals such as $\langle  f_1, g_2\rangle$ are zero. 
The resulting algebra structure 
then encodes the familiar operator algebra  of quantum mechanics including  $[q, p] = i\hbar$.

\section{Factorization algebras and Weyl algebra}

\subsection{Prefactorization algebras}

So far we have established a BV algebra structure on the symmetric algebra (\ref{eq:obs_q_def})  of quantum observables   
with differential $\delta_{\rm BV} := d + \Delta$. However, since this differential is of second order with respect to the 
graded commutative product  of the symmetric algebra, the cohomology of $\delta_{\rm BV}$ does not 
inherit a well-defined algebra structure. In order to obtain the familiar operator algebra of quantum mechanics we 
have to enlarge the framework once more to \textit{factorization algebras}. 
The latter assign a vector space to each open set of spacetime. 
In quantum mechanics, we consider `spacetime' to be $\mathbb{R}$, where to each open interval $I$, we assign the vector spaces $\mathfrak{F}(I):= {\rm Obs}^q(I)$ of quantum observables on $I$. 
The factorization algebra structure includes a bilinear 
product $m_{I_1,I_2}^{J}:\mathfrak{F}(I_1)\otimes \mathfrak{F}(I_2)\rightarrow \mathfrak{F}(J)$
only for \textit{disjoint} intervals $I_1$, $I_2$ contained in $J$. 
Thanks to the intervals being disjoint, the BV differential $\delta_{\rm BV}$ becomes effectively first order 
since its failure to be so encoded in the anti-bracket (\ref{antibracket}) vanishes. 
The algebra structure inherited by the cohomology of $\delta_{\rm BV}$ 
then turns out to be the Weyl algebra of quantum mechanics. 
Indeed, we will show that there  is a quasi-isomorphism from the  factorization algebra to the Weyl algebra of quantum mechanics. 

We will now give a minimal definition of a prefactorization algebra, which will determine a prefactorization algebra as defined in  \cite{costelloFactorization}:
\begin{defn}
A \textit{prefactorization algebra} on a topological space $X$, in the following usually denoted as $\mathfrak{F}$, consists of 
the following data and axioms:
\begin{itemize}
    \item{For each connected open set ${U} \subseteq X$ one assigns a vector space $\mathfrak{F}({U})$.\\[1ex] 
 \textit{\big[For us these are the vector spaces $\mathfrak{F}(I):= {\rm Obs}^q(I)$ of quantum observables on $I$.\big]} }
    
    \item{For each inclusion of connected open sets ${U} \subseteq {V} \subseteq X$ one has a linear map of vector spaces 
    $m_{{U}}^{{V}} : \mathfrak{F}({U}) \longrightarrow \mathfrak{F}({V})$.\\[1ex] 
   \textit{\big[For our example  this is the natural inclusion map $m_{I}^{J}: \,\mathfrak{F}(I)\rightarrow \mathfrak{F}(J)$ 
   for intervals $I \subseteq J$: any functional  $F$ in $\mathfrak{F}(I)$ also has compact support on $J$.\big] }
   }

    \item{For each two disjoint connected open sets ${U}_1, {U}_2$ contained in some connected open set ${V}$, we have a bilinear map of vector spaces  
    $m_{{U}_1\, {U}_2}^{{V}} : \mathfrak{F}({U}_1) \otimes \mathfrak{F}({U}_2) \longrightarrow \mathcal{F}({V})$. \\[1ex] 
  \textit{\big[In our example, the map 
  $m_{I_1,I_2}^{J}:\mathfrak{F}(I_1)\otimes \mathfrak{F}(I_2)\rightarrow \mathfrak{F}(J)$  for disjoint intervals $I_1, I_2$ contained in $J$ 
  takes functionals $F_1 \in \mathfrak{F}(I_1)$ 
  and $F_2\in \mathfrak{F}(I_2)$ and then defines $m_{I_1,I_2}^J(F_1,F_2) = m_{I_1}^{J}(F_1) \wedge m_{I_2}^{J}(F_2)$. Note that this is the (formal) wedge product, 
  which is non-trivial, in contrast to the point-wise product of functions which with the two functions having support in disjoint intervals 
  is trivially zero.\,\big]}}
    \item{The inclusion maps are compatible in the sense that, for each inclusion of open sets $U \subseteq V \subseteq W$, we have $m_V^W \circ m_U^V = m_U^W$. Furthermore, we require that $m_V^V = \text{id}_{\mathfrak{F}(V)}$.\\[1ex] 
 \textit{ \big[For our case, this is trivially obeyed.\big]}}
    \item{The inclusion $m_{U}^V$ and the product $m_{U_1,U_2}^V$ are compatible in the sense that for connected and open sets $U_1,U_2 \subseteq V \subseteq W$, with $U_1,U_2$ disjoint, we have 
    \begin{equation}
     m_{V}^W \circ m_{U_1,U_2}^V = m_{U_1,U_2}^W \, .
    \end{equation}
    Similarly, for $U_1 \subseteq V_1$, $U_2 \subseteq V_2$, $V_1 \cup V_2 \subseteq W$, with $U_1,U_2$ disjoint and separately $V_1,V_2$ disjoint, we have
        \begin{equation}
    m_{V_1,V_2}^W \circ (m_{U_1}^{V_1} \otimes m_{U_2}^{V_2}) = m_{U_1,U_2}^W \, .
    \end{equation}
    \textit{\big[In our example, this is obeyed due to the triviality of the embeddings $C^\infty_c(I) \rightarrow C^\infty_c(J)$ for any open intervals $I \subseteq J$.\big]}}
    
    \item{The maps $m_{U_1,U_2}^V$ obey natural associativity relations. For each three disjoint and connected open sets ${U}_1, {U}_2, {U}_3$ contained in some connected open set ${W}$ and connected open sets ${V}_1,{V}_2$, 
     so that ${U}_1, {U}_2 \subseteq {V}_1$ and ${U}_2, {U}_3 \subseteq {V}_2$,  that are themselves contained in ${W}$, one demands the compatibility condition 
    \begin{equation}\label{CompaAss} 
         m_{{V}_1{U}_3}^{{W}} \circ \left( m_{{U}_1{U}_2}^{{V}_1}\otimes {\rm id} \right)  
        = m_{{U}_1{V}_2}^{{W}} \circ 
        \left( {\rm id} \otimes m_{{U}_2 {U}_3}^{{V}_2}\right) \;. 
    \end{equation}
This is indicated in the figure.\\[1ex] 
 \textit{ \big[In our example, with all sets being open intervals, both sides of the equation, when evaluated  
  on $F_1\in \mathfrak{F}(U_1)$, $F_2\in \mathfrak{F}(U_2)$ and $F_3\in \mathfrak{F}(U_3)$, yields 
  $F_1\wedge F_2\wedge F_3$ with all $F_{i}$ considered as elements in $\mathfrak{F}(W)$, 
  and so the compatibility condition is obeyed.\big]}}
  \item{For all $U,V,W$ connected and open as well as $U$ and $V$ disjoint with $U \cup V \subseteq W$, we have
  \begin{equation}
      m_{U,V}^W(a,b) = (-)^{ab}m_{V,U}^W(b,a)
  \end{equation}
  for all $a \in \mathcal{F}(U)$ and $b \in \mathcal{F}(V)$.\\[1ex] 
 \textit{\big[This is obeyed in our case, since
 \begin{equation}
     m_{I_1,I_2}^J(F,G) = m_{I_1}^J(F) \wedge m_{I_2}^J(G) = (-)^{FG}m_{I_2}^J(G) \wedge m_{I_1}^J(F) = m_{I_2,I_1}^W(G,F)
 \end{equation}
 for all $F \in \mathfrak{F}(I_1)$ and $G \in \mathfrak{F}(I_2)$. \big]}
  }
\end{itemize}
\end{defn}

\begin{figure}[H]
    \centering
    \begin{subfigure}[b]{0.45\textwidth}
    \includegraphics[width=\textwidth]{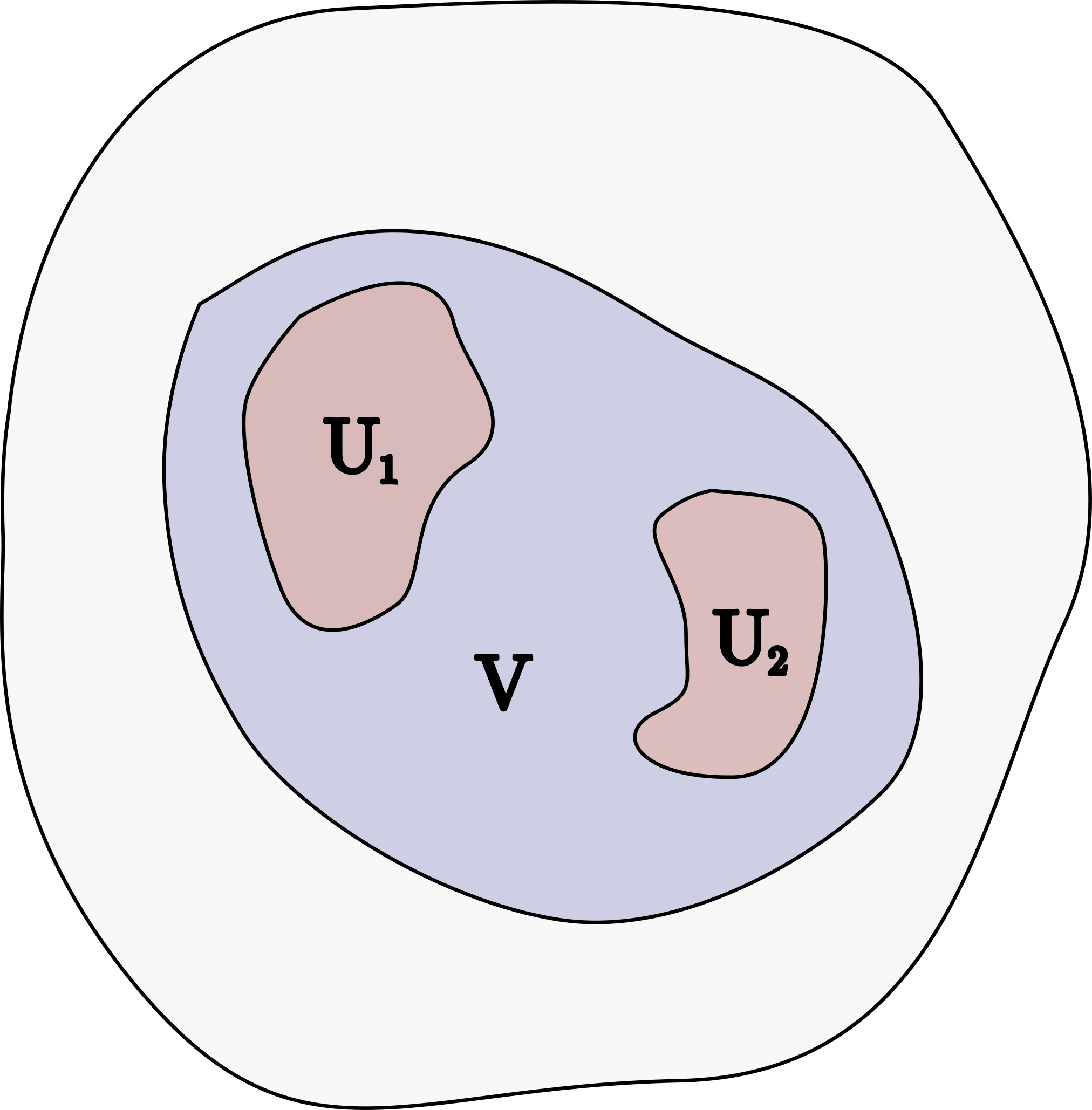}
    \caption{}
    \label{fig:prefac_map}
    \end{subfigure} 
    \hfill
    \begin{subfigure}[b]{0.45\textwidth}
    \includegraphics[width=\textwidth]{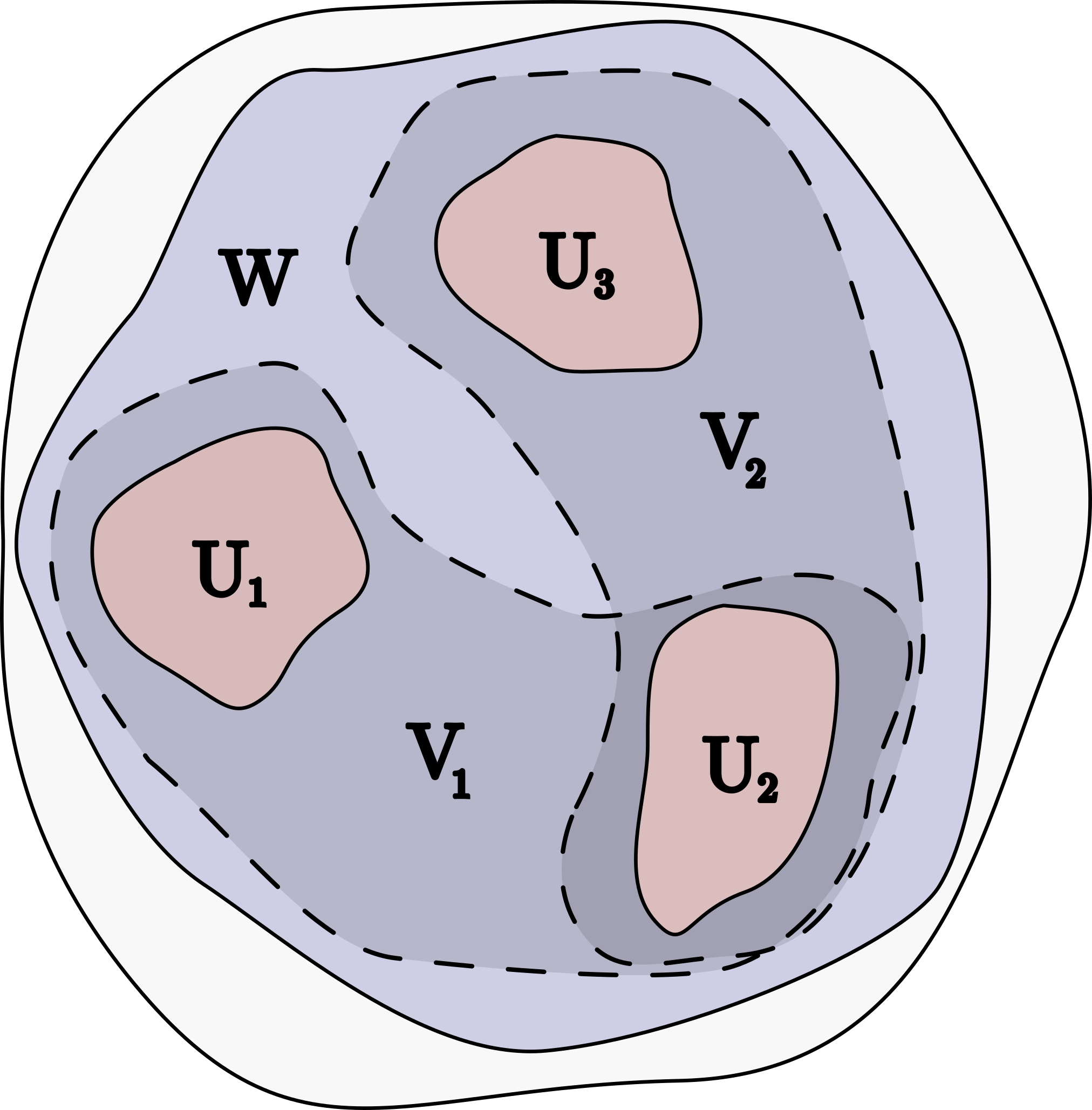}
    \caption{}
    \label{fig:prefac_ass}
    \end{subfigure}
\caption{(a): Pictorial representation of the prefactorization product map. Two elements in $\mathfrak{F}({U}_1)$ and $\mathfrak{F}({U}_2)$ are mapped into the including set $\mathfrak{F}({V})$. (b): Representation of the associativity condition. The dotted regions represent the sets ${V}_1$ and ${V}_2$. The order of inclusion (first 
${U}_1$ and ${U}_2$ into ${V}_1$ and then ${V}_1$ and ${U}_3$ into ${W}$ or first ${U}_2$ and ${U}_3$ into ${V}_2$ and then ${U}_1$ and ${V}_2$ into ${W}$) does not matter for the final answer.}
\label{fig:Ising_exchange_hexagon_flip}
\end{figure}

The standard definition of prefactorization algebras also comes with higher maps
\begin{equation}\label{higherproducts}
    m_{U_1,\ldots,U_n}^{W}:\, \mathfrak{F}(U_1) \otimes \cdots \otimes \mathfrak{F}(U_n) \longrightarrow \mathfrak{F}(W)
\end{equation}
for all (connected) open sets $U_1,\ldots,U_n \subseteq W$ with the $U_i$ pairwise disjoint. These maps are required to satisfy certain higher associativity relations. These relations imply that these maps are always determined by $m_{U,V}^W$ recursively as
\begin{equation}
  m_{U_1,\ldots,U_n}^{W} = m_{V,U_n}^W \circ (m_{U_1,\ldots,U_{n-1}}^{V} \otimes \text{id}_{\mathfrak{F}(U_n)})  
\end{equation}
for any open set $V$ containing $U_1,\ldots ,U_{n-1}$ and $V$ disjoint from $U_n$. For this reason, it is not necessary to give the higher maps \eqref{higherproducts} in the definition \cite{HomotopyPA}.

Further, note that we only defined $\mathfrak{F}$ on connected sets. The general definition of a prefactorization algebra requires $\mathfrak{F}$ to be defined on \emph{any} open set. The value of $\mathfrak{F}$ on generic open sets cannot be deduced from its value on connected open sets unless we require \emph{multiplicativity}. For this reason, we will always assume multiplicativity for all our examples. The implications of this assumption for the action of $\mathfrak{F}$ on generic open sets is discussed in appendix \ref{App:C}, where we explain how $\mathfrak{F}$ is determined on any open set from its value on the connected sets.

Note that for our above standard example  the spaces $\mathfrak{F}({U})$ have more structure 
since they are chain complexes. In general, one requires the structure maps  to be morphisms, 
and in the example of chain complexes these  morphisms must be chain maps, e.g., 
\begin{equation}\label{structuremapischain} 
    d \circ m_{{U}_1 {U}_2}^{{V}} = m_{{U}_1{U}_2}^{{V}} \circ d\;. 
\end{equation}
For the above example this follows since the differential obeys the Leibniz rule for disjoint intervals. 
Indeed, recalling the map $m_{I_1,I_2}^J:\mathfrak{F}(I_1)\otimes \mathfrak{F}(I_2)\rightarrow \mathfrak{F}(J)$ given by $m_{I_1,I_2}^J(f_1\otimes  f_2) = f_1 \wedge f_2\in \mathfrak{F}(J)$, we compute 
\begin{equation}
\begin{split} \label{m2isachainmap}
 (d\circ m_{I_1,I_2}^J) (f_1\otimes f_2) = d\big(f_1 \wedge f_2\big) &= d f_1 \wedge  f_2
 +(-1)^{f_1} f_1 \wedge d f_2\\
 &= m_{I_1,I_2}^J(df_1\otimes  f_2) + (-1)^{f_1} m^J_{I_1 I_2}(f_1\otimes  df_2) \\
 &= (m_{I_1,I_2}^J\circ d)(f_1\otimes  f_2)\;, 
\end{split} 
\end{equation} 
where we used the Leibniz rule (\ref{LEIBNIZ}). The right hand side involves the natural definition of a derivation $d$ on the tensor product $X \otimes Y$ of chain complexes $(X,d_X)$ and $(Y,d_Y)$, which is given by
\begin{equation}
d (x \otimes y) = d x \otimes y + (-)^{x} x \otimes d y \, .   
\end{equation}
The proof given in \eqref{m2isachainmap} is for the original differential $d$ but, 
as noted above, on disjoint intervals it also holds on the full $\delta_{\rm BV}$.

In this text we will exclusively work with prefactorization algebras without worrying whether our examples actually define factorization algebras. A \textit{factorization algebra}  is a special case  of a prefactorization algebra in much the same way as a cosheaf specializes a precosheaf \cite{Costello_Gwilliam_2016,costelloFactorization}. Even when considering prefactorization algebras, we will from now on call them factorization algebras for convenience.

The above discussion makes clear that assigning the symmetric algebra of quantum observables 
$\mathfrak{F}(I):= {\rm Obs}^q(I)$ to each open interval $I$ naturally yields a prefactorization algebra. 
Importantly, however, any conventional associative algebra $\pazocal{A}$ 
can also be viewed as a prefactorization algebra on the real line
\cite{Costello_Gwilliam_2016, costelloFactorization}. 
To this end one assigns to each open interval $I \subseteq \mathbb{R}$ the algebra $\pazocal{A}$: 
\begin{equation}
    \mathfrak{F}(I) = \pazocal{A}\;, 
\end{equation}
and takes the inclusion map $m_I^J$ for intervals  $I \subseteq J$ to be just  the identity: 
\begin{equation}
    m_I^J(a) = a \;\; \quad \forall a \in \pazocal{A}\;.
\end{equation}
The prefactorization product is the algebra product: 
\begin{equation}
    m_{I_1,I_2}^J(a, b) = a \cdot b \in \pazocal{A}\, , 
\end{equation} 
 so that the compatibility condition (\ref{CompaAss}) 
expresses associativity:
\begin{equation}
(a\cdot b)\cdot c = a\cdot (b\cdot c) \, .  
\end{equation} 

Since for this simple example all vector spaces are actually equal  
one  has for any  two open intervals $I$ and $J$ an isomorphism 
\begin{equation}\label{localconst}
  \mathfrak{F}(I) \cong \mathfrak{F}(J)\,,
\end{equation}
 that in fact is just  the identity. A prefactorization algebra satisfying the property given in \eqref{localconst} is called \textit{locally constant}. The prefactorization product maps, being only maps between different spaces $\mathfrak{F}(I)$, in general do not imply the existence of an algebra structure on the individual spaces $\mathfrak{F}(I)$. However, if a prefactorization algebra is locally constant on $\mathbb{R}$, then the following diagram
\begin{equation}
\begin{tikzcd}
\mathfrak{F}(I) \otimes \mathfrak{F}(I) \arrow[r,"\cong"] \arrow[d,"\mu"] & \mathfrak{F}(I) \otimes \mathfrak{F}(J) \arrow[d,"m_{I,J}^K"] \\
\mathfrak{F}(I) & \arrow[l,swap,"\cong"] \mathfrak{F}(K)
\end{tikzcd}
\end{equation}
induces an associative product $\mu$ on $\mathfrak{F}(I)$. Here, $J$ is such that $t < s$ for all $t \in I$ and $s \in J$. Therefore, any locally constant prefactorization algebra on $\mathbb{R}$ implies an algebra on a fixed space $\mathfrak{F}(I)$ (which can be taken to be $\mathfrak{F}(\mathbb{R})$).

In order to compare prefactorization algebras, we need to define morphisms between them. Recall that in the case of associative algebras, a morphism $R: (A,\mu_A) \rightarrow (B,\mu_B)$ is a linear map $R: A \rightarrow B$, such that
\begin{equation}
\begin{tikzcd}
A \otimes A \arrow[r,"\mu_A"] \arrow[d,"R \otimes R"] & A \arrow[d,"R"] \\
B \otimes B \arrow[r,"\mu_B"] & B
\end{tikzcd}
\end{equation}
commutes. A morphism of prefactorization algebras on a topological space $X$ is very similar, although it is given by a whole family of maps $R_U$, one for each open set.
\begin{defn}\label{def:FactAlgMor}
A morphism $R: \mathfrak{F} \rightarrow \mathfrak{G}$ of prefactorization algebras on $X$ is given by a family of chain maps $R_U: \mathfrak{F}(U) \rightarrow \mathfrak{G}(U)$ over the open sets $U \subseteq X$, such that
\begin{equation}
\begin{tikzcd}
    \mathfrak{F}(U) \arrow[r,"m_U^V"] \arrow[d,"R_U"] & \mathfrak{F}(V) \arrow[d,"R_V"] \\
    \mathfrak{G}(V) \arrow[r,"\bar m_U^V"] & \mathfrak{G}(V)
\end{tikzcd} \ , \qquad 
\begin{tikzcd}
    \mathfrak{F}(U_1) \otimes \mathfrak{F}(U_2) \arrow[rr,"m_{U_1,U_2}^V"] \arrow[d,"R_{U_1} \otimes R_{U_2}"] & &\mathfrak{F}(V) \arrow[d,"R_V"] \\
    \mathfrak{G}(V) \arrow[rr,"\bar m_{U_1,U_2}^V"] & & \mathfrak{G}(V)
\end{tikzcd}
\end{equation}
commute for all $U,U_1,U_2 \subseteq V$ with $U_1$ and $U_2$ disjoint. Here, the $m_U^V$, $m_{U_1,U_2}^V$ are the structure maps of $\mathfrak{F}$, while $\bar m_U^V$, $\bar m_{U_1,U_2}^V$ are the structure maps of $\mathfrak{G}$. 
$R$ is called an isomorphism (resp.~a quasi-isomorphism) if each $R_U$ is an isomorphism (resp.~a quasi-isomorphism).
\end{defn}

\subsection{Weyl algebra} 

Let us next introduce  the factorization algebra that encodes  the Weyl algebra of quantum mechanics and 
that will subsequently be shown to be quasi-isomorphic to the above `off-shell' factorization algebra based 
on the BV algebras. The Weyl algebra is the quantum mechanical algebra of polynomials in $p$ and $q$ that satisfy
\begin{equation}
    [q, p] = i\hbar\;, 
\end{equation}
or, upon the usual change of basis to annihilation and creation operators, 
 \begin{equation}\label{standaraadagger} 
    [a, a^\dagger] = 1\;. 
\end{equation}  
In the following the observables will be normal ordered polynomials in $a$ and $a^\dagger$. 

To this end let us define the Weyl algebra  as 
\begin{equation}
    {\rm Weyl} := \left({\rm Sym} (\mathbb{C}^2), \;\mu\right)\;,
\end{equation}
where we recall that $(a,a^\dagger)$ defines the basis of $\mathbb{C}^2$ we will use. Such a choice of basis allows us to identify ${\rm Sym}(\mathbb{C}^2)$ with $\mathbb{C}[a,a^\dagger]$, where the latter denotes the algebra of polynomials over two variables $a$ and $a\dagger$.
The map $\mu$ is an associative but non-commutative product 
 \begin{equation}
  \mu : \; \mathbb{C}[a,a^\dagger] \otimes  [a,a^\dagger] \longrightarrow [a,a^\dagger]\;, 
 \end{equation} 
satisfying 
\begin{equation}
\mu(\mu(a, b), c) = \mu(a, \mu(b, c))\;, 
\end{equation} 
that should  encode the commutation relations (\ref{standaraadagger}), 
\begin{equation}\label{aadaggercommmu} 
    \mu(a, a^\dagger) - \mu(a^\dagger, a) = 1\;. 
\end{equation}
It is important to recall  that on the symmetric algebra we always have 
\begin{equation}\label{Symmaadagger} 
    a^\dagger \wedge a = a \wedge a^\dagger\;. 
\end{equation}
The Weyl algebra we want to define should  be thought of in terms of  the normal ordered polynomials in $a^\dagger$ and $a$. The product $\mu$ takes two normal ordered polynomials, multiplies them, and commutes all $a^\dagger$s to the left so that the result is again normal ordered. 
We thus set 
\begin{equation}\label{aadaggerRULES} 
    \begin{aligned}
        \mu(a^\dagger, a) &= a^\dagger \wedge a\;, \\
        \mu(a, a^\dagger) &= a \wedge a^\dagger + 1 = a^\dagger \wedge a + 1\;, \\
        \mu(a,a) &= a\wedge a \;, \quad \mu(a^{\dagger}, a^{\dagger}) = a^\dagger \wedge a^\dagger \;, 
    \end{aligned}
\end{equation}
where the second equality follows with (\ref{Symmaadagger}), and 
the definition is so that it is compatible with (\ref{aadaggercommmu}). 
The extension to the full symmetric algebra is given by
\begin{equation}\label{Generalmu}
    \mu(F,G) = \wedge \circ 
    \exp\Big(\frac{\partial}{\partial a} \otimes \frac{\partial}{\partial a^{\dagger}} \Big)
    (F \otimes G)\, , 
\end{equation}
which  reproduces \eqref{aadaggerRULES}. Furthermore, one may convince oneself that for general $F$ and $G$ this formula computes the sum over all contractions between the normal ordered operator $F$ and the normal ordered operator $G$, which is exactly what one has to do when one wants  to bring expressions like
$(a^\dagger)^k a^l (a^\dagger)^m a^n$
into normal ordered form. 

We will now show that $\mu$ is associative, so that $(\mathbb{C}[a,a^\dagger],\mu)$ defines an associative algebra, which can be thought of as a deformation of the commutative wedge product.
\begin{prop}
Let $\mathrm{Sym}(\mathbb{C}^2) = \mathbb{C}[x,y]$ be the commutative algebra of polynomials in $x, y$ 
with product $\wedge$. Then
\begin{equation}
    \mu = \wedge \circ e^{t(\partial_x \otimes \partial_{y})}
\end{equation}
defines an associative deformation of $\wedge$ for all $t$.
\end{prop}
\begin{proof}
We have
\begin{equation}\label{muass}
    \mu \circ (\mu \otimes 1) = \wedge \circ \exp((t \partial_{x} \otimes \partial_y)) \circ (\wedge \otimes 1) \circ (\exp(t (\partial_{x} \otimes \partial_y)) \otimes 1) \, .
\end{equation}
Since $\partial_x$ is a derivation with respect to $\wedge$, we have
\begin{equation}
(\partial_{x} \otimes \partial_y) \circ (\wedge \otimes 1) =  (\wedge \otimes 1) \circ (\partial_x \otimes 1 \otimes \partial_y + 1 \otimes \partial_x \otimes \partial_y)\,.
\end{equation}
Computing the exponential on both sides shows that
\begin{equation}
\exp(t(\partial_{x} \otimes \partial_y)) \circ (\wedge \otimes 1) = (\wedge \otimes 1) \circ \exp(t(\partial_x \otimes 1 \otimes \partial_y + 1 \otimes \partial_x \otimes \partial_y)) \, .
\end{equation}
Using this, \eqref{muass} becomes
\begin{equation}\label{muass2}
    \mu \circ (\mu \otimes 1) = \wedge_3 \circ \exp(t(\partial_x \otimes 1 \otimes \partial_y + 1 \otimes \partial_x \otimes \partial_y)) \circ (\exp(t (\partial_{x} \otimes \partial_y)) \otimes 1) \, ,
\end{equation}
where $\wedge_3 = \wedge \circ (\wedge \otimes 1) = \wedge \circ (1 \otimes \wedge)$. Note also that
\begin{equation}
(\exp(t (\partial_{x} \otimes \partial_y)) \otimes 1) = \exp(t (\partial_{x} \otimes \partial_y \otimes 1)) \, .
\end{equation}
Finally, since $\partial_x \otimes \partial_y \otimes 1$ commutes with both $\partial_x \otimes 1 \otimes \partial_y$ and $1 \otimes \partial_x \otimes \partial_y$, \eqref{muass2} becomes
\begin{equation}
\mu \circ (\mu \otimes 1) = \wedge_3 \circ \exp(t(\partial_x \otimes \partial_y \otimes 1 +\partial_x \otimes 1 \otimes \partial_y + 1 \otimes \partial_x \otimes \partial_y)) \, .
\end{equation}
A similar computation shows that $\mu \circ (1 \otimes \mu)$ gives the same expression. Therefore, $\mu \circ (\mu \otimes 1) = \mu \circ (1 \otimes \mu )$ and so $\mu$ is associative.
\end{proof}
\noindent The proposition implies that $\mu$ defined in \eqref{Generalmu} is associative.

As previously shown for a general associative algebra, we can make 
the Weyl algebra  into a prefactorization algebra on the real line, 
where to each open interval $I$ one assigns the Weyl algebra: $\mathfrak{F}(I)=\mathbb{C}[a,a^\dagger]$. 
The only subtlety is that because one wants to interpret this real line as time one has to set, 
for any two polynomials $a,b \in \mathbb{C}[a,a^\dagger]$, 
\begin{equation}\label{Defmmu}  
     m_{\mu}(a,b) := 
   \begin{cases}
    \mu(b, a) & \text{if $I_1<I_2$}\\
    \mu(a,b) &\text{if $I_1> I_2$}
  \end{cases} \  \in  \ \mathbb{C}[a,a^\dagger]\;, 
\end{equation}
for disjoint intervals $I_1,  I_2   \subset J$, where the ordering of intervals 
is defined as 
  \begin{equation}
  I_1 < I_2 \quad \Leftrightarrow \quad \forall\,  t\in I_1, \, s\in I_2:\; t< s \;, 
 \end{equation} 
and similarly for $I_1>I_2$. 

When expressing an associative algebra as a factorization algebra on $\mathbb{R}$, the order of multiplication, as for example in the expression $(ab)c$, is encoded in a sequence of embeddings of open intervals, see figure \ref{fig:threeproduct} for an explicit example in the context of the Weyl algebra.
\begin{figure}[H]
    \centering
    \includegraphics[width=0.65\textwidth]{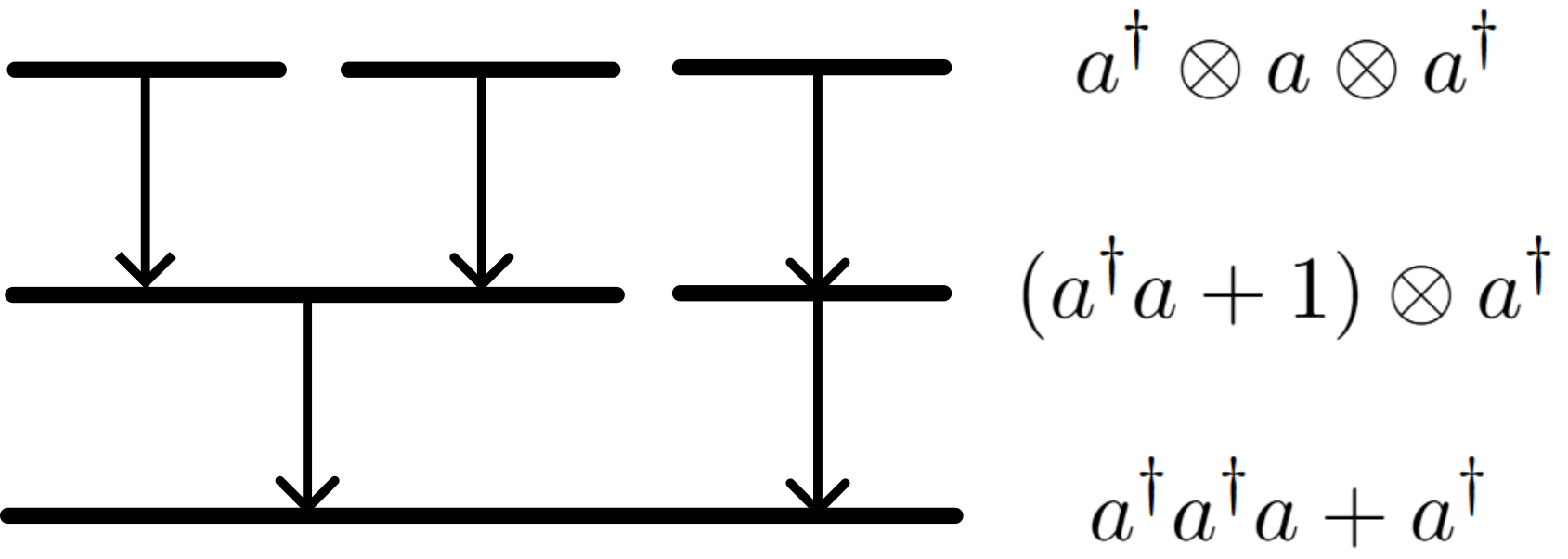}
    \caption{Prefactorization structure of the Weyl algebra. In the first step, the first two intervals are mapped $m_\mu(a^\dagger, a) = \mu(a, a^\dagger) = a^\dagger a + 1$, while the rightmost interval is mapped identically. In the second step the prefactorization product maps $m_\mu(a^\dagger a, a^\dagger) = \mu(a^\dagger, a^\dagger a) = a^\dagger a^\dagger a$ and $\mu(1, a^\dagger) = a^\dagger$.}  \label{fig:threeproduct}
\end{figure}

\subsection{Quasi-isomorphism} 

Our next goal is to relate the two prefactorization algebras discussed above: the one based on the 
`off-shell' BV algebra and the one based on the `on-shell' Weyl algebra. More precisely, we will show that 
these two prefactorization algebras are quasi-isomorphic and hence physically equivalent. 
Above we have introduced the notion of quasi-isomorphism for chain complexes: two chain complexes are 
quasi-isomorphic if there is a chain map (a degree zero map that commutes with the differential) between them
that preserves the cohomology. Here we need the  generalization of this notion to prefactorization algebras where to the 
open sets one assigns chain complexes. (Of course, any vector space can be viewed as a chain complex with 
trivial differential.) Let us recall that in this case one demands that 
the structure maps are chain maps, c.f.~(\ref{structuremapischain}), also with respect to $\delta_{\rm BV}$. 

In the previous section we have established a chain map that is in fact a quasi-isomorphism between 
the individual spaces of the two prefactrorization algebras, which for definiteness we call $\mathfrak{F}_{\rm BV}$ 
and $\mathfrak{F}_{\rm Weyl}$. For each open interval $I$ we thus have  
\begin{equation}\label{factorrizationArrow} 
    \xymatrix{
    & \mathfrak{F}_{\rm BV}(I)= {\rm Obs}^q(I) \ar[d]^\Pi\\
    & \mathfrak{F}_{\rm Weyl}(I) = \mathbb{C}[a,a^\dagger]
    }
\end{equation}
where we recall that the projection $\Pi$ is a chain map thanks to the improved formula 
(\ref{improvedPI}), 
\begin{equation}\label{improvedPI2} 
    \Pi := \Pi_0 \, e^{-C}\;, 
 \end{equation} 
with $\Pi_0$ the naive projection acting component-wise and  $C= \Delta \circ (1 \otimes h)$. 
This means that $\Pi\circ \delta_{\rm BV}=0$ as the differential on the Weyl algebra is trivial.

The above defines a quasi-isomorphism between chain complexes
but as emphasized before it does \textit{not} respect the algebra structure of each individual space. 
Rather, 
we will show that $\Pi$ defines a quasi-isomorphism for the corresponding factorization algebras. According to definition \ref{def:FactAlgMor}, we need to show that
\begin{equation}\label{FactAlMor}
\begin{tikzcd}
\mathfrak{F}_{\rm{BV}}(I)\arrow[d,"\Pi"] \arrow[r,"m_I^J"] & \mathfrak{F}_{\rm{BV}}(I)  \arrow[d,"\Pi"]\\
\mathfrak{F}_{\rm Weyl}(I) \arrow[r,"\text{id}"] & \mathfrak{F}_{\rm Weyl}(J) 
\end{tikzcd}
\ , \quad
\begin{tikzcd}
\mathfrak{F}_{\rm BV}(I_1) \otimes \mathfrak{F}_{\rm BV}(I_2) \arrow[r, "m_{I_1,I_2}^{J}"] \arrow[d,"\Pi \otimes \Pi"] & \mathfrak{F}_{\rm BV}(J) \arrow[d,"\Pi"] \\
\mathfrak{F}_{\rm Weyl}(I_1) \otimes \mathfrak{F}_{\rm Weyl}(I_2) \arrow[r,"m_\mu"] & \mathfrak{F}_{\rm Weyl}(J) \, .
\end{tikzcd}
\end{equation}
This then would establish that $\mathfrak{F}_{\rm BV}$ and $\mathfrak{F}_{\rm Weyl}$ are quasi-isomorphic, since $\Pi$ is also a quasi-isomorphism of chain complexes $\Pi: \mathfrak{F}_{\rm BV}(I) \rightarrow \mathfrak{F}_{\rm Weyl}(I)$ for all $I$. We universally call the morphism $\Pi$, although strictly speaking it should be thought of as a family of morphisms, one for each open interval $I$.

As we have defined above a multitude of spaces and corresponding $\Pi$ maps let us pause here 
for a second and review what we have. On a single compactly supported function $f$ in the dual chain complex, 
$\Pi$ acts via (\ref{eq:projector_classical}), and on a word or string of such objects (i.e.~a generic 
element of the BV algebra) $\Pi$ acts via (\ref{improvedPI2}), where $\Pi_0$ acts like a 
morphism in (\ref{Piviamorph}). In particular, due to the insertion of $e^{-C}$ in (\ref{improvedPI2}) 
the corrected $\Pi$ does \textit{not} act as a morphism. 
This is no contradiction, since we need to prove commutativity of the second diagram in \eqref{FactAlMor} with respect to the product $m_\mu$, and not with respect to the wedge product on $\mathfrak{F}_{\rm {Weyl}}(I) = \mathbb{C}[a,a^\dagger]$.

Let us then turn to the verification that the above diagram commutes, which is the statement that 
 \begin{equation}\label{commutativitydiagram} 
  \Pi\circ m_{12}  =  m_{\mu}\circ \Pi\;, 
 \end{equation} 
where we use the abbreviation $m_{12}=m_{I_1,I_2}^J$.  
We will first verify this relation acting on $f_1\otimes f_2 \in \mathfrak{F}(I_1)\otimes \mathfrak{F}(I_2)$, 
where $f_1$, $f_2$ are functions (linear monomials) with support on intervals with $I_1 < I_2$. (The other case immediately follows, since they  are related via $m_{I_1,I_2}^J(f_1,f_2) = m_{I_2,I_1}^J(f_2,f_1)$). 
The right-hand side of (\ref{commutativitydiagram}) now yields 
  \begin{equation}\label{RighthandSIDEcommcomp} 
   \begin{split}
    (m_{\mu}\circ \Pi)(f_1\otimes f_2) &= m_{\mu}(\Pi(f_1)\otimes \Pi(f_2))  \\
    &=\mu(\Pi(f_2), \Pi(f_1)) \\
    &= \mu\Big(f_{2-}a +f_{2+}a^\dagger\,, \; f_{1-}a+ f_{1+}a^\dagger\Big) \\
    &= f_{2-}f_{1-} \,a\wedge a + f_{2+} f_{1+} \,a^{\dagger} \wedge a^{\dagger} \\
    &\quad +\big(f_{2-} f_{1+} + f_{2+} f_{1-}\big) a^\dagger \wedge a + f_{2-} f_{1+} \;, 
   \end{split} 
  \end{equation} 
 where we used the definition (\ref{Defmmu})  of $m_{\mu}$, the $\Pi$ action 
 on a  function, $\Pi(f)=f_- a + f_+ a^\dagger$, c.f.~(\ref{eq:projector_classical}), 
 and the definition (\ref{aadaggerRULES}) of the $\mu$ product. 
 This we need to compare to the left-hand side of (\ref{commutativitydiagram}), 
 for which we compute 
   \begin{equation}\label{LefthandSIDEcommcomp} 
    \begin{split}
      ( \Pi\circ m_{12})(f_1\otimes f_2) &= \Pi(f_1\wedge f_2) 
      = \Pi_0 \big(1-C\big) (f_1\wedge f_2) \\
      &= \Pi_0(f_1) \wedge \Pi_0(f_2) -C(f_1\wedge f_2) \\
      &= \big(f_{1-} a + f_{1+} a^\dagger\big) \wedge \big(f_{2-} a + f_{2+} a^\dagger\big)
      -C(f_1\wedge f_2)
      \;, 
    \end{split} 
  \end{equation} 
where we used (\ref{improvedPI2}) on quadratic monomials and the $\Pi$ 
action recalled above.  
Multiplying out the quadratic terms in the last line of (\ref{LefthandSIDEcommcomp}) and subtracting 
from this (\ref{RighthandSIDEcommcomp}) we infer 
 \begin{equation}\label{intermediateRESULT} 
  ( \Pi\circ m_{12}-m_{\mu}\circ \Pi)(f_1\otimes f_2) =  -C(f_1\wedge f_2)  - f_{2-} f_{1+}  \;. 
 \end{equation} 
In order to show that the right-hand side is actually zero, thereby completing the proof, 
we compute, recalling $C= \Delta \circ (1 \otimes h)$, 
 \begin{equation}\label{CCOMPUUU} 
 \begin{split}
  C(f_1\wedge f_2) &= \Delta\big(f_1\wedge 
  h(f_2) \big) \\
  &=i\hbar \int_J dt\, f_1(t) \, h (f_2) (t) \\
  &= -\frac{\hbar}{2\omega} \int_{I_1} dt \, f_1(t) \int_{I_2} ds \, f_2(s) \, e^{-i\omega(s-t)} \\
  &=  -\frac{\hbar}{2\omega} \Big(\int_{I_1} dt \, f_1(t)e^{i\omega t }\Big) 
  \Big(\int_{I_2} ds \, f_2(s)e^{-i\omega s }\Big) \\
  &= -\langle f_1, \sigma_-\rangle \langle f_2, \sigma_+\rangle \\
  &= - f_{1+} f_{2-} \;. 
 \end{split} 
 \end{equation} 
Here we used in the third line the definition (\ref{eq:h}) of $h$,  together with $|t-s|=s-t$ on $I_1 < I_2$ and we used 
$e^{\pm i\omega t}=\sqrt{\frac{2\omega}{\hbar}} \sigma_{\mp}(t)$, c.f.~(\ref{sigmaPLUSM}).  
Using (\ref{CCOMPUUU}) back in (\ref{intermediateRESULT}) 
we infer that the terms on the right-hand side cancel, thus completing the proof 
of the compatibility property (\ref{commutativitydiagram}) for $I_1< I_2$.

For later use let us record the final result: 
 \begin{equation}
 C\big(f_1\wedge f_2\big)=
  \begin{cases}
    -f_{1+} f_{2-} & \text{if $I_1<I_2$}\\
    -f_{1-} f_{2+} &\text{if $I_1> I_2$}
  \end{cases}\;, 
 \end{equation}
where $I_1$ denotes the interval on which $f_1$ is 
compactly supported and similarly for $I_2$. 
Note that as a consequence of this definition we  have $C\big(f_1\wedge f_2\big)=C\big(f_2\wedge f_1\big)$, as it should be in order to be compatible with the symmetry of the wedge product. 
Let us also record the 
action of $C$ on the functions $f_{a}$ and $f_{a^{\dagger}}$ defined above, for which 
we have due to (\ref{fafadaggercomp})  
 \begin{equation}
  f_{a,-}=f_{a^{\dagger},+} =1\,, \qquad f_{a,+}=f_{a^{\dagger},-}= 0\,. 
 \end{equation}  
Concretely, taking 
$f_a$ to be compactly supported in $I_a$ and $f_{a^{\dagger}}$ to be compactly supported in $I_{a^{\dagger}}$, where $I_a$ and $I_{a^{\dagger}}$ have disjoint support, we then have \begin{equation}\label{Caadagger}
 C\big(f_a\wedge f_{a^{\dagger}}\big)=
 C\big(f_{a^{\dagger}}
 \wedge f_a\big) \ = \ 
  \begin{cases}
    \; \; 0 & \text{if $I_a<I_{a^{\dagger}}$}\\
    -1 &\text{if $I_a> I_{a^{\dagger}}$}
  \end{cases}\;. 
 \end{equation}
 In words, $C\big(f_a\wedge f_{a^{\dagger}}\big)=0$  if $f_a$ is associated to earlier times, while  $C\big(f_a\wedge f_{a^{\dagger}}\big)=-1$ if $f_a$ is associated to later times.

Finally we point out the following:  
Given that the factorization algebra encoding off-shell quantum mechanics is based on the \textit{symmetric} algebra, there can be no Heisenberg-type canonical commutation relation at this level. 
Rather, signs of the latter only emerge upon projecting down 
to on-shell quantum mechanics by means of  
the projector $\Pi=\Pi_0e^{-C}$. To see this fix two disjoint intervals $I_1$ and $I_2$ and fix representatives $f_a^1$, $f_{a^{\dagger}}^1$ that are compactly supported on $I_1$ and similarly $f_a^2$, $f_{a^{\dagger}}^2$ that are compactly supported on $I_2$, so that $\Pi(f_a^1)=\Pi(f_a^2)=a$
and $\Pi(f_{a^{\dagger}}^1)=\Pi(f_{a^{\dagger}}^2)=a^{\dagger}$. Using (\ref{Caadagger}) one finds:
 \begin{equation}
 \Pi\big(f_a^1\wedge f_{a^{\dagger}}^2 - 
 f_{a^{\dagger}}^1\wedge f_{a}^2\big)
  \ = \ 
  \begin{cases}
    -1 & \text{if $I_1<I_2$}\\
    \;\; \; 1 &\text{if $I_1> I_2$}
  \end{cases}\;. 
 \end{equation}
Note that while in this computation the terms $a\wedge a^{\dagger}-a^{\dagger}\wedge a$ are generated, they are of course zero in the symmetric algebra. Rather, all `non-commutativity' is now related to the order of intervals $I_1, I_2$.

\subsection{Proof }

After giving the quasi-isomorphism we are now ready to state and prove the following general theorem:  
\begin{thm}\label{thm:reals}
The family of maps $\Pi = \Pi_0 \circ e^{-C}$ over the open intervals in $\mathbb{R}$ define a quasi-isomorphism of factorization algebras $\Pi: \mathfrak{F}_{\rm BV} \rightarrow \mathfrak{F}_{\rm Weyl}$.
\end{thm}
We devide the proof  into two steps: 
\begin{prop}
The factorization product $\bar{m}$ defined via the diagram
\begin{equation}\label{Deformedmmu}
    \begin{tikzcd}
    \mathfrak{F}_{\rm BV}(I_1) \otimes \mathfrak{F}_{\rm BV}(I_2) \arrow[rr,"m_{I_1,I_2}^J"] & & \mathfrak{F}_{\rm BV}(J) \arrow[d, "e^{-C}"] \\
    \mathfrak{F}_{\rm cBV}(I_1) \otimes \mathfrak{F}_{\rm cBV}(I_2) \arrow[u,swap,"e^C \otimes e^C"] \arrow[rr,"\bar m_{I_1,I_2}^J"] & &\mathfrak{F}_{\rm cBV}(J) \, .
    \end{tikzcd}
\end{equation}
gives $\mathfrak{F}_{\rm cBV}$ the structure of a prefactorization algebra. The map $e^{-C}: \mathfrak{F}_{\rm BV}(I) \rightarrow \mathfrak{F}_{\rm cBV}(I)$ defines an isomorphism of prefactorization algebras.
\end{prop}
\begin{proof}
We begin with the following observation. Given a prefactorization algebra $\mathfrak{F}$ together with a family of chain complexes $C_U$ for each open connected set $U$, such that there is an isomorphism $R_U: \mathfrak{F}(U) \rightarrow C_U$ (meaning that $R_U$ is an invertible chain map), then $C_U$ can be given the structure of a prefactorization algebra $\mathfrak{C}$ isomorphic to $\mathfrak{F}$. This is achieved by defining
\begin{equation}
    \bar{m}_U^V = R_V \circ m_U^V \circ R^{-1}_U \, , \qquad \bar{m}_{U_1,U_2}^V = R_V \circ m_{U_1,U_2}^V \circ (R_{U_1}^{-1} \otimes R_{U_2}^{-1}) \, .
\end{equation}
The maps $\bar{m}_U^V$ and $\bar{m}_{U_1,U_2}^V$ then give the family $C_U$ the structure of a prefactorization algebra $\mathfrak{C}$ and $R_U: \mathfrak{F}(U) \rightarrow \mathfrak{C}(U)$, where $\mathfrak{C}(U) = C_U$ defines an isomorphism of prefactorization algebras.

Recall that $R_I = e^{-C}$ defines an isomorphism $R_I: \mathfrak{F}_{\rm BV}(I) \rightarrow \text{Obs}^{cl}(I)$. Therefore, $\text{Obs}^{cl}(I)$ combine into a prefactorization algebra $\mathfrak{F}_{\rm cBV}$ with inclusion maps
\begin{equation}
\bar m_{I}^J = R_J \circ m_I^J \circ R_I^{-1} = m_I^J  
\end{equation}
and factorization product as defined in \eqref{Deformedmmu}.
\end{proof}

Our next step is to show that $\mathfrak{F}_{\rm cBV}$ and $\mathfrak{F}_{\rm Weyl}$ are quasi-isomorphic via $\Pi_0$. We already know that $\Pi_0: \mathfrak{F}_{\rm cBV} \rightarrow \mathfrak{F}_{\rm Weyl}$ defines a quasi-isomorphism on the level of chain complexes. Before proving that it defines a quasi-isomorphism, we will first simplify the formula for $\bar m_{I_1,I_2}^J$. 

\begin{prop}
The factorization product $\bar m_{I_1,I_2}^J: \mathfrak{F}_{\rm cBV}(I_1) \otimes \mathfrak{F}_{\rm cBV}(I_2) \rightarrow \mathfrak{F}_{\rm cBV}(J)$ can be written as
\begin{equation}
\bar m_{I_1,I_2}^J = \wedge \circ e^{-C_2} \circ (m_{I_1}^J \otimes m_{I_2}^J) \, ,
\end{equation}
where $C_2: \mathfrak{F}_{\rm cBV}(J) \otimes \mathfrak{F}_{\rm cBV}(J) \rightarrow \mathfrak{F}_{\rm cBV}(J) \otimes \mathfrak{F}_{\rm cBV}(J)$ is defined as
\begin{equation}\label{C2action}
    C_2(f_1\cdots f_m \otimes g_1 \cdots g_n) = \sum_{i = 1}^m \sum_{j = 1}^n C(f_i,g_j) f_1\cdots \hat f_i \cdots f_m \otimes g_1 \cdots \hat g_i \cdots g_n \, .
\end{equation}
where
\begin{equation}
    C(f,g) = -\frac{\hbar}{2 \omega}\int d t d s \, f(t)e^{-i|t-s|}g(s)
\end{equation}
for $f$ and $g$ in degree zero, and zero otherwise.
\end{prop}
\begin{proof}
Let $C$ be the second order operator defining $e^{-C}$. It is straightforward to show that it satisfies
\begin{equation}\label{Bialgebra}
    C \circ \wedge = \wedge \circ (C \otimes 1 + 1 \otimes C + C_2) \, ,
\end{equation}
where $C_2$ is given in \eqref{C2action}. Furthermore, we also have 
\begin{equation}\label{OperatorsCommute}
[C_2,C \otimes 1] = [C_2,1 \otimes C] = [C \otimes 1, 1 \otimes C] = 0 
\end{equation}
as linear maps on $\mathfrak{F}_{\rm cBV}(J) \otimes \mathfrak{F}_{\rm cBV}(J)$. From these two properties we deduce that
\begin{equation}
\begin{split}
\exp(-C) \circ \wedge &= \wedge \circ \exp(-(C \otimes 1 + 1 \otimes C + C_2)) \\
&= \wedge \circ \exp(-C_2) \circ \exp(-(C \otimes 1)) \circ \exp(-(1 \otimes C))  \\
&= \wedge \circ \exp(-C_2) \circ (\exp(-C) \otimes 1) \circ (1 \otimes\exp(-C))
\\
&= \wedge \circ \exp(-C_2) \circ (\exp(-C) \otimes \exp(-C)) \, .
\end{split}
\end{equation}
The first equation follows from \eqref{Bialgebra}, while the second equation is due to \eqref{OperatorsCommute}. We can use this in the definition of the $\bar  m^J_{I_1,I_2}$:
\begin{equation}
\begin{split}
    \bar m^J_{I_1,I_2} &= e^{-C} \circ m^J_{I_1,I_2} \circ (e^{-C} \otimes e^{-C}) = e^{-C} \circ \wedge \circ (m^J_{I_1} \otimes m^J_{I_2}) \circ (e^{C} \otimes e^{C}) \\
    &= e^{-C} \circ \wedge  \circ (e^{C} \otimes e^{C}) \circ (m^J_{I_1} \otimes m^J_{I_2}) = \wedge \circ e^{-C_2} \circ (m^J_{I_1} \otimes m^J_{I_2}) \, .
\end{split}
\end{equation}
This is what we had to show.
\end{proof}

On $\mathbb{C}[a,a^\dagger]$, the product $\mu$ of the Weyl algebra takes a very similar form, see \eqref{Generalmu}. There, instead of $C_2$, we take the exponential of the operator
\begin{equation}
    -\partial_{a} \otimes \partial_{a^\dagger}: \mathbb{C}[a,a^\dagger] \otimes \mathbb{C}[a,a^\dagger] \longrightarrow \mathbb{C}[a,a^\dagger] \otimes \mathbb{C}[a,a^\dagger] \, 
\end{equation}
before applying the wedge product. We argued that $\mu$ defines a deformation of the commutative product $\text{Sym}(\mathbb{C}^2)$. In a similar way, we can think of the factorization product $\bar m_{I_1,I_2}^J$ as a deformation of the product $m_{I_1,I_2}^J$ on $\mathfrak F_{\rm {cBV}}(I)$.

Recall that we want to establish that $\Pi$ defines a quasi-isomorphism from $\mathfrak{F}_{\rm BV}$ to $\mathfrak{F}_{\rm Weyl}$. Since $\mathfrak{F}_{\rm BV}$ and $\mathfrak{F}_{\rm cBV}$ are isomorphic via $e^{-C}$, it remains to prove that $\mathfrak{F}_{\rm cBV}$ is quasi-isomorphic to $\mathfrak{F}_{\rm Weyl}$ with quasi-isomorphism $\Pi_0$.
\begin{prop}
The family of maps $\Pi_0: \mathfrak{F}_{\rm cBV}(I) \rightarrow \mathfrak{F}_{\rm Weyl}(I)$ defines a quasi-isomorphism.
\end{prop}
\begin{proof}
We need to prove commutativity of the diagrams
\begin{equation}\label{DeformQI}
    \begin{tikzcd}
    \mathfrak{F}_{\rm cBV}(I) \arrow[r,"\bar m_{I}^J"] \arrow[d," \Pi_0"] & \mathfrak{F}_{\rm cBV}(J) \arrow[d,"\Pi_0"]\\
    \mathfrak{F}_{\rm Weyl}(I) \arrow[r,"\text{id}"] &  \mathfrak{F}_{\rm Weyl}(J)
    \end{tikzcd}
    \ , \qquad 
    \begin{tikzcd}
    \mathfrak{F}_{\rm cBV}(I_1) \otimes \mathfrak{F}_{\rm cBV}(I_2) \arrow[rr,"\bar m_{I_1{,}I_2}^J"] \arrow[d,"\Pi_0 \otimes \Pi_0"] & & \mathfrak{F}_{\rm cBV}(J) \arrow[d,"\Pi_0"]\\
    \mathfrak{F}_{\rm Weyl}(I_1) \otimes \mathfrak{F}_{\rm Weyl}(I_2) \arrow[rr,"m_\mu"] & &  \mathfrak{F}_{\rm Weyl}(J)
    \end{tikzcd}
\end{equation} 
Since $\bar m_I^{J} = m_I^J$ is just the subspace inclusion $C^\infty_c(I) \rightarrow C^\infty_c(J)$, commutativity of the first diagram is trivial. For the second diagram, we assume without loss of generality that $I_2 > I_1$. The other case follows from $m_{I_1,I_2}^J(f,g) = m_{I_1,I_2}^J(g,f)$.

The strategy to prove commutativity of the second diagram in \eqref{DeformQI} is to first prove it for the case when we replace $\Pi_0 \otimes \Pi_0$ with a quasi-inverse $\iota_{I_1} \otimes \iota_{I_2}$. This quasi-inverse is defined by associating to $G(a,a^\dagger) \in \mathfrak{F}_{\rm Weyl}(I_2)$ the object
\begin{equation}
\iota_{I_2}(G(a,a^\dagger)) = G(f_a,f_{a^\dagger}) \in \text{Sym}(C^\infty_c(I_2)) \, ,
\end{equation} 
where $f_a$ and $f_a^\dagger$ are compactly supported in $I_2$ with normalization conditions given in \eqref{fafadaggercomp}. Explicitly, if $G$ is of the form
\begin{equation}
    G(a,a^\dagger) = \lambda (a^\dagger)^m a^n
\end{equation}
for some $\lambda \in \mathbb{C}$, $G(f_a,f_{a^\dagger})$ is given by
\begin{equation}
    G(f_a,f_{a^\dagger}) = \lambda \underbrace{f_{a^\dagger} \wedge \cdots \wedge f_{a^\dagger}}_\text{$m$ times} \wedge \underbrace{f_{a} \wedge \cdots \wedge f_{a}}_\text{$n$ times} \, .
\end{equation}
Similarly, given $F(a,a^\dagger) \in \mathfrak{F}_{\rm Weyl}(I_1)$, we define 
\begin{equation}
\iota_{I_1}(F(a,a^\dagger)) = F(g_a,g_{a^{\dagger}}) \, ,  
\end{equation}
where $g$ and $g_{a^{\dagger}}$ satisfy the conditions in \eqref{fafadaggercomp} with compact support in $I_1$. 
We have
\begin{equation}
   (\Pi_0 \otimes \Pi_0) (F(g_a,g_{a^{\dagger}}) \otimes G(f_a,f_{a^\dagger})) = F(a, a^\dagger) \otimes G(a, a^\dagger) \, .
\end{equation}

We now want to show that
\begin{equation}\label{mmufrommbar}
    \begin{tikzcd}
    \mathfrak{F}_{\rm cBV}(I_1) \otimes \mathfrak{F}_{\rm cBV}(I_2) \arrow[rr,"\bar m_{I_1{,}I_2}^J"] & & \mathfrak{F}_{\rm cBV}(J) \arrow[d,"\Pi_0"]\\
    \mathfrak{F}_{\rm Weyl}(I_1) \otimes \mathfrak{F}_{\rm Weyl}(I_2) \arrow[rr,"m_\mu"] \arrow[u,swap,"\iota_{I_1} \otimes \iota_{I_2}"] & &  \mathfrak{F}_{\rm Weyl}(J)
    \end{tikzcd}
\end{equation}
commutes. Note that we have
\begin{equation}
C(g_a,f_{a^\dagger}) = -1 \, , \qquad C(g_{a^\dagger},f_a) = C(g_a,f_a) = C(g_{a^\dagger},f_{a^\dagger}) = 0 \, .
\end{equation}
From this together with \eqref{C2action} it follows that $C_2$ acts on $F(g_a,g_{a^{\dagger}}) \otimes G(f_a,f_{a^\dagger})$ as $-\frac{\partial}{\partial g_a} \otimes \frac{\partial}{\partial f_{a^\dagger}}$. Applying $\Pi_0 \circ \wedge$ then implies \eqref{mmufrommbar}.

We use this to show \eqref{DeformQI}, which is non-trivial only in degree zero. Note that $I_{\iota_1}$ and $I_{\iota_2}$ are quasi-isomorphisms. Therefore, any degree zero $F \in \mathfrak{F}_{\rm cBV}(I_1)$ is of the form $F = \iota_{I_1} \circ \Pi_0(F) + d F'$. Similarly, any $G \in \mathfrak{F}_{\rm cBV}(I_2)$ is of the form $G = \iota_{I_2} \circ \Pi_0(G) + d G'$. This implies that
\begin{equation}
    F \otimes G = \iota_{I_1} \circ \Pi_0(F) \otimes \iota_{I_2} \circ \Pi_0(G) + d H \, ,
\end{equation}
where
\begin{equation}
    H = (F' \otimes  (dG' + \iota_{I_2} \circ \Pi_0(G)) + (\iota_{I_1} \circ \Pi_0(F) + d F') \otimes G') \, .
\end{equation}
With this we can compute
\begin{equation}\label{OkuptoExact}
\begin{split}
m_\mu \circ (\Pi_0(F) \otimes \Pi_0(G)) &= \Pi_0 \circ \bar m_{I_1,I_2}^J \circ (\iota_{I_1}\Pi_0(F) \otimes \iota_{I_2}\Pi_0(G)) \\
&= \Pi_0 \circ \bar m_{I_1,I_2}^J(F\otimes G) - \Pi_0 \circ \bar m_{I_1,I_2}^J(dH) \, .
\end{split}
\end{equation}
Since $\bar m_{I_1,I_2}^J$ is a chain map with respect to $d$ and $\Pi_0 \circ d = 0$, the second term in \eqref{OkuptoExact} drops. The remaining terms then tell us that
\begin{equation}
m_\mu (\Pi_0 \otimes \Pi_0) = \Pi_0 \circ \bar m_{I_1,I_2}^J
\end{equation}
also in degree zero. Therefore, the second diagram in $\eqref{DeformQI}$ is also commutative and so $\Pi_0: \mathfrak{F}_{\rm cBV} \rightarrow \mathfrak{F}_{\rm Weyl}$ is a quasi-isomorphism of prefactorization algebras.
\end{proof}

\section{Hilbert spaces associated to intervals with boundary}

In the previous section we have shown that the factorization algebra $\mathfrak{F}_{\rm BV}$ based on the BV algebra 
of quantum observables is quasi-isomorphic to the factorization algebra $\mathfrak{F}_{\rm Weyl}$ encoding 
the operator algebra of raising and lowering operators of the harmonic oscillator. 
In order to have a full-fledged `off-shell' version of the quantum mechanics of the harmonic oscillator it 
remains to encode the space of states that enables the  computation of  quantum expectation values. 
In this section we will  encode quantum mechanics as a factorization algebra on the closed interval 
$[t_i,t_f]$, for which we consider open subsets of three types:\footnote{This construction is motivated by 
the observation of Costello and Gwilliam that viewing an associative algebra $\mathcal{A}$ as a prefactorization algebra on open sets, the boundary corresponds to an $\mathcal{A}$-module --- a vector space on 
which multiplication with elements of $\mathcal{A}$ from the right (or left) is defined.}
 \textit{i)} open intervals $(a,b)$  to which the 
factorization algebra assigns the operator algebra as above;  \textit{ii)} half-open intervals $[t_i, a)$ or 
$(b, t_f]$ to which the factorization algebra assigns the space of states or its dual; \textit{iii)} the closed interval 
$[t_i,t_f]$ to which the factorization algebra assigns complex numbers $\mathbb{C}$ giving 
the `amplitude'.\footnote{Note that as subsets of $[t_i,t_f]$ all these sets are indeed open, since the whole topological space is always both open and closed by definition.}
Moreover, we will present an off-shell factorization algebra that is quasi-isomorphic to this, 
using versions of the  chain complex of the harmonic oscillator with different boundary 
conditions.

\subsection{More general intervals and states }

Our goal is to define a BV algebra corresponding on connected open subsets of the closed interval $[t_i,t_f]$. This includes the half open intervals $[t_i,a), (b,t_f]$ as well as the interval $[t_i,t_f]$ itself.

We could define the complex to consist of compactly supported functions on the connected open sets, just like when working on $\mathbb{R}$. It turns out that this is not exactly the correct thing to do. Rather, we need to put some boundary conditions on the end points of the closed interval $[t_i,t_f]$. We define the differential operators
\begin{equation}\label{delpm} 
\partial_{\pm} := \partial_t \pm i\omega \; ,
\end{equation}
which we use to define the subspace of smooth functions 
\begin{equation}
C^\infty_p([t_i,t_f]) = \{f \in C^\infty([t_i,t_f]) \; | \; \partial_- f(t_i) = \partial_+ f(t_f) = 0  \} \, . 
\end{equation}
Given an open subset $I \subseteq [t_i,t_f]$, we then write $C^\infty_{cp}(I) \subseteq C^\infty_p([t_i,t_f])$ for the space of functions in $C^\infty_p([t_i,t_f])$ with compact support in $I$. We define the space of linear observables on $I$ as the complex
\begin{equation}
\begin{tikzcd}
\text{Obs}^{\rm lin}(I): \quad 
0 \arrow[r] & C^\infty_{cp}(I) \arrow[r,"d"] & C^\infty_{c}(I) \arrow[r] & 0 \, ,
\end{tikzcd}
\end{equation}
where $d = \partial_t^2 + \omega^2$. 

If $I = (a,b)$ is an open interval, the cohomology of $\text{Obs}^{\rm lin}(I)$ is still $\mathbb{C}^2$, since in that case $C_{cp}^\infty(I) = C_{c}^\infty(I)$. Any function with compact support in $(a,b)$ trivially satisfies $\partial_-f(t_i) = \partial_+f(t_f) = 0$. In the case of the half-open or closed intervals, the cohomology is no longer $\mathbb{C}^2$. For example, when $I = [t_i,a)$ we can define the chain map
\begin{equation}\label{ketcohomology}
\begin{tikzcd}
C^\infty_{cp}([t_i,a)) \arrow[d] \arrow[r] & 
C^\infty_c([t_i,a)) 
\arrow[d,"\Pi_0^{+}"] \\
0 \arrow[r ]& \mathbb{C} 
\end{tikzcd} \, ,
\end{equation}
with $\Pi_0^+(f) = \langle f, \sigma_- \rangle a^\dagger$, where we identify $a^\dagger$ with a basis element of $\mathbb{C}$. Indeed, this is a chain map, since
\begin{equation}\label{Piplusischain}
\Pi_0^+(\ddot f + \omega^2 f) = \sqrt{\frac{\hbar}{2 \omega}} \int_{t_i}^{t_f} d t \; (\ddot f(t) \omega^2 f(t))e^{i \omega t} = \sqrt{\frac{\hbar}{2 \omega}} (i \omega f(t_i) -\dot f(t_i))e^{i \omega t_i} = 0 \, .
\end{equation}
In the last step, we used the fact that $f \in C^\infty_{cp}([t_i,a))$. We will show below that $\Pi_0^+$ is in fact a quasi-isomorphism.

Similarly, when $I = (b,t_f]$, we define the chain map
\begin{equation}\label{bracohomology}
\begin{tikzcd}
C^\infty_{cp}((b,t_f]) \arrow[d] \arrow[r] & C^\infty_c((b,t_f]) \arrow[d,"\Pi_0^-"] \\
0 \arrow[r ]& \mathbb{C} 
\end{tikzcd}
\end{equation}
with $\Pi_0^-(f) = \langle f, \sigma_+ \rangle a$. Again, $a$ is some basis vector in $\mathbb{C}$. Performing a computation very similar to \eqref{Piplusischain} shows that $\Pi_0^-$ a chain map due to the boundary conditions we place on an $f$ of degree $-1$ at $t = t_f$. 

Before discussing the cohomology the complex $\text{Obs}^{\rm lin}([t_i,t_f])$, we consider the non-linear observables, for the cases considered above. Similar to the case on $\mathbb{R}$, we define
\begin{equation}
\text{Obs}^q(I) = (\text{Sym}(\text{Obs}^{\rm lin})(I),\delta_{\rm BV})
\end{equation}
with $\delta_{\rm BV} = d + \Delta$. We write $\text{Obs}^{cl}(I)$ for the same complex with $\Delta = 0$. The quasi-isomorphisms $\Pi_0^\pm$ induce a quasi-isomorphism
\begin{equation}
\Pi_0: \text{Obs}^{cl}(I) \longrightarrow \text{Sym}(\mathbb{C}) \; ,
\end{equation}
where we identify $\text{Sym}(\mathbb{C})$ as the algebra of polynomials in either $a$ (when $I = (b,t_f]$) or $a^\dagger$ (when $ I = [t_i,a)$), which we denote by $\mathbb{C}[a]$ or $\mathbb{C}[a^\dagger]$. When $I = (a,b)$, we can use the map $\Pi_0$ defined in the previous section to obtain a quasi-isomorphism into $\text{Sym}(\mathbb C^2)$, the algebra of polynomials in $a$ and $a^\dagger$, denoted by $\mathbb{C}[a,a^\dagger]$.

To obtain a quasi-isomorphism
\begin{equation}
\Pi: \text{Obs}^{q}(I) \longrightarrow \text{Sym}(\mathbb{C}) \; ,
\end{equation}
we can define $\Pi = \Pi_0 e^{-C}$, where we continue to use the contraction operator $C$ defined with the Green's function. As a map into quantum objects, we will find that $\mathbb{C}[a]$ is identified with the Fock space of the harmonic oscillator, interpreted as `outgoing states'
\begin{equation}
\bra 0 \; , \bra 0 a \; , \bra 0 a a \; , \ldots \, .
\end{equation}
Simiarly, $\mathbb{C}[a^\dagger]$ will become the Fock space in terms of `incoming states'
\begin{equation}
\ket 0 \; , a^\dagger \ket 0 \; , a^\dagger a^\dagger\ket 0 \; , \ldots \, .
\end{equation}

We end this section by proving that $\Pi$ is a quasi-isomorphisms from $\text{Obs}^q(I)$ to $\text{Sym}(\mathbb{C})$ when $I$ is half open.

\begin{prop}
For $I = [t_i,a)$ and $I = (b,t_f]$, the projections $\Pi: \text{Obs}^q(I) \rightarrow \text{Sym}(\mathbb{C})$ define quasi-isomorphisms of chain complexes.
\end{prop}
\begin{proof}
We will only prove the case $I = [t_i,a)$, since the case  $I = (b,t_f]$ is equivalent.

We first show that the chain map $\Pi_0: \text{Obs}^{\rm lin}([t_i,a)) \rightarrow \mathbb{C}$ is a quasi-isomorphism of the linear complex 
\begin{equation}
\begin{tikzcd}
C^\infty_{cp}([t_i,a)) \arrow[d] \arrow[rr,"d = \partial_t^2 + \omega^2"] & &  C^\infty_c([t_i,a)) \arrow[d,"\Pi_0"]\\
0 \arrow[rr] & & \mathbb{C}
\end{tikzcd} \, .
\end{equation}
This amounts to establishing  the following two statements: 
\begin{itemize}
\item $\ker d = 0$. This means that $C^\infty_{pm}([t_i,a)) \rightarrow 0$ becomes an isomorphism on $\ker d$.
\item $\ker \Pi_0 = \text{im} d$. This implies that $\Pi_0: C^\infty_c([t_i,a)) \rightarrow \mathbb{C}$ becomes an isomorphism on the quotient $C^\infty_c([t_i,a))/\text{im} \, d$.
\end{itemize}
The first point follows immediately from the fact that $\ddot f + \omega^2 f = 0$ does not have any non-zero compactly supported solutions on $[t_i,a)$.\footnote{Note that the boundary condition on 
$C_{cp}^\infty([t_i,a))$ is irrelevant here.}

To prove the second statement, first note that $\text{im} \, d \subseteq \Pi_0$, since $\Pi_0$ is a chain map. It remains to show that $\ker \Pi_0 \subseteq \text{im} \, d$. For this, assume that $f \in \ker \Pi_0$. Recall from the appendix that if $f \in \ker \Pi_0$, i.e.~$f \perp e^{i\omega t}$, the propagator $h(f)$ vanishes for $t > \text{supp}(f)$. Furthermore, $\partial_- h(f)(t_i) = 0$ for any $f$. It follows that 
\begin{equation}
    h: \ker \Pi_0 \longrightarrow C^\infty_{cp}([t_i,a))
\end{equation}
is well defined. Using $d h(f) = f$, we see that $\ker \Pi_0^+ \subseteq \text{im} \, d$ and therefore $\ker \Pi_0 = \text{im} \, d$.

The map $\Pi_0$ lifts to a quasi-isomorphism 
\begin{equation}
\Pi_0: \text{Obs}^{cl}([t_i,a)) \longrightarrow \text{Sym}(\mathbb{C}) \, .
\end{equation}
We will next show that $e^{-C}: \text{Obs}^{q}([t_i,a)) \rightarrow \text{Obs}^{cl}([t_i,a))$ is an isomorphism, which then implies that $\Pi = \Pi_0 e^{-C}$ is a quasi-isomorphism
\begin{equation}
\Pi: \text{Obs}^{q}([t_i,a)) \longrightarrow \text{Sym}(\mathbb C) \, .
\end{equation}

As for the case of open intervals the fact that $e^{-C}: \text{Obs}^{q}([t_i,a)) \rightarrow \text{Obs}^{cl}([t_i,a))$ is an isomorphism amounts to showing that
\begin{equation}
    [d,C] = - \Delta
\end{equation}
on quadratic observables $f \wedge g$. The only non-trivial case is the one when $f \in C^\infty_c([t_i,a))$ and $g \in C^\infty_{cp}([t_i,a))$. In this case we have that
\begin{equation}
  [d,C](f \wedge g) = - \Delta (f \wedge h d g) \, .  
\end{equation}
Therefore, if $hd (g) = g$ we are done. Indeed, in the appendix we show that for $g \in C^\infty_{cp}([t_i,a))$ this is the case.
\end{proof}

\subsubsection*{Closed intervals and inner product }

We now turn to the closed interval $I=[t_i, t_f]$ with linear observables given by the complex
\begin{equation}\label{LinObsClosed}
\begin{tikzcd}
\text{Obs}^{\rm lin}([t_i,t_f]):\quad 
0 \arrow[r] & C_{cp}^\infty([t_i,t_f]) \arrow[r,"d"] & C_{c}^\infty([t_i,t_f]) \arrow[r] & 0 \, .
\end{tikzcd}
\end{equation}
The non-linear observables are then given by the symmetric algebra
\begin{equation}\label{LambdaBV}
    \text{Obs}^{q}([t_i,t_f]):= ({\rm Sym}(\text{Obs}^{\rm lin}(I)), 
    \delta_{\rm BV}) \; ,
\end{equation}
with differential $\delta_{\rm BV} = d + \Delta$. We write $\text{Obs}^{cl}([t_i,t_f])$ for the case $\Delta = 0$.

When $I=[t_i, t_f]$, the complex $\text{Obs}^{\rm lin}(I)$ has trivial cohomology, which we will now show.
\begin{prop}
The complex $\rm{Obs}^{\rm lin}([t_i,t_f])$ given in \eqref{LinObsClosed} has trivial cohomology.
\end{prop}
\begin{proof}
We first show that $\text{Obs}^{\rm lin}([t_i,t_f])$ has trivial cohomology in degree zero. We need to show that any $g \in C^\infty_c([t_i,t_f])$ can be written as $df$ with $f \in C^\infty_{cp}([t_i,t_f])$. In the appendix \ref{App:Propagator} we show that $h$ in \eqref{eq:h}  has boundary conditions so that it defines a map
\begin{equation}
h: C^\infty_c([t_i,t_f]) \rightarrow C^\infty_{cp}([t_i,t_f])
\end{equation}
and further satisfies $d h = 1$. Therefore, setting $f = h(g)$ shows that any $g \in C^\infty_{cp}(I)$ is of the form $g = d f$. Hence, $\text{Obs}^{\rm lin}([t_i,t_f])$ has trivial cohomology in degree $0$.

In degree $-1$ the cohomology is given by the functions $f \in C^\infty_{cp}(I)$ satisfying $\ddot f + \omega^2 f = 0$. Note that any smooth function on $I = [t_i,t_f]$ has compact support since $I$ itself is compact. The solutions to $\ddot f + \omega^2 f = 0$ for any $f \in C^\infty_c(I)$ are given by oscillators, which we can write as
\begin{equation}
f(t) = a e^{-i \omega t} + b e^{i\omega t}\;, 
\end{equation}
with $a,b \in \mathbb{C}$. The condition $\partial_- f(t_i) = 0$ on $C_{cp}^\infty(I)$ then forces $a = 0$, while the condition $\partial_+ f(t_f) = 0$ forces $b = 0$. Therefore, the equation $\ddot f + \omega^2 f = 0$ only has trivial solutions on $C^\infty_{cp}(I)$ and hence $\text{Obs}^{\rm lin}(I)$ has trivial cohomology in degree $-1$.
\end{proof}

The above proposition implies that the zero map
\begin{equation}
\Pi_0: \text{Obs}^{\rm lin}(I) \longrightarrow 0
\end{equation}
is a quasi-isomorphism. It therefore lifts to a quasi-isomorphism of symmetric algebras
\begin{equation}
\Pi_0: \text{Obs}^{cl}(I) \rightarrow \text{Sym}(0) \, ,
\end{equation}
where $\text{Sym}(0) = \mathbb{C}$ by definition. We have along the same lines as before:
\begin{prop}
The map 
\begin{equation}\label{Pi0box}
\Pi:= \Pi_0 e^{-C}:\;  {\rm Obs}^{q}([t_i,t_f]) \longrightarrow \mathbb{C}
\end{equation} 
is a quasi-isomorphism.
\end{prop}

Note that the map $\Pi_0: \text{Obs}^{cl}([t_i,t_f]) \rightarrow \mathbb{C}$ is non-zero only on constant observables. On the other hand, the precomposition with $e^{-C}$ makes $\Pi$ also non-zero on any non-linear observable.\footnote{More precisely, it is non-zero on elements in $\text{Obs}^{q}([t_i,t_f])$ with an even number of functions.} It yields a non-zero number after a polynomial in the functions is fully contracted by the operator $e^{-C}$. When introducing the factorization algebra associated to the $\text{Obs}^q(I)$, we will see that $\mathbb{C}$ is the space of expactation values and $\Pi$ computes these expectation values of observables via Wick contractions generated by $e^{-C}$. Even better, it will also compute the overlap of states.


\subsection{The  factorization algebras}

We now describe the complete factorization algebra $\mathfrak{F}_{\rm QM}$ of 
the off-shell quantum mechanics of the harmonic oscillator. It is defined on the topological space $I=[t_i, t_f]$, and we define chain complexes to its connected open sets as 
\begin{equation} 
\mathfrak{F}_{\rm QM}(I) = \text{Obs}^q(I) \, .
\end{equation}

In order to complete the definition of $\mathfrak{F}_{\rm QM}$ we have to 
give the structure maps of this factorization algebra. The inclusion map $m_{I}^{J}$ for intervals $I\subseteq J$ is well-defined: 
All complexes $\text{Obs}^q(I)$ are defined as subcomplexes of $\text{Obs}^q([t_i,t_f])$ and $\text{Obs}^q(I)$ is a subcomplex of $\text{Obs}^q(J)$ for $I \subseteq J$. Similarly, the structure map $m_{I_1 I_2}^{J}: \mathfrak{F}(I_1) \otimes \mathfrak{F}(I_2) \longrightarrow \mathfrak{F}(J)$, 
given  by 
\begin{equation}
        m_{I_1 I_2}^{J}(F_1\otimes F_2) := m_{I_1}^J(F_1) \wedge m_{I_2}^J(F_2)\,,
\end{equation}
are well-defined. 

\medskip

The above factorization algebra should now be compared (i.e.~shown to be quasi-isomorphic) 
to standard quantum mechanics of 
the harmonic oscillator, which we will also describe as a factorization algebra as follows. 
This factorization algebra, denoted  $\mathfrak{F}_{\rm H}$ (short for Heisenberg)  in the following, 
is also defined on the closed interval $I=[t_i, t_f]$. Denoting the 
state space of the harmonic oscillator, spanned by ket states $\ket{\psi} = (a^{\dagger})^n\ket{0}$, $n\in \mathbb{N}_0$, by ${\cal H}$ and its dual, 
spanned by bra states $\bra{\psi}=\bra{0}a^n$, $n\in \mathbb{N}_0$, by ${\cal H}^*$, 
the factorization algebra $\mathfrak{F}_{\rm H}$ assigns the following vector spaces 
to the open sets: 
\begin{equation}\label{FactHeisenberg} 
     \mathfrak{F}_{\rm H}(I)   = 
   \begin{cases}
    {\rm Weyl} & \text{for $I= (a,b)$}\\
    {\cal H} & \text{for $I= [t_i,a)$}\\
     {\cal H}^* & \text{for $I=(b, t_f]$}\\
     \mathbb{C}  &\text{for $I=[t_i, t_f]$}
  \end{cases} \;, 
\end{equation}
where in the first line we associate to open intervals 
the Weyl algebra ${\rm Sym}(\mathbb{C}^2)$ of polynomials in $a, a^\dagger$ 
with product $\mu$ described above.

In order to complete the definition of $\mathfrak{F}_{\rm H}$ we have to 
give the structure maps of this factorization algebra. The inclusion maps 
$m_I^J : \,\mathfrak{F}_{\rm H}(I)\rightarrow \mathfrak{F}_{\rm H}(J)$ 
for intervals $I\subseteq J$ is just the identity if $I$ and $J$ belong to the same 
of the four classes of open sets in (\ref{FactHeisenberg}). Denoting generic elements of ${\cal H}$ by $\ket{\psi}$
and generic (normal ordered) operators of the Weyl algebra by ${\cal O}$, the remaining  cases are defined as follows: 
 \begin{equation}\label{QMInclusions} 
 \begin{split}
  &m_{(a,b)}^{[t_i, b)} : \,\mathfrak{F}_{\rm H}((a,b))\rightarrow \mathfrak{F}_{\rm H}([t_i, b))\;, \qquad\;\;\,
   m_{(a,b)}^{[t_i, b)}({\cal O}) =  {\cal O}\ket{0} \in {\cal H}  \;, \\
   &m_{(a,b)}^{(a,t_f]} : \,\mathfrak{F}_{\rm H}((a,b) )\rightarrow \mathfrak{F}_{\rm H}((a,t_f])\;, \qquad
   m_{(a,b)}^{(a,t_f]}({\cal O}) = \bra{0}{\cal O}  \in {\cal H}^* \;,  \\
   &m_{(b, t_f]}^{[t_i,t_f]} : \,\mathfrak{F}_{\rm H}((b, t_f] )\rightarrow \mathfrak{F}_{\rm H}([t_i,t_f])\;, \qquad
   m_{(b, t_f]}^{[t_i,t_f]} (\bra{\psi} ) =\bra{\psi}0\rangle  \in \mathbb{C}   \;, \\
   &m_{[t_i, a)}^{[t_i,t_f]} : \,\mathfrak{F}_{\rm H}([t_i, a) )\rightarrow \mathfrak{F}_{\rm H}([t_i,t_f])\;, \qquad
   m_{[t_i, a)}^{[t_i,t_f]} (\ket{\psi} ) =\bra{0}\psi\rangle  \in \mathbb{C}   \;,  \\
   &m_{(a,b)}^{[t_i,t_f]} : \,\mathfrak{F}_{\rm H}((a,b) )\rightarrow \mathfrak{F}_{\rm H}([t_i,t_f])\;, \qquad
   m_{(a,b)}^{[t_i,t_f]}({\cal O}) = \bra{0}  {\cal O} \ket{0} \in \mathbb{C}  \;, 
 \end{split} 
 \end{equation} 
where $\ket{0}$ denotes the ground state of the harmonic oscillator. 
 These definitions are natural and such that 
for $I\subseteq I' \subseteq J$ we have 
 \begin{equation}\label{HeisenbergInclCompatible}
  m_{I'}^{J} \circ m_I^{I'} = m_I^{J} \,: \; \mathfrak{F}_{\rm H}(I) \rightarrow \mathfrak{F}_{\rm H}(J)\;. 
 \end{equation} 
For instance, for $I=(a,b)$, $I'=[t_i, b)$ and $J=[t_i,t_f]$ the left-hand side yields 
 \begin{equation}
  m_{[t_i, b)}^{[t_i,t_f]}\Big(m_{(a,b)}^{[t_i,b)} ({\cal O})\Big)
   =  m_{[t_i, b)}^{[t_i,t_f]}\big(\,{\cal O}\ket{0}  \, \big) = 
   \bra{0} {\cal O}\ket{0} = m_{(a,b)}^{[t_i,t_f]}({\cal O}) \;, 
 \end{equation} 
 which in turn fixes the map in the last line of (\ref{QMInclusions}). 
 The other cases follow similarly.

In \eqref{QMInclusions}, we gave the inclusion maps in the form most well known to physicists. Since the class of operators we consider is always polynomial in $a$ and $a^\dagger$, we can make the indentification ${\rm Weyl} \cong \mathbb{C}[a,a^\dagger]$. Similarly, we can identify $\mathcal{H} \cong \mathbb{C}[a^\dagger]$, $\mathcal{H}^* \cong \mathbb{C}[a]$. The inclusions are then given by
 \begin{equation}
 \begin{aligned}
 \mathfrak F_{\rm H}((a,b)) &\longrightarrow \mathfrak F_{\rm H}((c,d)) \; , \quad &\mathfrak F_{\rm H}((a,b)) &\longrightarrow \mathfrak F_{\rm H}([t_i,t_f]) \; , \\
 f(a,a^\dagger) &\longmapsto f(a,a^\dagger) \; , \quad &f(a,a^\dagger) &\longmapsto f(0,0) \; , \\
 \mathfrak F_{\rm H}((a,b)) &\longrightarrow \mathfrak F_{\rm H}([t_i,b)) \; , \quad &\mathfrak F_{\rm H}((a,b)) &\longrightarrow \mathfrak F_{\rm H}((a,t_f]) \; ,  \\
 f(a,a^\dagger) &\longmapsto f(0,a^\dagger) \; , \quad &f(a,a^\dagger) &\longmapsto f(a,0) \; , \\
 \mathfrak F_{\rm H}([t_i,b)) &\longrightarrow \mathfrak F_{\rm H}([t_i,t_f]) \; , \quad &\mathfrak F_{\rm H}((a,t_f]) &\longrightarrow \mathfrak F_{\rm H}([t_i,t_f]) \; ,  \\
 f(a^\dagger) &\longmapsto f(0) \; , \quad &f(a) &\longmapsto f(0) \;.
\end{aligned}
\end{equation}
In this presentation it is immediately obvious that \eqref{HeisenbergInclCompatible} is satisfied. On the other hand, the definitions in \eqref{QMInclusions} apply for any algebra of operators and states, which we further elaborate at the end of this section.

 We next have to define the structure maps $m_{I_1I_2}^{J}:\mathfrak{F}_{\rm H}(I_1) \otimes \mathfrak{F}_{\rm H}(I_2) 
 \rightarrow \mathfrak{F}_{\rm H}(J)$ for disjoint intervals $I_1, I_2$ inside $J$. If all three intervals are open this is the structure map of $\mathfrak{F}_{\rm Weyl}$ defined in the previous section. 
 The remaining cases are defined as follows: 
  \begin{equation}\label{m2strucutremaps} 
  \begin{split}
   m_{[t_i,a), (b,c)}^{[t_i, c)}\big(\ket{\psi},{\cal O}\big) &= {\cal O}\ket{\psi} \in {\cal H}\;, \\
   m_{(a,b),(c,t_f]}^{(a,t_f]}\big({\cal O},\bra{\psi}\big) &= \bra{\psi}{\cal O} \in {\cal H}^* \;, \\
   m_{[t_i, a), (b,t_f]}^{[t_i, t_f]}\big( \ket{\psi_2}, \bra{\psi_1}\big) &= \langle\psi_1\ket{\psi_2} \in \mathbb{C}\;, \\ 
   m_{(a, b), (c,t_f]}^{[t_i, t_f]}\big( {\cal O}, \bra{\psi} \big) &= \bra{0} {\cal O} \ket{\psi} \in \mathbb{C}\;, \\ 
   m_{(a, c), (c',b)}^{[t_i, t_f]}\big( {\cal O}_1, {\cal O}_2 \big) &= \bra{0} {\cal O}_1{\cal O}_2 \ket{0} \in \mathbb{C}\;, \\
   m_{(a,c), (c', b)}^{(a,t_f]} ({\cal O}_1, {\cal O}_2) &= \bra{0}  {\cal O}_1 {\cal O}_2 
   \;  \in\, {\cal H}^*\;, \\
   m_{[t_i, c), (c', b)}^{[t_i, t_f]} \big(\ket{\psi}, {\cal O}\big) &=\bra{0} {\cal O}\ket{\psi} \ \in  \ 
   \mathbb{C}\;. 
  \end{split}
  \end{equation} 
 These formulas 
 are compatible with the above consistency conditions such as:  
 for $I_1, I_2 \subseteq I' \subseteq J$ we have 
  \begin{equation}
    m_{I'}^{J} \circ m_{I_1 I_2}^{I'} = m_{I_1 I_2}^{J}\;. 
  \end{equation}

  \subsubsection*{A quasi-isomorphism to standard quantum mechanics }

  Having thus completed the definition of the factorization algebra $\mathfrak{F}_{\rm H}$ of the quantum 
  mechanics of the harmonic oscillator we can now check that it is quasi-isomorphic to $\mathfrak{F}_{\rm QM}$. 
  The projection map from $\mathfrak{F}_{\rm QM}$ to $\mathfrak{F}_{\rm H}$ is given by $\Pi: \text{Obs}^q(I) \rightarrow \text{Sym}(V)$, where $\Pi$ and $V = \mathbb{C}^2, \mathbb{C}, 0$ depending on the type of interval $I$. For convenience, we recall that $\Pi$ is defined on a function $f$ on the interval $I$ via the projection $\Pi_0$ on $\text{Obs}^{\rm lin}(I)$ as 
   \begin{equation}\label{pi0intervals} 
     \Pi_0(f)    = 
   \begin{cases}
    f_- a + f_+ a^\dagger & \text{for $I= (a,b)$}\\
   f_+ a^\dagger\ket{0}  & \text{for $I= [t_i,a)$}\\
     f_- \bra{0} a  & \text{for $I=(b, t_f]$}\\
     0 &\text{for $I=[t_i, t_f]$}
  \end{cases} \;, 
\end{equation}
where $f_{\pm} = \langle f, \sigma_{\mp}\rangle$, and $\Pi_0 = 0$ on degree $-1$ functions. The extension to the symmetric algebra of polynomials in $f$ acts as a morphism, c.f.~(\ref{Piviamorph}). For instance, for  degree zero functions $f_1, \ldots,  f_n\in C_c^{\infty}([t_i,a))$ we set
 \begin{equation}
  \Pi_0(f_1\wedge \ldots \wedge f_n) = f_{1+}\cdots f_{n+} (a^{\dagger})^n \in \mathbb{C}[a^\dagger] \;, 
 \end{equation}
which we identify with the state $f_{1+}\cdots f_{n+} (a^{\dagger})^n \ket 0 \in \cal{H}$. The action $\Pi$ on the full BV algebra based $\mathfrak{F}_{\rm QM}$ is then given by $\Pi=\Pi_0 e^{-C}$, 
where $C$ denotes the Wick contraction operator in (\ref{improvedPI}).

In order to show that $\Pi$ defines a quasi-isomorphism we have to verify  that 
 \begin{equation}
  \Pi\circ m_I^J = m_I^J \circ \Pi\;, 
 \end{equation} 
where  $m_I^J$ on the left-hand side  refers to $\mathfrak{F}_{\rm QM}$ and $m_I^J$ on the right-hand side  
refers to $\mathfrak{F}_{\rm H}$. This condition is straightforwardly verified. 
Moreover, we have to show that the following diagram commutes:  
\begin{equation}\label{QMHCommut} 
    \xymatrix{
        \mathfrak{F}_{\rm QM}(I_1) \otimes \mathfrak{F}_{\rm QM}(I_2) \ar[r]^-{m} \ar[d]^\Pi & \mathfrak{F}_{\rm QM}(J) \ar[d]^\Pi\\
        \mathfrak{F}_{\rm H}(I_1) \otimes \mathfrak{F}_{\rm H}(I_2) \ar[r]^-{m} & \mathfrak{F}_{\rm H}(J)
    } 
\end{equation}  
In the case that all intervals  are open this has been verified in the proof of theorem \ref{thm:reals}.

For the other cases, consider for instance $I_1=[t_i, a)$, $I_2=(b,t_f]$ and $J=[t_i, t_f]$ and, correspondingly,  functions (linear monomials) 
$f_1\in \text{Obs}^{\rm lin}([t_i, a))$, $f_2\in \text{Obs}^{\rm lin}((b,t_f])$, 
we have 
 \begin{equation}\label{moerphismsmmmsmss} 
  (\Pi\circ m)(f_1\otimes f_2) = \Pi(f_1\wedge f_2) = \Pi_0e^{-C}(f_1\wedge f_2) = -C(f_1\wedge f_2) 
  = f_{1+} f_{2-}\;, 
 \end{equation}  
where we used that in this case $\Pi_0$ is non-zero only on numbers. 
Moreover, we used (\ref{CCOMPUUU}) in the last step. 
The above diagram commuting requires that this equals $m\circ \Pi$, 
which follows by a direct computation: 
 \begin{equation}
 \begin{split} 
  (m\circ \Pi)(f_1\otimes f_2) &=m(\Pi(f_1)\otimes \Pi(f_2)) =m(\pi_0(f_1)\otimes \pi_0(f_2)) \\
  &= m\big(f_{1+} a^{\dagger}\ket{0}, f_{2-}\bra{0}a\big)
  = f_{1+}f_{2-} \bra{0} a a^{\dagger}\ket{0}  = f_{1+}f_{2-} \;, 
 \end{split}  
 \end{equation} 
where we used (\ref{pi0intervals}) and the third of the structure maps in (\ref{m2strucutremaps}). 

For a slightly more subtle example consider $I_1=[t_i, a)$, $I_2=(b,c)$ and $J=[t_i, c)$, 
with $f_1\in {\rm Obs}^{\rm lin}(I_1)$ and $f_2\in {\rm Obs}^{\rm lin}(I_2)$. One then computes 
\begin{equation}\label{LeftHANdSIDEEE} 
\begin{split} 
 (\Pi\circ m)(f_1\otimes f_2) &= \Pi(f_1\wedge f_2) = \Pi_0 e^{-C} (f_1\wedge f_2) \\
 &= \Pi_0(f_1\wedge f_2) -C(f_1\wedge f_2)\ket{0}  \\
 &= f_{1+}f_{2+} (a^{\dagger})^2\ket{0} + f_{1+}  f_{2-}\ket{0} \;, 
\end{split} 
\end{equation}  
where we used 
$f_1\wedge f_2\in \mathfrak{F}_{\rm QM}([t_i,c))$ 
and 
(\ref{CCOMPUUU}). On the other hand, 
 \begin{equation}
 \begin{split} 
  (m\circ \Pi)(f_1\otimes f_2) &= m(\pi_0(f_1)\otimes  \pi_0(f_2)) 
  = m\big(f_{1+}a^{\dagger} \ket{0}\,,\; f_{2-} a + f_{2+} a^{\dagger}\big)  \\
  &= (f_{2-} a + f_{2+} a^{\dagger}) f_{1+} a^{\dagger} \ket{0} \\
  &= f_{2-} f_{1+} \ket{0} + f_{1+} f_{2+} (a^{\dagger})^2 \ket{0} \;, 
 \end{split}  
 \end{equation} 
where we used in the second line the first of the structure maps (\ref{m2strucutremaps}). 
This agrees with (\ref{LeftHANdSIDEEE}), hence proving for this special case 
that (\ref{QMHCommut}) commutes. The general proof works in a very similar way as the proof for $\mathfrak{F}_{\rm BV}$ on $\mathbb{R}$: 
\begin{thm}\label{thm:interval}
The map $\Pi = \Pi_0 \circ e^{-C}: \mathfrak{F}_{\rm BV}(I) \rightarrow \mathfrak{F}_{\rm H}(I)$, which is defined for any connected open set $I \subseteq [t_i,t_f]$, defines a morphism of factorization algebras 
\begin{equation}
    \Pi: \mathfrak{F}_{\rm BV} \longrightarrow \mathfrak{F}_{\rm H} \; .
\end{equation}
\end{thm}
\begin{proof}
We will not give the general proof, since it is very similar to the case on $\mathbb{R}$ after making the following observation. The map
\begin{equation}
    \rho_L: \text{Weyl} \otimes \mathcal{H} \longrightarrow \mathcal{H}
\end{equation}
defining the action of a normal ordered operator $\mathcal{O} \in \text{Weyl}$ on a state $\ket \psi \in \mathcal{H}$ can be written as
\begin{equation}
    \rho_L(f(a,a^\dagger), g(a^\dagger)) = \wedge|_{a = 0} \circ e^{\partial_a \otimes \partial_{a^\dagger}} (f(a,a^\dagger) \otimes g(a^\dagger)) \; ,
\end{equation}
where we identify $\mathcal{O} = f(a,a^\dagger)$ and $\ket \psi = g(a^\dagger)$ with polynomials and $\wedge|_{a = 0}$ means multiplication of polynomials and setting $a = 0$ in the end. Observe that the map $\rho_L$ therefore takes a very similar form as $\mu$ in $\eqref{Generalmu}$. Along the same lines, the map
\begin{equation}
    \rho_R: \mathcal{H}^* \otimes \text{Weyl} \longrightarrow \mathcal{H}^*
\end{equation}
that multiplies an operator on a bra-state takes the form
\begin{equation}
    \rho_R(g(a), f(a,a^\dagger)) = \wedge|_{a^\dagger = 0} \circ e^{\partial_a \otimes \partial_{a^\dagger}} (g(a) \otimes f(a,a^\dagger)) \; .
\end{equation}
Finally, the inner product
\begin{equation}
    s: \mathcal{H}^* \otimes \mathcal{H} \longrightarrow \mathbb{C}
\end{equation}
between states can be written as
\begin{equation}
s(f(a),g(a^\dagger)) = \wedge|_{a^\dagger = a = 0} \circ e^{\partial_a \otimes \partial_{a^\dagger}} (f(a) \otimes g(a^\dagger)) \; .
\end{equation}
With these definitions of $\rho_L,\rho_R,s$ together with the operator product $\mu$, the proof follows along the same lines as the proof of theorem \ref{thm:reals}.
\end{proof}

\subsection{Operators and  expectation values}

In this subsection we aim to illustrate the above result by giving examples of how to relate to familiar computations in ordinary quantum mechanics.

\subsubsection*{Computing expectation values}

We will now illustrate how the above quasi-isomorphism allows one to express  the computation 
of quantum expectation values at the level of the off-shell quantum mechanics encoded 
in the factorization algebra $\mathfrak{F}_{\rm QM}$. Consider the following bra and ket states 
of the harmonic oscillator: 
 \begin{equation}
  \ket{\psi_1}:= (a^{\dagger})^2\ket{0}\;, \qquad \bra{\psi_2}:= \bra{0}a^2\;. 
 \end{equation} 
Our goal is to express the computation of the overlap $\bra{\psi_2}\psi_1\rangle$ in terms 
of the structure maps of $\mathfrak{F}_{\rm QM}$.\footnote{Of course, $\bra{\psi_2}$ is just the dual  
state corresponding to $\ket{\psi_1}$, so one could skip the labels $1, 2$, but we find it notationally 
advantageous for the following discussion to keep them.}
Since bra and ket states are associated to half-open 
intervals, which here we take to be $I_1=[t_i,a) < I_2 = (b, t_f]$,  we consider the functionals 
 \begin{equation}\label{functrepaa} 
 \begin{split} 
  F^1 &:= f_{a^{\dagger}}^1 \wedge f_{a^{\dagger}}^1  \ \in \ \mathfrak{F}_{\rm QM}([t_i,a)) \,, \\
   F^2 &:= f_{a}^2 \wedge f_{a}^2 \;\,\,  \ \in \ \mathfrak{F}_{\rm QM}((b,t_f]) \;, 
 \end{split} 
 \end{equation} 
where $f_a$ and $f_{a^{\dagger}}$ are defined as above, so that they obey the integral conditions
 \begin{equation}\label{ffdaggerintegrals} 
  f_{a,-}=f_{a^{\dagger},+} =1\,, \qquad f_{a,+}=f_{a^{\dagger},-}= 0\,,
 \end{equation}  
and the superscripts $1,2$ indicate that the functions are compactly supported on the intervals 
$I_1$ and $I_2$, respectively. One can always find such functions, but they are far from 
unique. Here we simply assume again that a specific choice obeying the above conditions has been made. 
Despite this  ambiguity, these functionals are projected to the desired  states $\ket{\psi_1}$ and $\bra{\psi_2}$, 
as follows with a quick computation: 
 \begin{equation}\label{PiofF1} 
 \begin{split} 
  \Pi(F_1) &= \Pi_0\, e^{-C} (f_{a^{\dagger}}^1 \wedge f_{a^{\dagger}}^1) = 
  \pi_0(f_{a^{\dagger}}^1) \wedge \pi_0(f_{a^{\dagger}}^1) \\
  &= (f_{a^{\dagger},+}^1)^2(a^{\dagger})^2\ket{0}
  =(a^{\dagger})^2\ket{0} = \ket{\psi_1} \;. 
 \end{split}  
 \end{equation}  
 Here we used that by (\ref{CCOMPUUU}) and (\ref{ffdaggerintegrals}) the operator $e^{-C}$ acts trivially 
 on $F_1$, and we used the explicit form (\ref{pi0intervals}) of the projector for $I_1=[t_i,a)$, 
 together with (\ref{ffdaggerintegrals}) again. Similarly, using the form of the projector for the interval 
 $ I_2 = (b, t_f]$ one obtains 
  \begin{equation}\label{PiofF2} 
  \begin{split} 
   \Pi(F_2) &=   \Pi_0\, e^{-C}( f_{a}^2 \wedge f_{a}^2) 
 =   \pi_0(f_{a}^2) \wedge \pi_0(f_{a}^2) \\
    &=(f_{a,-}^2)^2 \bra{0} a^2 = \bra{0} a^2
    =\bra{\psi_2} \;.  
  \end{split}   
  \end{equation}

Next we compute the overlap $\bra{\psi_2}\psi_1\rangle$. In fact, we compute it twice,  by evaluating both sides of the quasi-isomorphism condition 
 \begin{equation}\label{morhoismDCVsdpfsanv} 
  \big(\Pi \circ m_{I_1I_2}^{[t_i, t_f]}\big)(F_1\otimes F_2)  = \big( m_{I_1I_2}^{[t_i, t_f]} \circ \Pi\big) (F_1\otimes F_2)\;.
 \end{equation} 
On the right-hand side we have 
 \begin{equation}\label{RHSslvnflvbel} 
 \begin{split}
  \big( m_{I_1I_2}^{[t_i, t_f]} \circ \Pi\big) (F_1\otimes F_2) &= m_{I_1I_2}^{[t_i, t_f]}\big(\Pi(F_1)\otimes \Pi(F_2)\big) \\
  &= m_{I_1I_2}^{[t_i, t_f]}\big(\ket{\psi_1}\,, \,  \bra{\psi_2} \big)\\
  &= \bra{\psi_2}\psi_1\rangle 
   \;, 
 \end{split}   
 \end{equation}
where we used (\ref{PiofF1}) and (\ref{PiofF2}), and we recalled 
that this structure map of $\mathfrak{F}_{\rm H}$, encoded in the third line of (\ref{m2strucutremaps}), 
is just the pairing of a bra and a ket state. 
This computation shows  what the right-hand side of (\ref{morhoismDCVsdpfsanv}) does: it projects 
the functionals $F_1$ and $F_2$ down to the familiar quantum mechanical states $\ket{\psi_1}$ and $\bra{\psi_2}$
 in $\mathfrak{F}_{\rm H}$, after which, by definition, the structure map just yields  the standard overlap. 
 Thus, 
 \begin{equation}\label{standardoverlappppp} 
  \big( m_{I_1I_2}^{[t_i, t_f]} \circ \Pi\big) (F_1\otimes F_2) = 
   \bra{0} a^2(a^{\dagger})^2\ket{0} =  \bra{0} a[a,(a^{\dagger})^2]\ket{0} 
   = 2  \bra{0} a a^{\dagger}\ket{0} = 2 \;. 
 \end{equation}

 Let us now contrast this with the left-hand side of (\ref{morhoismDCVsdpfsanv}), 
 where the structure map $m_{I_1I_2}^{[t_i, t_f]}$ is just the wedge product: 
  \begin{equation}
  \begin{split} 
     \big(\Pi \circ m_{I_1I_2}^{[t_i, t_f]}\big)(F_1\otimes F_2)  &= \Pi(F_1\wedge F_2) 
     =\Pi_0\,e^{-C}\big( f_{a^{\dagger}}^1 \wedge f_{a^{\dagger}}^1 \wedge 
      f_{a}^2 \wedge f_{a}^2  \big)\\
      &= \Pi_0\big(1-C+\tfrac{1}{2}C^2+\cdots\big) \big( f_{a^{\dagger}}^1 \wedge f_{a^{\dagger}}^1 \wedge 
      f_{a}^2 \wedge f_{a}^2  \big)\\
      &= \tfrac{1}{2}C^2 \big( f_{a^{\dagger}}^1 \wedge f_{a^{\dagger}}^1 \wedge f_{a}^2 \wedge f_{a}^2  \big)\\
      &= \tfrac{1}{2}C \big(4\, C\big(f_{a^{\dagger}}^1\wedge f_a^2\big)\wedge  f_{a^{\dagger}}^1\wedge  f_{a}^2  \big)\\
      &= 2\,  \big[C\big(f_{a^{\dagger}}^1\wedge f_a^2\big)\big]^2\\
      &=2\;, 
  \end{split} 
  \end{equation} 
 giving the same result as (\ref{standardoverlappppp}). 
 Here we used that $F_1\wedge F_2\in \text{Obs}^q([t_i,t_f])$ for which $\Pi_0$ is non-trivial only on numbers, 
 so that only the $C^2$ term contributes. Moreover, we used 
 $C\big(f_{a^{\dagger}}^1\wedge f_a^2\big)=-f^1_{a^{\dagger},+} f^2_{a,-}=-1$ following from 
(\ref{CCOMPUUU}) together with (\ref{ffdaggerintegrals}).\footnote{Note that even though the functions 
$f_{a^{\dagger}}^1$ and $ f_a^2$ have support on disjoint intervals, $C$ acting on their wedge product 
is still non-zero, as is evident from the computation in (\ref{CCOMPUUU}).}

 \subsubsection*{Operators and their action on states}

After discussing ket and bra states and their overlap in terms of factorization algebras we next 
illustrate how to represent familiar operators and their action on states in such terms.  
We begin with the number operator of the harmonic oscillator, 
 \begin{equation}\label{numberop} 
  \Hat{N} = a^\dagger a\,, 
 \end{equation}  
and ask which objects of the off-shell quantum mechanics encoded in the factorization algebra 
$\mathfrak{F}_{\rm QM}$ project to this operator. To this end let $f_{a}^1$ have the properties above and be compactly 
supported on $I_1=(a,b)$ and similarly let $f_{a^{\dagger}}^2$ be compactly supported on $I_2=(c,d)$ 
with $I_1<I_2$. We then claim that the functional 
 \begin{equation}\label{numberopFunctional} 
  F_N := f_a^1\wedge f_{a^{\dagger}}^2 \ \in \ \mathfrak{F}_{\rm QM}((a,d))\;, 
 \end{equation} 
 projects to the number operator (\ref{numberop}). This follows as above by a direct computation: 
  \begin{equation}
  \begin{split}
   \Pi(F_N) &= \Pi_0\, e^{-C} \big(f_a^1\wedge f_{a^{\dagger}}^2\big)\\
   &= \pi_0(f_a^1)\wedge \pi_0(f_{a^{\dagger}}^2) - C\big(f_a^1\wedge f_{a^{\dagger}}^2\big) \\
   &= a\wedge a^{\dagger}= a^{\dagger}\wedge a \;,  
  \end{split} 
  \end{equation} 
 where we used $C(f_a^1\wedge f_{a^{\dagger}}^2)=-f_{a,+}^{1} f_{a^{\dagger},-}^2=0$ by 
(\ref{ffdaggerintegrals}). Since a word in $a, a^{\dagger}$ in the symmetric algebra is 
identified with the corresponding \textit{normal ordered} quantum-mechanical operator, 
the above establishes that (\ref{numberopFunctional}) is projected to the number 
operator (\ref{numberop}).

Finally, let us illustrate how the action of an operator on a ket state is encoded in the 
off-shell factorization algebra. Concretely, we want to re-derive that the number operator 
acts as expected: 
 \begin{equation}
  \ket{n} := (a^{\dagger})^n\ket{0} \qquad \Rightarrow \qquad \Hat{N}\ket{n} = n\ket{n}\;. 
 \end{equation} 
As in (\ref{functrepaa}) we represent $\ket{n}$ by the functional 
 \begin{equation}\label{nstatefunctional} 
  F_{\ket{n}} \ := \ \underbrace{f_{a^{\dagger}}^1\wedge \cdots \wedge f_{a^{\dagger}}^1}_{\text{$n$ times}}
   \ \in \ \mathfrak{F}_{\rm QM}([t_i, a)) \;, 
 \end{equation} 
for which one may verify as above: $\Pi(F_{\ket{n}})=\ket{n}\in {\cal H}$.  
The action of $\Hat{N}$ on $ \ket{n}$ is now represented by the functional that is the wedge 
product of (\ref{numberopFunctional}) with (\ref{nstatefunctional}). More precisely, renaming the 
intervals as $I_1=[t_i,a)$, $I_2=(a,b)$ and $I_3=(c,d)$, where $I_1<I_2<I_3$, we write 
 (\ref{numberopFunctional}) as $f_a^2\wedge f_{a^{\dagger}}^3$ and build 
 the functional 
  \begin{equation}
   F_{\hat{N}\ket{n}} := F_{\ket{n}}\wedge f_a^2\wedge f_{a^{\dagger}}^3 \ \in \ \mathfrak{F}_{\rm QM}([t_i, d))\;. 
  \end{equation} 
 We then compute 
  \begin{equation}
  \begin{split} 
   \Pi(F_{\hat{N}\ket{n}}) &= \Pi_0\,e^{-C} \Big(
   \underbrace{f_{a^{\dagger}}^1\wedge \cdots \wedge f_{a^{\dagger}}^1}_{\text{$n$ times}} 
   \ \wedge \ f_a^2\wedge f_{a^{\dagger}}^3\Big) \\
   &= -n\,\Pi_0\Big(C(f_{a^{\dagger}}^1\wedge f_a^2) \, 
    \underbrace{f_{a^{\dagger}}^1\wedge \cdots \wedge f_{a^{\dagger}}^1}_{\text{$n-1$ times}} \,
    \wedge\,  f_{a^{\dagger}}^3\Big) \\
    &= n  \Big(  \underbrace{\pi_0(f_{a^{\dagger}}^1)\wedge \cdots \wedge \pi_0(f_{a^{\dagger}}^1)}_{\text{$n-1$ times}}
    \,\wedge\,  \pi_0(f_{a^{\dagger}}^3) \Big)\\
    &= n(a^{\dagger})^{n}\ket{0}\;. 
  \end{split} 
  \end{equation} 
Here we used that only the term linear in $C$ is non-zero, which follows since  $\pi_0(f_{a}^2)=0$ and 
$C(f_{a}^2\wedge f_{a^{\dagger}}^3)=0$, and so there are only the $n$ Wick contractions displayed above.

More generally, it should now be clear that familiar quantum-mechanical operations like operator 
products, operators acting on kets or bras and overlaps, can be expressed at the level of the 
off-shell quantum mechanics encoded in the factorization algebra $\mathfrak{F}_{\rm H}$.

\section{Spin-\(\frac{1}{2}\) System}

\subsection{Off-shell spin-\(\frac{1}{2}\) quantum mechanics }

We begin by constructing the factorization algebra encoding 
 the off-shell quantum mechanics of the spin-$\frac{1}{2}$ system. 
 To this end we need to define its chain complex, pass over to its dual and 
 define a BV algebra on the space of observables, which then will be used 
 to construct the factorization algebra $\mathfrak{F}_{{\rm spin}-\frac{1}{2}}$ that is quasi-isomorphic 
 to the standard formulation in terms of Pauli matrices acting on the Hilbert space $\mathbb{C}^2$. 
 
Let us then define  the chain complex encoding the classical dynamics.  
It is natural to seek some kind of one-dimensional Dirac operator to define the differential. 
The higher-dimensional Dirac operator is characterized by being a `square root' of 
the Klein-Gordon operator and acting on a multi-component vector (spinor). 
The remnant  of the Klein-Gordon operator in one dimension 
is $\partial_t^2+\omega^2$, which factorizes as 
 \begin{equation}
  \partial_t^2+\omega^2 = (\partial_t+ i\omega )(\partial_t -i\omega) \equiv  \partial_+ \partial_-\;, 
 \end{equation}  
where we used   the abbreviations (\ref{delpm}). 
Since we expect fermions to correspond to spinors, which carry at least two components, 
we double the function spaces in degree zero and one and define the chain complex
 \begin{equation}
 \begin{tikzcd}
 C^\infty(\mathbb{R})\oplus  C^\infty(\mathbb{R}) \arrow[rr,"d_0 = (\partial_+{,}\partial_-)"] & & C^\infty(\mathbb{R})\oplus  C^\infty(\mathbb{R})
 \end{tikzcd}
 \end{equation} 
where the notation indicates that the differential acts as 
 \begin{equation}
  d_0(\psi_1\oplus \psi_2) \equiv  d_0(\psi_1,\psi_2) = (\partial_+ \psi_1, \partial_- \psi_2) \;. 
 \end{equation} 

We next discuss its cohomology. As above, the cohomology in degree zero is just 
${\rm ker}(d_0)$, i.e., the space of classical solutions to the 
equations $(\partial_+ \psi_1, \partial_- \psi_2) =(0,0)$. In terms of 
 \begin{equation}
  \sigma_+(t) :=  e^{-i\omega t}\;, \qquad \sigma_-(t) :=   e^{i\omega t} \;, 
 \end{equation} 
satisfying $\partial_+\sigma_+=\partial_-\sigma_-=0$,  
we can parameterize a general solution as 
 \begin{equation}
 (\psi_1(t) ,\psi_2(t)) =  (\psi_1(0) \sigma_+(t) ,\, \psi_2(0) \sigma_-(t))\;. 
 \end{equation} 
Since the equations are first order in time derivatives, the solution is uniquely determined if 
the pair of functions $(\psi_1, \psi_2)$ is specified at time zero. On the total complex 
we can define the projector that acts as 
 \begin{equation}
  \pi(\psi_1,\psi_2) = ( \psi_1(0), \psi_2(0))\in \mathbb{C}^2\;, 
 \end{equation} 
in degree zero and is zero in degree one.  (Recall that the functions are complex valued.)   
Displaying this projection as a diagram we we have 
\begin{equation}
    \begin{tikzcd}
    C^\infty(\mathbb{R})\oplus C^\infty(\mathbb{R}) \arrow[d,"\pi"] \arrow[r,"\partial_+  +  \partial_-"] & 
    C^\infty(\mathbb{R})\oplus C^\infty(\mathbb{R})  \arrow[d] \\
    \mathbb{C}^2 \arrow[r] & 0
    \end{tikzcd} \, .
\end{equation}
As above it is easy to see that the cohomology in degree one is trivial, and hence the above 
presents the projection to the cohomology $\mathbb{C}^2$ concentrated in degree zero.

Let us next turn to the `dual' 
complex of linear functionals. Denoting a pair of functions 
as $\Psi =(\psi_1,\psi_2)$ we consider the functionals 
 \begin{equation}
  F[\Psi] = \int_I dt \big( {{f_1(t)}}  \psi_1(t) +{f_2(t)} \psi_2(t) \big) \ \equiv \ 
  \langle{\bf f}, \Psi\rangle\;, 
 \end{equation} 
where $f_1, f_2$ are compactly supported functions on the interval $I$, 
and we sometimes use the abbreviation ${\bf f}=(f_1,f_2)$, together with the 
notation $\langle \;,\, \rangle$ for the pairing defined by the integral above. 
Demanding $F[d\Psi]=dF[\Psi]$ for $\Psi$ of degree zero one finds: 
 \begin{equation}\label{fermiond-1} 
  d_{-1}({\bf f}) = d_{-1}(f_1, f_2) = -(\partial_-f_1, \partial_+  f_2) \,, 
 \end{equation}  
for ${\bf f}$ of degree $-1$, where we used that under integration by parts $\partial_+\rightarrow -\partial_-$ 
and $\partial_-\rightarrow -\partial_+$. 
The dual chain complex of linear observables is thus given by 
 \begin{equation}\label{xdcom}
 (X(I) ,d) \ :=
 \begin{tikzcd}
C_c^\infty({I})\oplus  C_c^\infty({I}) \arrow[rrr,"d_{-1} = (-\partial_-{,}-\partial_+)"] & & & C_c^\infty({I})\oplus  C_c^\infty({I})\; . 
 \end{tikzcd}
 \end{equation} 
We can also define the projector onto the cohomology of this  complex. In degree $-1$ the 
projector acts trivially, and in degree zero it acts as 
 \begin{equation}\label{Pibffff} 
  \Pi_0({\bf f}) = \Pi(f_1, f_2) = \big(\langle f_1, \sigma_+\rangle , \langle f_2, \sigma_-\rangle\big)
  \equiv f_{1-} a +f_{2+} a^{\dagger} \ \in \ \mathbb{C}^2\;, 
 \end{equation} 
where $\langle f_1, \sigma_-\rangle=\int_I dt {f_1(t)}\sigma_-(t)$ and similarly 
for the second component, and $a, a^{\dagger}$ some basis of $\mathbb{C}^2$.
Thus, diagramatically we have 
\begin{equation}
\begin{tikzcd}
C^\infty_c(I) \oplus C^\infty_c(I) \arrow[d] \arrow[r,"d_{-1}"] & \arrow[d,"\Pi_0"] C^\infty_c(I) \oplus C^\infty_c(I) \\
0 \arrow[r] & \mathbb C^2
\end{tikzcd}
\; .
\end{equation}
With (\ref{fermiond-1}) and (\ref{Pibffff}) it quickly follows that $\Pi$ is a chain map, i.e., 
that for ${\bf f}$ in degree $-1$ we have $(\Pi\circ d)({\bf f})=0$.

\medskip

We next pass over to the BV algebra on the symmetric space of polynomials. However, since  we are 
dealing with fermions the underlying wedge product will be  anticommutative 
for objects in degree zero, which is implemented by a degree shift by $+1$ in 
the graded commutativity. More precisely, the resulting space based on the chain complex $(X,d)$ 
defined in (\ref{xdcom}), which we denote  by 
 \begin{equation}\label{classicalfermionobs} 
 {\rm Obs}^{cl}_{\rm ferm}(I) := {\rm Sym}(X(I) ,d)\;,  
 \end{equation} 
 consists of polynomials 
 \begin{equation}
  F = {\bf f}_1\wedge {\bf f}_2\wedge \cdots \wedge {\bf f}_n\;, 
 \end{equation} 
subject to 
 \begin{equation}\label{newedgesigns}  
  {\bf f}_1\wedge {\bf f}_2 = (-1)^{({\bf f}_1+1)({\bf f}_2+1) } {\bf f}_2\wedge {\bf f}_1\;. 
 \end{equation}
With this convention, the wedge product of two degree zero objects is antisymmetric, 
while the wedge product of a degree zero and a degree $-1$ object is symmetric. 
The differential $d$ can be extended to the full space (\ref{classicalfermionobs}) by demanding 
the Leibniz rule 
 \begin{equation}
   d({\bf f}_1\wedge {\bf f}_2) = d{\bf f}_1 \wedge {\bf f}_2 + (-1)^{{\bf f}_1+1} {\bf f}_1\wedge d{\bf f}_2\;, 
 \end{equation} 
 upon which this space becomes a (degree-shifted) differential graded associative 
algebra. In order to define a BV algebra we next have to deform this differential by the 
de Rham differential, which acts trivially on numbers and linear monomials and on quadratic monomials 
acts as 
 \begin{equation}
   \Delta\big({\bf f}\wedge {\bf g}\big) 
   := 
   - \int_I dt\,\big(f_1(t) g_2(t) +f_2(t) g_1(t)\big)  
   \;, 
 \end{equation} 
where ${\bf f}$ has degree zero, and ${\bf g}$ has degree $-1$, thus making $\Delta$ a degree $+1$ operator. 
(Note the off-diagonal pairing and the absence of the $i\hbar$ factor, i.e., the naive expression has effectively been 
multiplied by $\frac{i}{\hbar}$. This will be explained below.)
The BV operator is extended to the full symmetric algebra by (\ref{fullDeltaAction}), where the signs 
generated are now those implied by (\ref{newedgesigns}). 
We now define the algebra of quantum observables 
\begin{equation}
    {\rm Obs}_{\rm ferm}^q(I) := \Big({\rm Sym} (X(I)),\,\delta_{\rm BV}\Big)\;, 
    \label{eq:obs_q_defFERM}
\end{equation}
based on the same space as in (\ref{xdcom}) but with the 
differential having the quantum deformation given by the BV operator: $\delta_{\rm BV} := d + \Delta$.

Next, we turn to the projector $\Pi$ from the full algebra of quantum observables down to its cohomology. 
As in (\ref{chainmapproppp}) one finds that the naive extension acting as a morphism does not define a chain map
and hence needs to be modified using $\Delta$ and a homotopy map of degree $-1$. To define the latter we 
set 
 \begin{equation}
  {\bf h}({\bf f}) =  {\bf h}(f_1, f_2) = -\big(\partial_+(h(f_1)),\, \partial_-(h(f_2))\big)\;,  
 \end{equation} 
where $h$ was defined in (\ref{eq:h}) and satisfies $(\partial_t^2+\omega^2)(h(f))=f$. 
Recalling $\partial_t^2+\omega^2= \partial_+ \partial_-$ it then follows with (\ref{fermiond-1}) that 
\begin{equation}\label{bfhinverse} 
 d({\bf h}({\bf f})) = {\bf f}\;. 
\end{equation} 
With the original projection acting as a morphism by $\Pi_0$, the full projection is now given by 
 \begin{equation}
     \Pi := \Pi_0 e^{-C}\;, \quad \text{where}\quad C:= \Delta \circ (1 \otimes {\bf h})\;. 
 \end{equation} 
For ${\bf f}$ of degree zero and ${\bf g}$ of degree $-1$ we then have 
 \begin{equation}
 (\Pi\circ \delta_{\rm BV})({\bf f}\wedge {\bf g}) = \Pi_0 e^{-C} (d+\Delta)({\bf f}\wedge {\bf g})
 = \Pi_0 e^{-C}\left(-{\bf f}\wedge d{\bf g}+ \Delta ({\bf f}\wedge {\bf g})\right) \;, 
 \end{equation} 
where we used the Leibniz rule and $d{\bf f}=0$. We observe that there are two quantum 
contributions: $\Delta ({\bf f}\wedge {\bf g})$ and 
 \begin{equation}
  C({\bf f}\wedge d{\bf g}) = -\Delta({\bf f}\wedge {\bf h}(d{\bf g})) = -\Delta({\bf f}\wedge {\bf g})\;, 
 \end{equation}  
 which cancels the first quantum contribution. Here we used (\ref{bfhinverse}) and the convention that 
 the odd ${\bf h}$ creates a minus sign when moved passed an object ${\bf f}$ where $|{\bf f}|+1$ is odd, 
 i.e., $(1\otimes {\bf h})({\bf f}\otimes {\bf g})=(-1)^{{\bf f}+1} {\bf f} \otimes {\bf h} ({\bf g})$. 
 This shows that $\Pi\circ \delta_{\rm BV}=0$ and hence that $\Pi$ is a chain map. 
 
 For later use we record that for two degree zero objects ${\bf f}^1=(f^1_1, f^1_2)$ and ${\bf f}^2=(f^2_1, f^2_2)$ we have 
  \begin{equation}\label{fermionicCaction} 
 C\big({\bf f}^1\wedge {\bf f}^2\big)=
  \begin{cases}
    f^2_{1-} f^1_{2+} & \text{if $I_1<I_2$}\\
    -f^1_{1-} f^2_{2+} &\text{if $I_1> I_2$}
  \end{cases}\;, 
 \end{equation}
where $I_1$ denotes the interval on which ${\bf f}^1$ is 
compactly supported and similarly for $I_2$. 
This follows by a direct computation analogous to (\ref{CCOMPUUU}). Note that this result is such that 
 \begin{equation}
  C\big({\bf f}^1\wedge {\bf f}^2\big)= - C\big({\bf f}^2\wedge {\bf f}^1\big)\;, 
 \end{equation} 
as it should be since degree zero objects anticommute in our convention.  
 
 Finally,  the factorization algebra $\mathfrak{F}_{{\rm spin}-\frac{1}{2}}$  that encodes the off-shell quantum mechanics 
 of the spin$-\frac{1}{2}$ particle assigns to each open interval 
   \begin{equation}
    \mathfrak{F}_{{\rm spin}-\frac{1}{2}}(I) =  {\rm Obs}_{\rm ferm}^q(I) = {\rm Sym} \Big(X(I),\,\delta_{\rm BV}\Big)\;, 
   \end{equation} 
 and the structure maps are given as above by inclusion and wedge product. 
 
 \subsubsection*{Derivation of the chain complex from the Dirac action}

 The form of the differential $\delta_{\rm BV}$ can be motivated by the Dirac action in (1+0) dimensions. In that case, the Dirac action is given by
 \begin{equation}
S_D[\psi_2,\psi_1] =  \int_I d t \, \psi_2(i \hbar\partial_t - m)\psi_1 \, , 
 \end{equation}
 where $m$ is a constant with energy dimension. Defining $\omega = \hbar m$, we obtain a quantity that has the dimension of a frequency. The action in terms of $\omega$ is then
\begin{equation}
S_D[\psi_2,\psi_1] = i\hbar \int_I d t \,  \psi_2(\partial_t +i \omega)\psi_1 \, ,  
\end{equation}
Note that this action has equations of motion
\begin{equation}
(\partial_t +i \omega)\psi_1 = (\partial_t - i \omega) \psi_2 = 0 \, . 
\end{equation}

In the textbook version of the BV formalism, the BV Laplacian is
\begin{equation}
\Delta = \int \text d t \, \Big(\frac{\delta^2}{\delta \psi_1(t) \delta \psi_1^*(t)} + \frac{\delta^2}{\delta \psi_2(t) \delta \psi_2^*(t)} \Big) \, .
\end{equation}
The action then induces the differential
\begin{equation}
\begin{split}
d(\psi^*_1(t)) &:= \{S,\psi^*_1(t)\} := \Delta(S\psi^*_1(t)) = i \hbar (\partial_t -i \omega)\psi_2(t) \; , \\
d(\psi^*_2(t)) &:= \{S,\psi^*_1(t)\} := \Delta(S\psi^*_1(t)) = i\hbar(\partial_t + i \omega) \psi_1(t) \; .
\end{split}
\end{equation}
Note that the differential on the anti-field $\psi_1^*$ produces the equation of motion for the field $\psi_2$. For this reason, we make a redefinition $(\psi_1^*,\psi_2^*) \rightarrow (\psi_2^*,\psi_1^*)$. The Laplacian and the differential then take the form
\begin{equation}
\begin{split}\label{StandardDiracBVdifferential}
\Delta &= \int \text d t \, \Big(\frac{\delta^2}{\delta \psi_1(t) \delta \psi_2^*(t)} + \frac{\delta^2}{\delta \psi_2(t) \delta \psi_1^*(t)} \Big) \; , \\
\quad d(\psi^*_1(t)) &= i\hbar(\partial_t + i \omega) \psi_1(t) \;, \\ 
d(\psi^*_2(t)) &= i\hbar(\partial_t - i \omega) \psi_2(t) \, .
\end{split}
\end{equation}

The full quantum BV differential takes the form
\begin{equation}
\delta_{BV} = \{S,-\} - i \hbar \Delta = d - i\hbar \Delta \, .
\end{equation}
Note that since $d$ has a itself a global factor of $i\hbar$, we can rescale $\delta_{\rm BV}$ by this factor to obtain
\begin{equation}
\delta_{\rm BV} = d - \Delta \, ,
\end{equation}
where now $d(\psi_1^*,\psi_2^*) = (\partial_+ \psi_1,\partial_- \psi_2)$.  Note that any rescaling of $\delta_{\rm BV}$ is harmless, since it does not alter the cohomology.

In order to relate this to our BV algebra of smooth, compactly supported functions, we define the linear observables to be
\begin{equation}
F_{\textbf f}[\psi_1,\psi_2] := \int_I d t \, (f_1(t) \psi_1(t) + f_2 \psi_2(t))
\end{equation}
in degree zero and
\begin{equation}
F_{\textbf g}[\psi_1^*,\psi_2^*] := \int_I d t \, (g_1(t) \psi_1^*(t) + g_2 \psi_2^*(t))
\end{equation}
in degree $-1$, where $f_1,f_2,g_1,g_2$ are all compactly supported. The action of $d = \{S,-\}$ then is given by  
\begin{equation}
d(g_1,g_2) = (-\partial_- g_1,-\partial_+ g_2) \, ,
\end{equation}
while $\Delta$ acts on $\textbf f = (f_1,f_2)$ and $\textbf g = (g_1,g_2)$ as
\begin{equation}
\Delta(\textbf f \wedge \textbf g) := \Delta(F_{\textbf f} F_{\textbf g}) = \int_I d t \, (f_1(t) g_2(t) + f_2(t) g_1(t)) \, .
\end{equation}
After making a final redefinition $\Delta \rightarrow -\Delta$, we find that
$\delta_{\rm BV} = d + \Delta$ 
is the differential on $\text{Obs}^q_{\rm ferm}(I)$ in \eqref{eq:obs_q_defFERM}. The chain complex $\text{Obs}^q_{\rm ferm}(I)$ therefore is by no means any chain complex, but it can derived by applying the textbook BV formalism to the Dirac action in (1+0) dimensions. Furthermore, the fact that the Dirac action comes with a global factor of $i\hbar$ explains the absence of this factor in our definition of $\delta_{BV}$. We also learn that there is no classical limit in sending $\hbar \rightarrow 0$.\footnote{On the other hand, there is another kind of classical limit obtained by just setting $\Delta = 0$. This recovers the complex \eqref{classicalfermionobs}.}
 
 \subsection{Quasi-isomorphism to standard  spin-$\frac{1}{2}$ quantum mechanics }
 
 We begin by describing the factorization algebra $\mathfrak{F}_{\rm Pauli}$ that encodes the 
 standard formulation of spin-$\frac{1}{2}$ quantum mechanics. The operator algebra of the latter 
 is conventionally given by Pauli matrices $\sigma_i=(\sigma_x, \sigma_y, \sigma_z)$ satisfying 
 \begin{equation}\label{sigmamatricesrel} 
    \begin{gathered}
        [\sigma_i, \sigma_j] = 2i\epsilon_{ijk}\sigma_k\;, \qquad 
        \{\sigma_i, \sigma_j\} = 2\delta_{ij}\;, 
    \end{gathered}
\end{equation}
but here it is more convenient to express them in terms of a fermionic Weyl algebra of operators 
$a, a^{\dagger}$ obeying 
\begin{equation}
    \begin{gathered}
        \{a, a^\dagger\} = 1\;, \qquad 
        \{a, a\} = \{a^\dagger, a^\dagger\} = 0\;, 
    \end{gathered}
\end{equation}
where $\{\,,\}$ denotes the anticommutator. Note that 
these relations imply  in particular $a^2=(a^{\dagger})^2=0$. 
Defining then
\begin{align}
    \sigma_x &:= a+ a^\dagger \;, \\
    \sigma_y &:= {i}(a-a^\dagger)\;, \\
    \sigma_z &:= {i}[a, a^\dagger]\;, 
\end{align} 
one finds that these operators satisfy the desired relations (\ref{sigmamatricesrel}). 

We consider now the fermionic Weyl algebra 
 \begin{equation}\label{Weylferm} 
  {\rm Weyl}_{\rm ferm} := \mathbb{C}[a,a^\dagger] \;, 
 \end{equation} 
consisting of polynomials in odd variables $a, a^{\dagger}$ over $\mathbb{C}^2$. We denote the antisymmetric concatenation again by 
the wedge symbol $\wedge$, so that  $a\wedge a=a^{\dagger}\wedge a^{\dagger}=0$. The algebra $\mathbb{C}[a,a^\dagger]$ is then finite-dimensional with basis 
 \begin{equation}
  1\;, \quad a\;, \quad a^{\dagger} \;, \quad a\wedge a^{\dagger}=-a^{\dagger}\wedge a\;, 
 \end{equation}
i.e., ${\rm Weyl}_{\rm ferm}$ is four-dimensional as a vector space over the complex numbers. (This is nothing but the algebra of 
complex $2\times 2$ matrices.) We next define the non-commutative product 
$\mu: {\rm Weyl}_{\rm ferm}\otimes {\rm Weyl}_{\rm ferm}\rightarrow {\rm Weyl}_{\rm ferm}$ by 
 \begin{equation}
 \begin{split} 
  \mu(a^{\dagger},a) &= a^{\dagger}\wedge a = - a\wedge a^{\dagger}\;, \\
  \mu(a, a^{\dagger})  &=a\wedge a^{\dagger} +1 = -a^{\dagger} \wedge a +1\;, \\
  \mu(a,a) &= \mu(a^{\dagger}, a^{\dagger}) = 0\; .
 \end{split} 
 \end{equation} 
The general formula for $\mu$ is similar to the bosonic case:
\begin{equation}\label{RecursiveActionFERM} 
\mu = \wedge \circ e^{\partial_a \otimes \partial_{a^\dagger}} \, .
\end{equation}
Note that in this case, the exponential is necessarily finite.

The factorization algebra $\mathfrak{F}_{\rm Pauli}$ assigns to every open interval $I$ 
the same space (\ref{Weylferm}). The structure maps are the identity, 
and  for any two words $a,b$ we set 
 \begin{equation}\label{DefmmuFERM}  
     m_{\mu}(a,b) := 
   \begin{cases}
    -\mu(b, a) & \text{if $I_1<I_2$}\\
    \mu(a,b) &\text{if $I_1> I_2$}
  \end{cases} \  \in  \ \mathbb{C}[a,a^\dagger]\;. 
\end{equation}

With these definitions and (\ref{fermionicCaction}) we can now verify the quasi-isomorphism property 
for the special case 
 \begin{equation}
  (\Pi\circ m)({\bf f}^1\otimes {\bf f}^2) = (m\circ \Pi)({\bf f}^1\otimes {\bf f}^2)\;, 
 \end{equation} 
 for both orderings of the respective (disjoint) intervals $I_1$ and $I_2$ on which 
 ${\bf f}^1$ and ${\bf f}^2$ are supported.

\subsection{Half-open and closed intervals and the Hilbert space \texorpdfstring{$\mathbb C^2$}{C\texttwosuperior}}

As with the harmonic oscillator, we can incorporate states by by working on the closed interval $[t_i,t_f]$. In order to obtain the correct cohomologies, we impose the following boundary conditions.
\begin{equation}
\begin{split}
C^\infty_{i}([t_i,t_f]) &:= \{f \in C^\infty([t_i,t_f]) \; | \; f(t_i) = 0\} \; ,    \\
C^\infty_{f}([t_i,t_f]) &:= \{f \in C^\infty([t_i,t_f]) \; | \; f(t_f) = 0\} \; .
\end{split}
\end{equation}
Given a connected open subset $I \subseteq [t_i,t_f]$, we set $C^\infty_{ci}(I) \subseteq C^\infty_{i}([t_i,t_f])$ for the space of smooth functions $f \in  C^\infty_{i}([t_i,t_f])$ with compact support in $I$. Similarly, we write $C^\infty_{cf}(I) \subseteq C^\infty_{f}([t_i,t_f])$ for the space of smooth functions $f \in  C^\infty_{f}([t_i,t_f])$ with compact support in $I$.

Given a connected and open set $I \subseteq [t_i,t_f]$, we define the following complex of linear observables
\begin{equation}\label{Fermlinearobs}
\text{Obs}^{\rm lin}_{\rm ferm}(I) \ := 
\begin{tikzcd}
 C^\infty_{cf}(I) \oplus C^\infty_{ci}(I) \arrow[rr,"(-\partial_-{,}-\partial_+)"] & & C^\infty_c(I) \oplus C^\infty_c(I)\; .
\end{tikzcd}
\end{equation}


To compute the cohomologies of these chain complexes, we first give two partial results.
\begin{prop}
The chain complexes
\begin{equation}
\begin{tikzcd}\label{TrivCom1}
0 \arrow[r] & C^\infty_{cf}([t_i,a)) \arrow[r,"-\partial_-"] & C^\infty_c([t_i,a)) \arrow[r] & 0 \, ,
\end{tikzcd}
\end{equation}
\begin{equation}\label{TrivCom2}
\begin{tikzcd}
0 \arrow[r] & C^\infty_{ci}((b,t_f]) \arrow[r,"-\partial_+"] & C^\infty_c((b,t_f]) \arrow[r] & 0
\end{tikzcd}
\end{equation}
have vanishing cohomology. 
\end{prop}
\begin{proof}
We will only prove that \eqref{TrivCom1} has vanishing cohomology, since the proof for \eqref{TrivCom2} is essentially the same.

We will use the following fact from homological algebra \cite{weibel1994introduction}, usually called \emph{zig-zag lemma}. Loosely speaking, it says that a short exact sequence of chain complexes induces a long exact sequence in (co-)homology. This means the following. Let $A, B, C$ be chain complexes, $f: A \rightarrow B^\bullet$ an injective chain map and $g: B \rightarrow C$ a surjective chain map such that $\ker g = \text{im} \; f$. In this case, we say that
\begin{equation}
\begin{tikzcd}
0 \arrow[r] & A \arrow[r,"f"] & B \arrow[r,"g"] & C \arrow[r] & 0        
\end{tikzcd}
\end{equation}
is a short exact sequence.\footnote{In general, a sequence of composable maps is called exact, if the kernel of one map is the image of the previous one. A sequence is called short exact, if there are only two non-zero maps.} The zig-zag lemma then states that there are linear maps $s_i: H^{i}(C) \rightarrow H^{i+1}(A)$ between cohomology spaces, such that
\begin{equation}
\begin{tikzcd}
\ldots \arrow[r,"H^i(g)"] & H^{i}(C) \arrow[r,"s_i"] & H^{i+1}(A) \arrow[r,"H^{i+1}(f)"] & H^{i+1}(B) \arrow[r,"H^{i+1}(g)"] & H^{i+1}(C) \arrow[r,"s_{i+1}"] & \ldots
\end{tikzcd}
\end{equation}
is an exact sequence. Here, $H^{i}(f)$ and $H^{i}(g)$ denote the linear maps induced by $f$ and $g$ on the $i$th cohomology spaces. We have the following immediate corollaries: Iff $A$ has trivial cohomology, then $g$ is a quasi-isomorphism. Similarly, if $C$ has trivial cohomology, then $f$ is a quasi-isomorphism. If at least two out of the three complexes have trivial cohomology, then also the third one has trivial cohomology.

Coming back to the proof of the proposition, we claim that there is a short exact sequence of chain complexes
\begin{equation}\label{Trivialinitialcomplex}
\begin{tikzcd}
& 0 \arrow[d] & 0 \arrow[d] & 0 \arrow[d] & \\
0 \arrow[r] & C^\infty_{ci}([t_i,a)) \arrow[d,"-\partial_-"]\arrow[r,"\iota_{-1}"] & C^\infty_{cf}([t_i,a)) \arrow[d,"-\partial_-"] \arrow[r,"\pi_{-1}"] & \mathbb{C} \arrow[d,"1"] \arrow[r] & 0 \, , \\
0 \arrow[r] & \ker \pi_0 \arrow[r,"\iota_{0}"] \arrow[d] & C^\infty_c([t_i,a)) \arrow[r,"\pi_{0}"] \arrow[d]& \mathbb{C} \arrow[r] \arrow[d] & 0 \, . \\
& 0 & 0 & 0 &
\end{tikzcd}
\end{equation}
Here, the complexes are drawn vertically and we define the first non-trivial row to be of degree $-1$ and the second non-trivial row to be of degree zero. The map $-\partial_+ h$ contracts the complex in the first column to zero, while the the last column obviously also has trivial cohomology. If the claim is true (i.e.~the sequence in \eqref{Trivialinitialcomplex} is short exact), then the zig-zag lemma implies that the cohomology of the middle complex is also zero.

We now prove the claim. We first construct the chain map
\begin{equation}
\begin{tikzcd}
0 \arrow[r] & C^\infty_{cf}([t_i,a)) \arrow[d,"\pi_{-1}"] \arrow[r,"-\partial_+"] & C^\infty_{c}([t_i,a)) \arrow[d,"\pi_0"] \arrow[r] & 0 \\
0 \arrow[r] & \mathbb C \arrow[r,"1"] & \mathbb{C} \arrow[r] & 0 
\end{tikzcd} \ .
\end{equation}
where $\pi_{-1}(f) = e^{-i\omega t_i}f(t_i)$ and $\pi_0 = \langle f,\sigma_+\rangle$. We check that this defines a chain map. Indeed,
\begin{equation}
\pi_0(-\partial_- f) = \langle - \partial_- f, \sigma_+\rangle = f(t_i) e^{-i \omega t_i} + \langle f,\partial_+\sigma_+ \rangle = f(t_i) e^{-i \omega t_i} = \pi_{-1}(f) \, .
\end{equation}
Next, we obviously have $\ker \pi_{-1} = C^\infty_{cf}([t_i,a))$. Therefore, the kernels $\ker \pi_{-1} $ and $\ker \pi_0$ constitute exactly the chain complex
\begin{equation}
\begin{tikzcd}
0 \arrow[r] & C^\infty_{ci}([t_i,a)) \arrow[r,"-\partial_-"] & \ker \pi_0 \arrow[r] & 0 \, .
\end{tikzcd}
\end{equation}
This proves that \eqref{Trivialinitialcomplex} is a short exact sequence of chain complexes.
\end{proof}
\begin{prop}
The chain complexes
\begin{equation}
\begin{tikzcd}\label{SpinCom1}
0 \arrow[r] & C^\infty_{ci}([t_i,a)) \arrow[r,"-\partial_+"] & C^\infty_c([t_i,a)) \arrow[r] & 0 \, ,
\end{tikzcd}
\end{equation}
\begin{equation}\label{SpinCom2}
\begin{tikzcd}
0 \arrow[r] & C^\infty_{cf}((b,t_f]) \arrow[r,"-\partial_-"] & C^\infty_c((b,t_f]) \arrow[r] & 0
\end{tikzcd}
\end{equation}
have cohomology isomorphic to $\mathbb{C}$.
\end{prop}
\begin{proof}
We will not give the proof of this proposition, since it is almost the same as the previous one. For example, to compute the cohomology of $\eqref{SpinCom1}$, one can show that 
\begin{equation}
\begin{tikzcd}
& 0 \arrow[d] & 0 \arrow[d] & 0 \arrow[d] & \\
0 \arrow[r] & C^\infty_{ci}([t_i,a)) \arrow[d,"-\partial_+"]\arrow[r,"1"] &  C^\infty_{ci}([t_i,a)) \arrow[r] \arrow[d,"-\partial_+"] & 0 \arrow[d] \arrow[r] & 0 \, , \\
0 \arrow[r] & \ker \pi_0 \arrow[r,"\iota_{0}"] \arrow[d] & C^\infty_c([t_i,a)) \arrow[r,"\pi_{0}"] \arrow[d]& \mathbb{C} \arrow[r] \arrow[d] & 0 \, . \\
& 0 & 0 & 0 &
\end{tikzcd}
\end{equation}
is a short exact sequence of chain complexes. Since the first column has trivial cohomology, it follows that $\pi_0: C^\infty_c([t_i,a)) \rightarrow \mathbb{C}$ is a quasi-isomorphism.\footnote{This argument here is just a fancy way of saying that $\ker \pi_0$ has trivial cohomology and so it has to be a quasi-isomorphism.}
\end{proof}
The two propositions show that the chain complexes $\text{Obs}_{\rm ferm}^{\rm lin}([t_i,a))$ and $\text{Obs}^{\rm lin}_{\rm ferm}((b,t_f])$ have cohomology isomorphic to $\mathbb{C}$, since each of them is a direct sum of one chain complex with trivial cohomology and one with cohomology isomorphic to $\mathbb{C}$. The cohomology of $\text{Obs}^{\rm lin}_{\rm ferm}([t_i,t_f])$ can be computed similarly.
\begin{prop}
The chain complexes
\begin{equation}
\begin{tikzcd}\label{ClosedCohom1}
0 \arrow[r] & C^\infty_{cf}([t_i,t_f]) \arrow[r,"-\partial_-"] & C^\infty_c([t_i,t_f]) \arrow[r] & 0 \, ,
\end{tikzcd}
\end{equation}
\begin{equation}\label{ClosedCohom2}
\begin{tikzcd}
0 \arrow[r] & C^\infty_{ci}([t_i,t_f]) \arrow[r,"-\partial_+"] & C^\infty_c([t_i,t_f]) \arrow[r] & 0
\end{tikzcd}
\end{equation}
have vanishing cohomology. In particular, $\text{Obs}^{\rm lin}_{\rm ferm}([t_i,t_f])$ has vanishing cohomology, as it is the direct sum of the above two complexes.
\end{prop}
\begin{proof}
Following appendix \ref{App:Propagator}, the propagator $-\partial_+ h$ maps exactly into $C^\infty_{f,0}([t_i,t_f])$ and therefore can be used to show that the complex \eqref{ClosedCohom1} has trivial cohomology. Similarly, $-\partial_- h$ can be used to show that $\eqref{ClosedCohom2}$ has trivial cohomology.
\end{proof}

The classical space of all observables on a connected open subsets $I \subseteq [t_i,t_f]$ is now given by the symmetric algebra over the chain complex \eqref{Fermlinearobs}, which we denote by
\begin{equation}
\text{Obs}^{cl}_{\rm ferm}(I) = \text{Sym} \big(\text{Obs}^{\rm lin}_{\rm ferm}(I)\big) \, .
\end{equation}
Here, $\rm{Sym}(-)$ accounts for the parity of the linear observables, which at the case at hand means that the linear observables in degree $0$ are anti-commuting. 

At the level of cohomology, we have
\begin{equation}\label{SpinBoundaryCoho}
\text{Obs}^{cl}_{\rm ferm}([t_i,a)) \simeq \mathbb{C}[a^\dagger] \; , \quad \text{Obs}^{cl}_{\rm ferm}((b,t_f]) \simeq \mathbb{C}[a]\; , \quad \text{Obs}^{cl}_{\rm ferm}([t_i,t_f]) \simeq \mathbb{C} \, .
\end{equation}
The space $\mathbb{C}[a^\dagger]$ is spanned by polynomials in a single anticommuting variable $a^\dagger$. When identifying $a$ and $a^\dagger$ with the matrices
\begin{equation}\label{aadaggerasmatrices}
a = \begin{pmatrix}
    0 & 0 \\
    1 & 0
\end{pmatrix} 
\; , \quad
a^\dagger = \begin{pmatrix}
0 & 1 \\
0 & 0
\end{pmatrix} \; ,
\end{equation}
the state $\lambda_0 + \lambda_1 a^\dagger \in \mathbb{C}[a^\dagger] $ is identified with
\begin{equation}\label{adaggerasvectors}
\lambda_0 + \lambda_1 a^\dagger \simeq
\lambda_0 \begin{pmatrix}
0 \\
1
\end{pmatrix}
+ \lambda_1\begin{pmatrix}
0 & 1 \\
0 & 0
\end{pmatrix}
\begin{pmatrix}
0 \\
1
\end{pmatrix}
=
\begin{pmatrix}
\lambda_1 \\
\lambda_0
\end{pmatrix} \, ,
\end{equation}
where we think of the constant $1 \in \mathbb{C}[a^\dagger]$ as the ground state. In this way, $\mathbb{C}[a^\dagger]$ spans the two dimensional Hilbert space of a quantum mechanical spin-$\tfrac{1}{2}$ system. Along the same lines, the cohomology $\text{Obs}_{\rm ferm}^{cl}((b,t_f])$, which is $\mathbb{C}[a]$, consists of polynomials in a single odd variable $a$. As with the harmonic oscillator, the cohomology of $\text{Obs}_{\rm ferm}^{cl}([t_i,t_f])$ is identified with the expectation values of operators and overlaps of states.

At the quantum level, we define the complexes with the differential $\delta_{\rm BV} = d + \Delta$. Explicitly, we set
\begin{equation}
    \text{Obs}_{\rm ferm}^{q}(I) := \big(\text{Sym}(\text{Obs}_{\rm ferm}^{\rm lin}(I)),\delta_{\rm BV}\big) \; .
\end{equation}
As in the bosonic case, the quasi-isomorphisms indicated in \eqref{SpinBoundaryCoho} still apply, but the quasi-isomorphism is now given by
\begin{equation}\label{QuantumSpinQI}
\Pi = \Pi_0 e^{-C} 
\end{equation}
with
\begin{equation}
\Pi_0(\mathbf f_1 \wedge \cdots \wedge \mathbf f_n) = \pi_0(\mathbf f_1) \wedge \cdots \wedge \pi_0(\mathbf f_n)
\end{equation}
and $C$ is the second order derivation extending
\begin{equation}
C(\mathbf f,\mathbf g) = \frac{i}{2\omega}\int_{t_i}^{t_f}  d t ds \, \big( \partial_+ f_1(t) e^{-i\omega|t-s|} g_2(s) + \partial_ -f_2(t) e^{-i\omega|t-s|} g_1(s)\big) \; ,
\end{equation}
where $\mathbf f = (f_1,f_2)$ and $\mathbf g = (g_1, g_2)$. Note that $C(\mathbf f, \mathbf g) = -C(\mathbf g,\mathbf f)$, which means that it can be extended to the alternating algebra. We will not give a proof of the fact that \eqref{QuantumSpinQI}, since the proof is essentially the same as for the harmonic oscillator. 

We define the factorization algebra $\mathfrak{F}_{\rm {spin}-\frac{1}{2}}$ on a closed interval as follows.
\begin{equation}
\mathfrak{F}_{\rm {spin}-\frac{1}{2}}(I) = \text{Obs}_{\rm ferm}^{q}(I)
\end{equation}
The inclusions $m_{I}^J$ and mulitplications $m_{I_1,I_2}^J$ are the the standard ones.

The cohomology of $\mathfrak{F}_{\rm {spin}-\frac{1}{2}}$ is given by factorization algebra derived from the algebra $\text{Mat}_{2 \times 2}(\mathbb{C})$ of $2\times 2$ matrices with complex entries acting on its defining representation $\mathbb{C}^2$ from the left and from the right, together with the standard inner product
\begin{equation}
    \mathbb{C}^2 \otimes \mathbb{C}^2 \longrightarrow \mathbb{C} \; .
\end{equation}
We denote this factorization algebra by $\mathfrak{F}_{\rm Pauli}$. The maps $\mathfrak{F}_{\rm Pauli}(a,b) \rightarrow \mathfrak{F}_{\rm Pauli}([t_i,b))$ and $\mathfrak{F}_{\rm Pauli}(a,b) \rightarrow \mathfrak{F}_{\rm Pauli}((a,t_f])$ can be chosen such that it picks the vector $(0,1) \in \mathbb{C}^2$. The algebra and its representations are identified as $\text{Alt}(\mathbb{C}^2) \cong \text{Mat}_{2 \times 2}(\mathbb{C})$, $\text{Alt}(\mathbb{C}) \cong \mathbb{C}^2$ as indicated in \eqref{aadaggerasmatrices} and \eqref{adaggerasvectors}.

The following statement is analogous to theorem \ref{thm:interval}.
\begin{thm}
The factorization algebra $\mathfrak{F}_{\rm{spin}-\frac{1}{2}}$ on is quasi-isomorphic to the algebra $\mathfrak{F}_{Pauli}$ as factorization algebras on $[t_i,t_f]$.
\end{thm}

\section{Conclusion and Outlook} 

In this paper we presented a reformulation of 
the quantum mechanics of the harmonic oscillator and the spin-$\frac{1}{2}$ system, 
utilising the notions of BV algebras and factorization algebras, following the program of Costello and Gwilliam. This formulation 
is based on off-shell field configurations
and as such is much closer to the path integral formulation than to the canonical or on-shell formulation of quantum mechanics. 
In this one uses the BV formalism, 
which was originally introduced in the context of path integral quantization in order to deal with subtleties arising in certain gauge theories, in a somewhat non-standard fashion, 
by showing its power for non-gauge theories 
and as an alternative to the path integral. 
Our off-shell formulation generalizes previous work by Costello and Gwilliam by including boundary points, adding ket and bra states and the computation of overlaps to 
the structure maps of the factorization algebra. While this off-shell quantum mechanics is quasi-isomorphic to on-shell or canonical quantum mechanics, and thus equivalent in the world of factorization algebras, the framework of factorization algebras may be more appropriate in QFT, where 
both canonical quantization and the path integral are plagued with problems that complicate the proper mathematical definition of QFT.

The results presented here may be extended 
in various directions such as: 

\begin{itemize}

\item One of the most obvious  shortcomings 
of the off-shell formulation presented here 
is the restriction to the free systems of the harmonic oscillator and the spin-$\frac{1}{2}$ system. 
The extension to interacting theories in perturbation theory, with polynomial interactions, should  be rather straightforward and was in fact already partially done by  Costello and Gwilliam \cite{Costello_Gwilliam_2021}. 
It remains, however,  to add boundary points so as to include states to have a full-fledged reformulation of even just perturbative quantum mechanics.

\item 

The most exciting extension would be to the truly non-perturbative context. In this case 
the symmetric algebra of polynomials is not going to be sufficient, and so here it remains to identify a suitable class of local or multi-local functionals. 
Such progress could be particularly fruitful 
when applied to gauge theories such as Yang-Mills theory, where one would like to find a new tool to attack its strong coupling mysteries. The classical version of this problem has been considered in \cite{Alfonsi:2023qpv}, which does not require factorization algebras since factorization algebras are only necessary at the quantum level.

\item 

It would also be interesting to explore the 
relation of the present formulation to 
geometric quantization 
and the so-called BV-BFV formalism \cite{Cattaneo:2012qu,Cattaneo:2015vsa}. It would in particular be desirable to find a framework which allows one to treat different polarizations.  

\end{itemize}

 \section*{Acknowledgments} 

We are grateful to 
Roberto Bonezzi, Owen Gwilliam,  
Maria Kallimani and Allison Pinto for discussions, explanations 
and collaborations on related subjects. 

\noindent
 The work of C.C.~is funded by the Deutsche Forschungsgemeinschaft (DFG, German Research Foundation), ”Homological Quantum Field Theory”, Projektnummer 9710005691.

\appendix

\section{BV algebras  and the path integral}

In this appendix we give a self-contained discussion  of the BV cohomology and 
its relation to finite-dimensional path integrals (Gaussian integrals). 
 This is based on an interpretation of 
 the anti-bracket of the BV formalism due to Witten \cite{Witten:1990wb}
 and the subsequent 
observation due to 
Gwilliam and Johnson-Freyd 
that the  cohomology of the BV differential computes 
 finite-dimensional path integrals \cite{gwilliamFeynman, gwilliam2012factorization}.
 
 We take the coordinates $x^i$, $i=1,\ldots, n$,  of  $\mathbb{R}^n$ as the dynamical `fields' with the action 
 function $S:\mathbb{R}^n\rightarrow \mathbb{R}$ given by 
  \begin{equation}\label{classicalactionnn} 
   S(x) = \frac{1}{2} x^i A_{ij} x^{j}\;, 
  \end{equation} 
 where $A$ is a symmetric non-degenerate $n\times n$-matrix. 
 The $x^i$ are the finite-dimensional analogue of the $\phi(t)$ in the main text. 
 The path integrals here are the Gaussian integrals that  compute 
 the expectation values of  polynomial functions  $F:\mathbb{R}^n\rightarrow \mathbb{R}$, 
\begin{equation}\label{pathintegralGauss} 
    \expval{F} := \frac{1}{N}\int d^nx\, F(x)\, e^{-\frac{1}{\hbar}S(x)}\,,
\end{equation}
with 
${N} = (2\pi\hbar)^{\frac{n}{2}}(\det A)^{-\frac{1}{2}}$
the usual normalization factor so that $\expval{1}=1$. 

In the following we will rephrase the computation of the above integrals in terms of differential forms 
and their de Rham cohomology, 
which we now briefly recall. 
On $\mathbb{R}^n$ the differential forms form the chain complex 
  \begin{equation} 
\begin{array}{ccccccccccc} 
 0 \,\longrightarrow\,
  \Omega^{0} \, \xlongrightarrow{d}\, \Omega^1 \,\xlongrightarrow{d}\,   \cdots  \,\xlongrightarrow{d}\,
\Omega^{n-1} 
\, \xlongrightarrow{d}\, \Omega^{n}   
\longrightarrow \,0\;, 
\end{array}
\end{equation} 
where $\Omega^p$ is the space of $p$-forms, and the de Rham differential maps $d:\Omega^p\rightarrow \Omega^{p+1}$ 
so that $d^2=0$. In addition, there is the familiar wedge product of differential forms that \textit{i)} is graded commutative, 
meaning that $\alpha\wedge \beta=(-1)^{pq}\beta\wedge \alpha$ for a $p$-form $\alpha$ and a $q$-form $\beta$, 
\textit{ii)} is such that the Leibniz rule holds for the de Rham differential $d$, 
and 
\textit{iii)} is associative, meaning $(\alpha\wedge \beta)\wedge \gamma= \alpha\wedge (\beta\wedge \gamma)$.
An algebraic structure with these data (a differential $d$ and  a product $\wedge$) satisfying these 
axioms is called a differential graded (dg) commutative algebra. 
The $k$-th de Rham cohomology is the quotient vector space 
 \begin{equation}\label{deRhamcoho} 
  H_{\rm dR}^k := \frac{{\rm ker}\, d_k}{{\rm im}\,d_{k-1}}\;, 
 \end{equation} 
where the notation  $d_k$ indicates the de Rham differential acting on $k$-forms. 
The elements of $H_{\rm dR}^k$ are equivalence classes $[\alpha]$ of closed $k$-forms $\alpha\in \Omega^k$, 
$d\alpha=0$, with the equivalence relation $\alpha \sim \alpha+d\beta$ for any $(k-1)$-form $\beta$. 
The wedge product gives rise to a graded commutative and associative product on cohomology, defined by 
 \begin{equation}\label{cohomologyalgebra} 
   [\alpha_1] \cdot [\alpha_2] := [\alpha_1\wedge \alpha_2]\;. 
 \end{equation} 
This product is independent of the chosen representatives and hence  well-defined, 
which follows with the Leibniz rule and $d\alpha_1=d\alpha_2=0$.  

We now explain that the computation of the path integral can be interpreted 
in terms of de Rham cohomology. 
We define  the  top-form ($n$-form or volume form) 
 \begin{equation}\label{volumeform} 
  \omega :=  \tfrac{1}{{N}}\, e^{-\frac{1}{\hbar}S(x)}\, dx^1\wedge \ldots \wedge dx^n\;, 
 \end{equation} 
obeying $\int \omega=1$, 
in terms of which the path integral (\ref{pathintegralGauss}) becomes 
 \begin{equation}\label{Quantumexp} 
  \expval{F} = \int_{\mathbb{R}^n}  F\,  \omega \;. 
 \end{equation} 
On a compact connected manifold without boundary 
the de Rham cohomology in top degree is isomorphic to $\mathbb{R}$. For the non-compact space $\mathbb{R}^n$ we could instead pick the subcomplex of differential forms
 \begin{equation}
   \alpha= \tfrac{1}{p!} 
   e^{-\frac{1}{\hbar}S(x)} 
   \alpha_{i_1\ldots i_p}(x)dx^{i_1}\wedge 
   \cdots \wedge dx^{i_p}\;, 
 \end{equation} 
where the functions $\alpha_{i_1\ldots i_p}(x)$ are polynomials, in which case the cohomology in top degree is also isomorphic to $\mathbb{R}$. However, the prize to pay is that the wedge product is not closed on this subspace. In any case, given a choice of $n$-form $\omega$, any other $n$-form is a scalar (function) times $\omega$. 
The cohomology being isomorphic to $\mathbb{R}$ then implies that any $n$-form can be written as  a real \textit{number}  
times $\omega$, up to exact terms. Applied to $F \omega$ this means 
 \begin{equation}
  F  \omega = \expval{F}   \omega+d\omega_{n-1}\;, 
 \end{equation} 
 for some $(n-1)$-form $\omega_{n-1}$, 
where the number  on the right-hand side was immediately identified as $\expval{F}$, which follows 
by integrating both sides of this equation. Thus, taking the cohomology class, 
 \begin{equation}
   [ F  \omega] =  \expval{F}  [\omega] \;.
 \end{equation}  
In this sense, the computation of the expectation value $\expval{F}$ is a 
problem in the de Rham cohomology of differential forms.

Despite the observable or functions being in degree zero, the relevant cohomology was in top degree $n$, 
using  the choice of top-form $\omega$. However,  in 
the infinite-dimensional setting of QFT there is no  notion of top-form. 
Another shortcoming is that the wedge product of the de Rham complex does not encode the product of observables (functions). 
This is the reason that (\ref{cohomologyalgebra}) does \textit{not} imply $ \expval{F_1 F_2}= \expval{F_1} \expval{F_2}$. 
Nevertheless, given 
a choice of top form in the finite-dimensional setting, the above is equivalent to a BV algebra with a BV differential, 
whose product does encode the algebra of functions. Such a BV algebra 
may be well-defined in the infinite-dimensional context of QFT independent of a notion of  top-form. 

This BV algebra arises, in finite dimensions, on the `dual' complex of polyvector fields as follows. 
We first  note that the `classical theory' 
defined by the action (\ref{classicalactionnn}) can be 
encoded  in the chain complex  
 \begin{equation}\label{finitedimComplex} 
 \begin{split}
  0 \longrightarrow &\; V_0 \longrightarrow V_1\longrightarrow 0 \\
  & \{ x^i \}  \quad\;\{ \theta_i \} \
 \end{split} 
 \end{equation} 
with the one non-trivial differential given by 
 \begin{equation}
  (d x)_i := A_{ij} x^j\;. 
 \end{equation} 
The space of observables is then the space of functions on this space: the space of functions of $x^i$ and $\theta_i$: 
 \begin{equation}\label{polyvectorsasfunctions} 
  F(x,\theta) = \sum_{p=0}^{n} \frac{1}{p!} F^{i_1\ldots i_p}(x)\, \theta_{i_1}\cdots \theta_{i_p}\;, 
 \end{equation} 
where we take the $\theta_i$ to be anticommuting, $\theta_i\theta_j=-\theta_j\theta_i$ (which is  suggested  by 
the $\theta_i$ living in the odd degree one in (\ref{finitedimComplex})). 
Consequently, the coefficient functions $ F^{i_1\ldots i_p}$ are totally antisymmetric and hence  define  polyvector fields 
on $\mathbb{R}^n$, 
with a natural grading given by $p$: for  $p=0$ we have the functions, for $p=1$ we have the 
vector fields, for $p=2$ we have the bivector fields, etc. 
The polyvectors form the chain complex 
   \begin{equation} 
\begin{array}{ccccccccccc} 
 0 \,\longrightarrow\,
  {\rm Vec}^{n} \, \xlongrightarrow{\Delta}\,{\rm Vec}^{n-1} \,\xlongrightarrow{\Delta }\,   \cdots  \,\xlongrightarrow{\Delta }
  \,{\rm Vec}^{1} 
\, \xlongrightarrow{\Delta}\, {\rm Vec}^{0}   
\longrightarrow \,0\;, 
\end{array}
\end{equation} 
where ${\rm Vec}^k$ is the space of polyvectors of rank $k$, and the differential $\Delta$ (the BV Laplacian),  
defined by 
 \begin{equation}\label{secondorderDElta} 
 \Delta:= \frac{\partial^2}{\partial x^i \partial \theta_i} \;, 
 \end{equation} 
obeys $\Delta^2=0$ and acts as the divergence, hence reducing the rank by one. 
Furthermore, there is a wedge product of polyvectors, which in terms of the presentation in 
(\ref{polyvectorsasfunctions}) is just the point-wise product of superfields. 
In particular, for polyvectors of rank 0, i.e., for functions, this reduces to the ordinary product of functions. 
Importantly,  $\Delta$ does not obey the Leibniz rule w.r.t.~this wedge product because 
it  is a second-order operator, as is evident from (\ref{secondorderDElta}). 
The second-order character can be expressed purely algebraically as follows: $\Delta$ is of second order if 
its failure  to act via the Leibniz rule (the anti-bracket), 
 \begin{equation}\label{anTIBracket} 
  \{ F_1, F_2\}  := \Delta (F_1 F_2) - \Delta(F_1)F_2- (-1)^{F_1} F_1\Delta(F_2)\;, 
 \end{equation} 
 defines a first-order operator $\partial_F:=\{F,\, \cdot\, \}$ that does act via the Leibniz rule, 
 which is the case here. 
 We can now  define a BV algebra: \\[1ex] 
{\bf Definition:} A BV algebra is a graded vector space together with a degree $+1$ differential $\Delta $, 
obeying $\Delta^2=0$, and a graded symmetric and associative product so that $\Delta$ is a second-order 
operator. 

\medskip

We next note that given a choice of top-form, which here is (\ref{volumeform}), 
there is an isomorphism between the de Rham complex 
and the complex of polyvector fields: there is an invertible map $\iota: {\rm Vec}^{k}\rightarrow \Omega^{{n-k}}$, 
defined for $F\in {\rm Vec}^{k}$ as 
 \begin{equation}
  \iota(F) = \frac{1}{(n-k)!k!}\omega_{i_1\ldots i_{n-k}j_1\ldots j_k} F^{j_1\ldots j_k} dx^{i_1}\wedge \cdots \wedge dx^{i_{n-k}}  
  \in \Omega^{n-k} \;. 
 \end{equation} 
Using this map one can define an invariant pairing between a rank-$k$ polyvector $F$ and 
a $k$-form $\alpha$ as follows: 
 \begin{equation}
  \langle F, \alpha \rangle = \int_{\mathbb{R}^n} \iota(F)\wedge \alpha \ \in  \ \mathbb{R}\;. 
 \end{equation}  
Writing in local coordinates 
$\alpha = \frac{1}{p!} \alpha_{i_1\ldots i_p}(x) dx^{i_1}\wedge \ldots \wedge dx^{i_p}$ 
and similarly for $F$, and using 
(\ref{volumeform}), this can also be written as 
  \begin{equation}\label{pairingoncemore} 
  \langle F, \alpha \rangle  
  = \frac{1}{p!} \frac{1}{{N}} \int_{\mathbb{R}^n}   \big(F^{i_1\ldots i_p} \alpha_{i_1\ldots i_p} \big)  e^{-\frac{1}{\hbar}S} d^nx\;. 
 \end{equation} 
The de Rham differential $d$ can be dualized to a differential $\delta$ (BV operator) on the polyvectors  by demanding that for
a rank-$p$ polyvector $F$ and  a $(p-1)$-form $\beta$
 \begin{equation}
  \langle F, d\beta\rangle  + \langle \delta F, \beta\rangle  =  0\;, 
 \end{equation} 
from which one obtains with (\ref{pairingoncemore}) 
 \begin{equation}
  \delta F = \sum_{p=1}^{n} \frac{1}{(p-1)!} \, e^{\frac{1}{\hbar}S}\partial_j\big(e^{-\frac{1}{\hbar}S}F^{ji_1\ldots i_{p-1}}\big)\theta_{i_1}\cdots \theta_{i_{p-1}}\;.
 \end{equation} 
 For instance, for a vector field $V=V^i\theta_i\in {\rm Vec}^1$ one computes 
  \begin{equation}\label{deltaVcomp} 
  \begin{split} 
   \delta V &= e^{\frac{1}{\hbar}S}\partial_i\big(e^{-\frac{1}{\hbar}S} V^i\big) \\
   &= \partial_i V^i - \tfrac{1}{\hbar} \partial_iS \,V^i
   = \Delta(V) - \tfrac{1}{\hbar}\{S,V\}\;, 
  \end{split} 
  \end{equation} 
where we used that the anti-bracket (\ref{anTIBracket}) is $\{S,V\}=\partial_iS V^i$ for a function $S$ and 
vector field $V$. Thus, $\delta=-\frac{1}{\hbar}\delta_{\rm BV}$ with $\delta_{\rm BV} = -\hbar\Delta+\{S, \,\cdot\,\}$ 
having   the more familiar normalization 
of the BV differential, which we  used in the main text.  
One may also verify that $\delta$ makes the following diagram  commute
\begin{equation}  
    \begin{tikzcd}
        \text{Vec}^k \arrow{r}{\delta} \arrow{d}{\iota} & \text{Vec}^{k-1} \arrow{d}{\iota}\\
        \Omega^{n-k} \arrow{r}{d} & \Omega^{n-k+1}
    \end{tikzcd}
\end{equation}
Therefore, the BV operator $\delta$ is transported by the isomorphism $\iota$ 
to the de Rham differential $d$. While the differentials are 
thus equivalent, the wedge products on forms and on polyvectors are \textit{not} equivalent. 
This follows  from the fact that $d$ acts on the wedge product of forms via the Leibniz rule, while 
$\delta$ does not act so on the wedge product of polyvectors. 

We now turn to the cohomology of the BV differential, which is defined as in (\ref{deRhamcoho}), 
but w.r.t.~$\delta$. We will show that the BV cohomology in degree zero computes 
the path integral and hence the quantum expectation value (\ref{Quantumexp}). 
We first note that for a vector field  $V=V^i\theta_i$ 
the expectation value of the function $\delta V$ gives a total divergence and hence vanishes: 
 \begin{equation}
  \expval{\delta V} = \int_{\mathbb{R}^n} \delta V\,\omega = \frac{1}{N}  \int_{\mathbb{R}^n}
  \partial_i\big(e^{-\frac{1}{\hbar}S} V^i\big) d^nx =  0\;, 
 \end{equation} 
where we used the first line of (\ref{deltaVcomp}). Thus, the expectation value is only well-defined 
up to $\delta$-exact terms and hence belongs to the cohomology. 
More precisely, 
in $\delta$ cohomology, a function $F$ (a polyvector of rank zero) is equivalent  to the number 
$ \expval{F}$ computed by the Gaussian path integral: 
 \begin{equation}\label{vevFincoho} 
  F = \expval{F}  +\delta V\;, 
 \end{equation} 
for a suitable vector field $V$. For instance, for the quadratic functions 
 \begin{equation}
  F(x) = c_{ij} x^i x^j\;, 
 \end{equation} 
with $c_{ij}$ an arbitrary symmetric matrix, we take the vector field 
 \begin{equation}
  V^i := -\hbar (A^{-1})^{ij} c_{jk} x^k \;, 
 \end{equation} 
where $A^{-1}$ is the inverse to the `kinetic matrix' defining the action $S(x) =\frac{1}{2} A_{ij} x^i x^j$, 
for which $\partial_iS=A_{ij}x^j$. 
One then computes with the second line of (\ref{deltaVcomp}) 
\begin{equation}
 \delta V = \partial_i V^i - \tfrac{1}{\hbar} \partial_iS V^i  = -\hbar c_{ij} (A^{-1})^{ij} + c_{ij} x^i x^j
 \;, 
\end{equation} 
and thus 
 \begin{equation}
 F = \hbar c_{ij} (A^{-1})^{ij} + \delta V \;. 
 \end{equation} 
This confirms  (\ref{vevFincoho}) with $\expval{F} = c_{ij}  \expval{x^i x^j} = \hbar c_{ij} (A^{-1})^{ij}$,  
which is usually summarized as $ \expval{x^i x^j}=\hbar (A^{-1})^{ij}$.  
It is left as an exercise for the reader to verify  that the expectation values of higher polynomials follow similarly, 
which are equivalent to those obtained by complete  Wick contractions.\footnote{For instance, 
for quartic functions $F(x)=c_{ijkl}x^i x^j x^k x^l$, where the  coefficients 
$c_{ijkl}$ are completely symmetric, one defines the vector field 
$V^i=-\hbar (A^{-1})^{ij}c_{jklm} x^k x^l x^m-3\hbar^2(A^{-1})^{ij} (A^{-1})^{kl} c_{jklm} x^m$ to show 
$F=3\hbar^2(A^{-1})^{ij} (A^{-1})^{kl} c_{ijkl} + \delta V$, encoding the expectation value $\expval{x^ix^jx^kx^l}$. }

 Finally, we can understand in this homological language why in general $\expval{F_1F_2}\neq\expval{F_1}\expval{F_2}$. 
 To this end assume $V_1$ and  $V_2$ 
 are vector fields so that 
  \begin{equation}
  \begin{split} 
    F_1 &= \expval{F_1}  +\delta V_1\;, \\
    F_2 &= \expval{F_2}  +\delta V_2\;. \\
  \end{split} 
  \end{equation} 
Note that $V_1$ and $V_2$  
are well-defined, up to $\delta$ exact terms, by the requirement 
that $ \expval{F_1}$ and  $\expval{F_2}$ 
are numbers.  From this we infer 
 \begin{equation}\label{F1F2cohomology} 
 \begin{split} 
  F_1F_2&=( \expval{F_1}  +\delta V_1) (\expval{F_2}  +\delta V_2) \\
  &=\expval{F_1} \expval{F_2} +\delta\big(\expval{F_1}V_2+\expval{F_2}V_1\big) + \delta V_1\delta V_2\;. 
 \end{split} 
 \end{equation} 
If $\delta$ was a first-order operator we could pull out a $\delta$ from the last term $\delta V_1\delta V_2$ 
and hence show that $F_1F_2$ is equal to  $\expval{F_1} \expval{F_2}$ up to $\delta$ exact terms, 
which in turn would imply $\expval{F_1F_2}= \expval{F_1}\expval{F_2}$. However, $\delta$ is of second order, 
with its failure to be first order given by the anti-bracket (\ref{anTIBracket}). (Here one uses that $\delta$ and $\Delta$ 
differ by a first-order operator which does act via the Leibniz rule and hence drops out off (\ref{anTIBracket}).) 
Thus, (\ref{F1F2cohomology}) implies only 
 \begin{equation}
  F_1F_2 = \expval{F_1} \expval{F_2} + \delta\big(\expval{F_1}V_2+\expval{F_2}V_1
  + V_1\delta V_2\big) - \{V_1, \delta V_2\} \;. 
 \end{equation} 
Therefore, the anti-bracket $ \{V_1, \delta V_2\}$, which is generally non-zero 
due to $\delta$ being second order, encodes the 
difference between $\expval{F_1F_2}$ and $\expval{F_1} \expval{F_2}$.

Summarizing this appendix, the computation of finite-dimensional path integrals (Gaussian integrals) 
can be formulated either in terms of the de Rham cohomology of the complex $\Omega^{\bullet}$ of 
differential forms or in terms of the cohomology 
of the BV operator $\delta$ on the `dual' complex ${\rm Vec}^{\bullet}$ of polyvector fields. 
Given a choice of top form there is an isomorphism $\iota: {\rm Vec}^{\bullet}\rightarrow  \Omega^{\bullet}$
that transports $\delta$ to $d$, which are thus equivalent. However, the natural wedge products on $\Omega^{\bullet}$
and ${\rm Vec}^{\bullet}$ are \textit{inequivalent} as on one the differential acts as a first-order operator while on 
the other it acts as a second-order operator. The wedge product on polyvectors is more natural as it properly encodes 
the algebra of functions (observables). The differential being second-order w.r.t.~wedge product then implies 
that in general $\expval{F_1F_2}\neq\expval{F_1}\expval{F_2}$ and similar relations.

\section{Feynman propagator and boundary conditions}
\label{App:Propagator}

In this appendix we summarize some facts about  the Feynman propagator and its boundary conditions.

For the harmonic oscillator, we study equations of the form
\begin{equation}\label{DrivenHamOsc}
\ddot g + \omega^2 g = f\,,
\end{equation}
with $f\in C^\infty_c (I)$ and $g \in C^\infty(I)$ on an interval $I$. We consider the cases 
\begin{equation}
    I \ = \  \mathbb{R}\,, \quad 
    (-\infty,t_i]\,, \quad 
    [t_i,\infty)\,, \quad 
    [t_i,t_f] \, .
\end{equation}
A solution to equation \eqref{DrivenHamOsc} can be obtained via the $(1+0)$-dimensional Feynman propagator:
\begin{equation}\label{BosonProp}
g(t) = h[f](t) = \int_I \text d s \, G_F(t,s) f(s) \, ,
\end{equation}
where
\begin{equation}\label{AlmostTextbookFeynmanP}
    G_F(t,s) = \frac{i}{2\omega} e^{-i\omega|t-s|} = \theta(t-s)\frac{i}{2\omega} e^{-i\omega(t-s)} + \theta(s-t)\frac{i}{2\omega} e^{i\omega(t-s)} \, ,
\end{equation}
where on the right hand side $G_F$ is written in the form we believe is the most well known by physicists. We will continue to work with expressions like \eqref{BosonProp}, where the Feynman propagator $G_F$ is integrated against a source term $f$.

In order to study the properties of \eqref{BosonProp}, we write it as
\begin{equation}\label{PropWithoutAV}
h[f](t) = \frac{i}{2\omega} e^{-i \omega t}\int_{a}^t \text d s \, e^{i\omega s} f(s) + \frac{i}{2\omega} e^{i \omega t}\int_{t}^{b} \text d s \, e^{-i\omega s} f(s) \, ,
\end{equation}
where $a = \{-\infty,t_i\}$ and $b \in \{\infty,t_f\}$ depending on $I$.

We first point out that $h$ indeed defines a linear map $h: C^\infty_c(I) \rightarrow C^\infty(I)$: 
\begin{prop}
 Let $f \in C^\infty_c(I)$. Then
\begin{equation}\label{hpm}
h_-[f](t) = \frac{i}{2\omega} e^{-i\omega t}\int_a^t d s \, e^{i\omega s} f(s) \, , \quad  h_+[f](t) = \frac{i}{2\omega} e^{i\omega t}\int_t^b d s \, e^{-i\omega s} f(s)
\end{equation}
are smooth functions for $a \in \{-\infty,t_i\}$ and $b \in \{\infty,t_f\}$.
\end{prop}
\begin{proof}
This follows immediately from 
products of smooth functions and integrals of smooth functions being smooth. 
\end{proof}

We next show that $g = h[f]$ indeed solves \eqref{DrivenHamOsc}.
\begin{prop}
Let $f \in C^\infty_c(I)$. Then $g_\pm(t) = h_\pm[f]$ with $h_\pm$ given in \eqref{hpm} satisfies
\begin{equation}
(\partial_t^2 + \omega^2) h_\pm[f](t) = \frac{1}{2}f(t) \mp \frac{i}{2 \omega} \dot f(t) \, .
\end{equation}
In particular, we have
\begin{equation}
(\partial_t^2 + \omega^2) (h_-[f](t) + h_+[f](t)) = f(t) \, .
\end{equation}
\end{prop}
\begin{proof}
We will give the proof for $h_+[f](t)$. We write
\begin{equation}\label{Justuseproductrule}
h_+[f](t) = u(t)v(t) \ \text{with} \quad u(t) = \frac{i}{2\omega}e^{i\omega t} \, ,\quad v(t) = \int_t^b \text d s \, e^{-i\omega s}f(s) \, .
\end{equation}
In this notation we have
\begin{equation}
(\partial_t^2 + \omega^2) h_+[f](t) = (\ddot u(t) + \omega^2 u(t))v(t) + 2 \dot u(t) \dot v(t) + u(t)\ddot v(t) \, .
\end{equation}
Since $u(t) \sim e^{i\omega t}$, we have $\ddot u(t) + \omega^2 u(t) = 0$. Two find the other two terms in \eqref{Justuseproductrule}, we compute the derivatives of $v(t)$. We find
\begin{equation}
\begin{split}
    \dot v(t) &= \partial_t \int_t^b \text d s \, e^{-i\omega s}f(s) = - e^{-i\omega t} f(t) \, , \\
    \ddot v(t) & = -\partial_t(- e^{-i\omega t} f(t))  = i\omega e^{-i\omega t} f(t) - e^{-i\omega t} \dot f(t) \, .
\end{split}
\end{equation}
From this we deduce that
\begin{equation}
(\partial_t^2 + \omega^2) h_+[f](t) =  2 \dot u(t) \dot v(t) + u(t)\ddot v(t) = f(t) - \frac{1}{2} f(t) - \frac{i}{2\omega} \dot f(t) = \frac{1}{2} f(t) - \frac{i}{2\omega} \dot f(t) \, ,
\end{equation}
which is what we wanted to show.
\end{proof}
Note that the proofs of the above two propositions only use very elementary calculus and could be easily proven by most readers. Despite that, we still want to give these proofs as they show that the properties of the Feynman propagator can be deduced with the help of basic mathematics. The physics textbook treatment of the Feynman propagator obscures this, since there the Feynman propagator is given in the form
\begin{equation}\label{TextbookFeynmanP}
    G_F(t,s) = -\int \frac{d k_0}{2 \pi} \frac{e^{-ik_0(t-s)}}{k_0^2 - \omega^2 + i \varepsilon} \, .
\end{equation}
Only after applying the residue theorem one arrives at the form given in \eqref{AlmostTextbookFeynmanP}. Since the proofs of the properties of $h[f]$ are so elementary, we prefer to work with $h[f]$ in the form given in \eqref{BosonProp} and \eqref{PropWithoutAV} instead of \eqref{TextbookFeynmanP}.

Coming back to our analysis of $h$, we showed that $(\partial_t^2 + \omega^2) h[f](t) = f(t)$. One could say that $h$ is a right-inverse of $\partial_t^2 + \omega^2$, although for that statement to be true one needs to be very careful with the domains of $\partial_t^2 + \omega^2$ and $h$. We ignore this issue for now and will ask whether $h$ is also a left-inverse of $\partial_t^2 + \omega^2$.
\begin{prop}
Let $g \in C^\infty_c(I)$. We then have
\begin{equation}\label{fullhboundary}
    h[\ddot g + \omega^2 g](t) = g(t) - \frac{i}{2\omega}e^{-i \omega(t-a)}(\dot g(a) - i\omega g(a)) + \frac{i}{2 \omega}e^{i \omega(t-b)}  (\dot g(b) +  i \omega g(b)) \, ,
\end{equation}
where $g(a) = \dot g(a) = 0$ when $a = - \infty$ and $g(b) = \dot g(b) = 0$ when $b = \infty$.
\end{prop}
\begin{proof}
Since $(\partial_s^2 + \omega^2) e^{i \omega s} = 0$, one finds that $h_\pm[\ddot g + \omega^2 g]$ only consists of boundary terms. Explicitly,
\begin{equation}\label{hboundaryterms}
\begin{split}
h_-[\ddot g + \omega^2 g] &= \frac{i}{2 \omega} e^{-i\omega t} \left[e^{i \omega s} \overleftrightarrow {\partial_s} g(s) \right]_{s = t} - \frac{i}{2 \omega} e^{-i\omega t} \left[e^{i \omega s} \overleftrightarrow {\partial_s} g(s) \right]_{s = a} \, , \\
h_+ [\ddot g + \omega^2 g] &= \frac{i}{2 \omega} e^{i\omega t}\left[e^{-i \omega s} \overleftrightarrow {\partial_s} g(s) \right]_{s = b} - \frac{i}{2 \omega} e^{i\omega t}\left[e^{-i \omega s} \overleftrightarrow {\partial_s} g(s) \right]_{s = t} \, .
\end{split}
\end{equation}
where $f(s) \overleftrightarrow {\partial_s} g(s) = f(s) \dot g(s) - \dot f(s) g(s)$. We find
\begin{equation}
e^{\pm i \omega s} \overleftrightarrow {\partial_s} g(s) = e^{\pm i \omega s}(\dot g(s) \mp i \omega g(s)) \, .
\end{equation}
Using this in \eqref{hboundaryterms}, we find
\begin{equation}
\begin{split}
h_-[\ddot g + \omega^2 g] &= \frac{i}{2 \omega}(\dot g(t) - i\omega g(t)) - \frac{i}{2 \omega} e^{-i\omega (t-a)}(\dot g(a) - i\omega g(a)) \, , \\
h_+ [\ddot g + \omega^2 g] &= \frac{i}{2 \omega}e^{i\omega(t-b)}( \dot g(b) + i \omega g(b)) - \frac{i}{2 \omega}(\dot g(t) + i \omega g(t))\, .
\end{split}
\end{equation}
Adding these two expression we obtain \eqref{fullhboundary}.
\end{proof}

The above result shows that $h[\ddot g + \omega^2 g] = g$ only when $I = \mathbb{R}$. When $I = [a,b]$, we find that this relation is obstructed due to the boundary terms
\begin{equation}
\frac{i}{2 \omega}e^{i\omega(t-b)}\partial_+ g(b) - \frac{i}{2 \omega}e^{-i\omega(t-a)}\partial_- g(a) \, ,
\end{equation}
where we defined $\partial_\pm = \partial_t \pm i \omega$. Therefore, we can define the spaces
\begin{equation}\label{smoothwithboundary}
    C_{cp}^\infty(I) := \{f \in C^\infty_c(I) \, | \, \partial_- f(a) = \partial_+ f(b) = 0\} \, .
\end{equation} 
We conclude that in general $h[\ddot g + \omega^2 g] = g$ for $g \in C^\infty_{cp}(I)$.
\vspace{1ex}

\noindent \emph{Boundary conditions of the Feynman propagator}
\vspace{1ex}

In this part, we will discuss the boundary conditions satisfied by $h[f]$ for various $f$. For begin with the case $I = \mathbb{R}$. In the following we write $t > \text{supp}(f)$ if $t > s$ for all $s \in \text{supp}(f)$. Along the same lines, we write $t < \text{supp}(f)$ if $t < s$ for all $s \in \text{supp}(f)$.

For a given $f \in C^\infty_c(\mathbb{R})$, we consider $h$ as given in \eqref{PropWithoutAV}. For $t > \text{supp}(f)$, $h[f]$ becomes
\begin{equation}\label{hlate}
h[f](t) = \frac{i}{2\omega} e^{-i \omega t}\int_{-\infty}^\infty \text d s \, e^{i\omega s} f(s) \, ,
\end{equation}
while for $t < \text{supp}(f)$, it becomes
\begin{equation}\label{hearly}
h[f](t) = \frac{i}{2\omega} e^{i \omega t}\int_{-\infty}^{\infty} \text d s \, e^{-i\omega s} f(s) \, .
\end{equation} 
Note that both expressions are solutions to $\partial_t^2 h[f](t) + \omega^2 h[f](t) = 0$. This is consistent with the fact that the right hand side of \eqref{BosonProp} vanishes outside of $\text{supp}(f)$. However, we can even make a stronger statement. Defining $\partial_\pm = \partial_t \pm i \omega$ as above, we can say that for $t < \text{supp}(f)$ we have $\partial_- h[f](t) = 0$ and for $t > \text{supp}(f)$ we have $\partial_+ h[f](t) = 0$. This then automatically implies that $(\partial_t^2 + \omega^2) h[f] = \partial_+ \partial_- h[f] = 0$.

Another interesting case happens when $f$ is orthogonal to solutions. If
\begin{equation}
\int \text d s \, f(s) e^{-i \omega s} = 0 \, , 
\end{equation}
we find from \eqref{hearly} that $h[f](t) = 0$ for $t < \text{supp}(f)$. Similarly from \eqref{hlate} we deduce that orthogonality to $e^{i\omega s}$ implies that $h[f](t) = 0$ for $t > \text{supp}(f)$. In summary, we have that
\begin{itemize}
\item $\partial_- h[f](t) = 0$ for $t < \text{supp}(f)$.
\item $\partial_+ h[f](t) = 0$ for $t < \text{supp}(f)$ and $f \perp e^{-i \omega t}$.
\item $\partial_+ h[f](t) = 0$ for $t > \text{supp}(f)$.
\item $\partial_- h[f](t) = 0$ for $t > \text{supp}(f)$ and $f \perp e^{i\omega t}$.
\end{itemize}
Note that $\partial_+ f(t) = \partial_- f(t) = 0$ implies that $f(t) = 0$.

The function $h[f]$ has a surprisingly similar behavior when we work on the closed interval $[t_i,t_f]$. Given $f \in C^\infty_c([t_i,t_f]) = C^\infty([t_i,t_f])$, we write \eqref{BosonProp} as
\begin{equation}
h[f](t) = \frac{i}{2\omega} \int_{t_i}^t \text d s \, e^{-i\omega(t-s)} f(s) + \frac{i}{2\omega} \int_{t}^{t_f} \text d s \, e^{i\omega(t-s)} f(s) \, .
\end{equation}
We act with $\partial_-$, in which case we find
\begin{equation}
\partial_- h[f](t)  = -\int_{t_i}^t e^{-i\omega(t-s)} f(s) \, ,
\end{equation}
which vanishes at $t = a$. If $f \perp e^{i\omega s}$, it also vanishes at $t = b$. Likewise, let us compute $\partial_+ h[f](t)$. We find
\begin{equation}
\partial_+ h[f](t) = - i\int_t^{t_f} e^{i\omega(t-s)}f(s) \, .
\end{equation}
We have $\partial_- h[f](t_f) = 0$. Further, if $f \perp e^{-i\omega t}$, we also have $\partial_+ h[f](t_i) = 0$. Altogether, we deduce that
\begin{itemize}
\item $\partial_- h[f](t_i) = 0$.
\item $\partial_+ h[f](t_i) = 0$ for $f \perp e^{-i \omega t}$.
\item $\partial_+ h[f](t_f) = 0$.
\item $\partial_- h[f](t_f) = 0$ for $f \perp e^{i\omega t}$.
\end{itemize}
This is a very similar behavior as before, although we consider $t = t_i$ instead of $t < \text{supp}(f)$ and $t = t_f$ instead of $t > \text{supp}(f)$.

The above observations translate to the case of compactly supported functions $f$ on half-open intervals $[t_i,\infty)$ and $(-\infty,t_f]$.

\section{Reconstructing prefactorization algebras from connected open sets}
\label{App:C}

In section 3, we gave the definition of prefactorization algebras only in terms of connected open sets. 
In this appendix we discuss the relation to the general case.
Indeed, the standard definition of a prefactorization algebra on a topological space $X$ is defined with respect to any open sets, not only the connected ones. If we assume that the topological space $X$ is locally connected, the value of $\mathfrak{F}$ on an arbitrary open set is then determined once one requires multiplicativity, in the following sense. For $U_1,U_2$ disjoint, we require that
\begin{equation}\label{Multiplicativity}
    m_{U_1,U_2}^{U_1 \sqcup U_2}: \mathfrak{F}(U_1) \otimes \mathfrak{F}(U_2) \longrightarrow \mathfrak{F}(U_1 \sqcup U_2)
\end{equation}
is an isomorphism. Without loss of generality, we can just define $\mathfrak{F}(U_1 \sqcup U_2) := \mathfrak{F}(U_1) \otimes \mathfrak{F}(U_2)$. Note that \eqref{Multiplicativity} implies that $\mathfrak{F}(\emptyset) = \mathbb{C}$. The inclusion $m_{\emptyset}^U =: \eta_U$ then induces a map $\mathbb{C} \rightarrow \mathfrak{F}(U)$ for all open sets $U$.

Suppose now that we defined $\mathfrak{F}$ on connected open sets of $X$. If an open set $V \subseteq X$ is the disjoint union of \emph{finitely many} connected and disjoint open sets $U_1,\ldots U_n$, we define
\begin{equation}
    \mathfrak{F}(V) := \mathfrak{F}(U_1) \otimes \cdots \otimes \mathfrak{F}(U_n) \, .
\end{equation}
Given $U_{1}, \ldots, U_{n}$ disjoint, we also need to define the inclusion
$m_{U_1 \sqcup \ldots \sqcup U_{k}}^{U_1 \sqcup \ldots \sqcup U_n}$ for $k \le n$. As a special case, we consider $m_{U_1}^{U_1 \sqcup U_2}$. Its value is determined by
\begin{equation}
   m_{U_1,U_2}^{U_1 \sqcup U_2} \circ (\text{id}_{\mathfrak{F}(U_1)} \otimes m_{\emptyset}^{U_2}) =  m_{U_1}^{U_1 \sqcup U_2} \circ m_{U_1,\emptyset}^{U_1} \, .
\end{equation}
Since $m_{U_1,\emptyset}^{U_1}$ and $m_{U_1,U_2}^{U_1 \sqcup U_2}$ are the identity (up to the canonical isomorphism $\mathfrak{F}(U_1) \otimes \mathbb{C} \cong \mathfrak{F}(U_1)$), we deduce that $m_{U_1}^{U_1 \sqcup U_2} = \text{id}_{\mathfrak{F}(U_1)} \otimes \eta_{U_2}$. More generally, associativity can be used to show that 
\begin{equation}
    m_{U_1\sqcup \ldots \sqcup U_k}^{U_1\sqcup \ldots \sqcup U_n} = \text{id}_{\mathfrak{F}(U_1)} \otimes \cdots \text{id}_{\mathfrak{F}(U_k)} \otimes \eta_{U_{k+1}} \otimes \cdots \otimes \eta_{U_n} \, . 
\end{equation}

A generic open set $V \subseteq X$ can be written as a (potentially uncountable) union of its connected components $\{U_{i}\}_{i \in I}$, i.e. $\bigsqcup_{i \in I} U_i = V$, which are all necessarily open by the assumption that $X$ is locally connected. One then defines
\begin{equation}
    \mathfrak{F}(V) = \bigoplus_{\substack{J \subseteq I \\ J finite}} \bigotimes_{i \in J} \mathfrak{F}(U_i) /\sim \, ,
\end{equation}
where the equivalence relation identifies
\begin{equation}
    \mathfrak{F}(U_{i_1}) \otimes \cdots\otimes \mathfrak{F}(U_{i_k}) \longrightarrow \mathfrak{F}(U_{i_1}) \otimes \cdots \otimes \mathfrak{F}(U_{i_n})
\end{equation}
as a subspace under the map $m_{U_{i_1} \sqcup \ldots \sqcup U_{i_k}}^{U_{i_1} \sqcup \ldots \sqcup U_{i_n}}$ derived above.


\bibliography{Bibliography}

\providecommand{\href}[2]{#2}\begingroup\raggedright\begin{thebibliography}{10}

\bibitem{Costello_Gwilliam_2016}
K.~Costello and O.~Gwilliam, {\em Factorization Algebras in Quantum Field Theory,Volume 1}.
\newblock New Mathematical Monographs. Cambridge University Press, 2016.

\bibitem{Costello_Gwilliam_2021}
K.~Costello and O.~Gwilliam, {\em Factorization Algebras in Quantum Field Theory, Volume 2}.
\newblock New Mathematical Monographs. Cambridge University Press, 2021.

\bibitem{Chiaffrino:2021pob}
C.~Chiaffrino, O.~Hohm, and A.~F. Pinto, ``{Homological quantum mechanics},'' \href{https://dx.doi.org/10.1007/JHEP02(2024)137}{{\em JHEP} {\bfseries 02} (2024) 137}, \href{https://arxiv.org/abs/2112.11495}{{\ttfamily arXiv:2112.11495 [hep-th]}}.

\bibitem{BWilliam}
B.~Williams, ``The virasoro vertex algebra and factorization algebras on riemann surfaces,'' \href{https://dx.doi.org/10.1007/s11005-017-0982-7}{{\em Letters in Mathematical Physics} {\bfseries 107} no.~12, (2017) 2189--2237}. \url{https://doi.org/10.1007/s11005-017-0982-7}.

\bibitem{GwilliamRejzner}
O.~Gwilliam and K.~Rejzner, ``Relating nets and factorization algebras of observables: Free field theories,'' \href{https://dx.doi.org/10.1007/s00220-019-03652-9}{{\em Communications in Mathematical Physics} {\bfseries 373} no.~1, (2020) 107--174}. \url{https://doi.org/10.1007/s00220-019-03652-9}.

\bibitem{Gwilliam:2022vja}
O.~Gwilliam and K.~Rejzner, ``{The observables of a perturbative algebraic quantum field theory form a factorization algebra}.'' 12, 2022.

\bibitem{Costello:2022wso}
K.~Costello and N.~M. Paquette, ``{Celestial holography meets twisted holography: 4d amplitudes from chiral correlators},'' \href{https://dx.doi.org/10.1007/JHEP10(2022)193}{{\em JHEP} {\bfseries 10} (2022) 193}, \href{https://arxiv.org/abs/2201.02595}{{\ttfamily arXiv:2201.02595 [hep-th]}}.

\bibitem{nishinaka2024}
Y.~Nishinaka, ``An algebraic construction of functors between vertex algebras and costello-gwilliam factorization algebras.'' 2024.
\newblock \url{https://arxiv.org/abs/2408.00412}.

\bibitem{gwilliam2012factorization}
O.~Gwilliam, {\em Factorization algebras and free field theories}.
\newblock PhD thesis, Northwestern University, 2012.

\bibitem{beilinson2004chiral}
A.~Beilinson and V.~Drinfeld, {\em Chiral Algebras}.
\newblock American Mathematical Society Colloquium publications. American Mathematical Society, 2004.

\bibitem{Chiaffrino:2020akd}
C.~Chiaffrino, O.~Hohm, and A.~F. Pinto, ``{Gauge Invariant Perturbation Theory via Homotopy Transfer},'' \href{https://dx.doi.org/10.1007/JHEP05(2021)236}{{\em JHEP} {\bfseries 05} (2021) 236}, \href{https://arxiv.org/abs/2012.12249}{{\ttfamily arXiv:2012.12249 [hep-th]}}.

\bibitem{weibel1994introduction}
C.~A. Weibel, {\em An Introduction to Homological Algebra}.
\newblock Cambridge Studies in Advanced Mathematics. Cambridge University Press, 1994.

\bibitem{batalin1981gauge}
I.~A. Batalin and G.~Vilkovisky, ``Gauge algebra and quantization,'' {\em Physics Letters B} {\bfseries 102} no.~1, (1981) 27--31.

\bibitem{batalin1983quantization}
I.~A. Batalin and G.~Vilkovisky, ``Quantization of gauge theories with linearly dependent generators,'' {\em Physical Review D} {\bfseries 28} no.~10, (1983) 2567.

\bibitem{Markl:1997bj}
M.~Markl, ``{Loop homotopy algebras in closed string field theory},'' \href{https://dx.doi.org/10.1007/PL00005575}{{\em Commun. Math. Phys.} {\bfseries 221} (2001) 367--384}, \href{https://arxiv.org/abs/hep-th/9711045}{{\ttfamily arXiv:hep-th/9711045}}.

\bibitem{crainic2004perturbationlemmadeformations}
M.~Crainic, ``On the perturbation lemma, and deformations.'' 2004.
\newblock \url{https://arxiv.org/abs/math/0403266}.

\bibitem{costelloFactorization}
K.~Costello and O.~Gwilliam, ``Factorization algebra.'' 2023.
\newblock \url{https://arxiv.org/abs/2310.06137}.

\bibitem{HomotopyPA}
N.~Idrissi and E.~Rabinovich, ``Homotopy prefactorization algebras,'' \href{https://dx.doi.org/10.1007/s40687-024-00456-9}{{\em Research in the Mathematical Sciences} {\bfseries 11} no.~3, (2024) 45}. \url{https://doi.org/10.1007/s40687-024-00456-9}.

\bibitem{Alfonsi:2023qpv}
L.~Alfonsi and C.~A.~S. Young, ``{Towards non-perturbative BV-theory via derived differential geometry}.'' 7, 2023.

\bibitem{Cattaneo:2012qu}
A.~S. Cattaneo, P.~Mnev, and N.~Reshetikhin, ``{Classical BV theories on manifolds with boundary},'' \href{https://dx.doi.org/10.1007/s00220-014-2145-3}{{\em Commun. Math. Phys.} {\bfseries 332} (2014) 535--603}, \href{https://arxiv.org/abs/1201.0290}{{\ttfamily arXiv:1201.0290 [math-ph]}}.

\bibitem{Cattaneo:2015vsa}
A.~S. Cattaneo, P.~Mnev, and N.~Reshetikhin, ``{Perturbative quantum gauge theories on manifolds with boundary},'' \href{https://dx.doi.org/10.1007/s00220-017-3031-6}{{\em Commun. Math. Phys.} {\bfseries 357} no.~2, (2018) 631--730}, \href{https://arxiv.org/abs/1507.01221}{{\ttfamily arXiv:1507.01221 [math-ph]}}.

\bibitem{Witten:1990wb}
E.~Witten, ``{A Note on the Antibracket Formalism},'' \href{https://dx.doi.org/10.1142/S0217732390000561}{{\em Mod. Phys. Lett. A} {\bfseries 5} (1990) 487}.

\bibitem{gwilliamFeynman}
O.~Gwilliam and T.~Johnson-Freyd, ``How to derive feynman diagrams for finite-dimensional integrals directly from the bv formalism,'' {\em Topology and quantum theory in interaction} {\bfseries 718} (2012) 175--185.

\end{thebibliography}\endgroup
\bibliographystyle{utphys}

\end{document}